%% file: article.tex
\documentclass[a4paper,12pt]{article}

\pdfoutput=1

\usepackage[T1]{fontenc}
\usepackage[utf8]{inputenc}

\usepackage{bm}
\usepackage{dsfont}

\usepackage[english]{babel}

\usepackage{graphicx}
\usepackage{subcaption}
\usepackage{listings}

\usepackage[version=4]{mhchem}  
\usepackage{amsthm, amssymb}
\usepackage{siunitx}

\usepackage[x11names]{xcolor}
\usepackage[skins, breakable]{tcolorbox}

\usepackage[square, numbers, sort&compress]{natbib}

\usepackage{catchfile}  
\usepackage[unicode]{hyperref}  
\usepackage{doi}  
\hypersetup{
	linktocpage=true,
	bookmarksopen=true,
	bookmarksnumbered=true,
}

\begin{document}


\newcommand{\citex}[2]{\cite[\expandafter{#1}]{#2}}  
\newcommand{\dlmf}[3][]{\href{http://dlmf.nist.gov/#2#1#3}{#2#3}}
\newcommand{\wikipedia}[2][en]{\href{http://#1.wikipedia.org/wiki/#2}{#2}}
\newcommand{\elektrodo}{https://gitlab.com/cfgy/elektrodo/tree/p2019mar}

\newcommand{\fourier}{\operatorname{\mathcal{F}}_{x}}
\newcommand{\laplace}{\operatorname{\mathcal{L}}_{t}}
\newcommand{\ud}[1]{\,\mathrm{d}#1}

\newcommand{\e}{\operatorname{e}}
\newcommand{\sech}{\operatorname{sech}}
\newcommand{\arctanh}{\operatorname{arctanh}}

\newcommand{\arcsn}{\operatorname{arcsn}}
\newcommand{\sn}{\operatorname{sn}}
\newcommand{\cn}{\operatorname{cn}}
\newcommand{\dn}{\operatorname{dn}}
\newcommand{\sd}{\operatorname{sd}}
\newcommand{\cd}{\operatorname{cd}}
\newcommand{\nd}{\operatorname{nd}}
\newcommand{\pq}{\operatorname{pq}}
\newcommand{\pr}{\operatorname{pr}}
\newcommand{\qr}{\operatorname{qr}}
\newcommand{\pp}{\operatorname{pp}}

\newcommand{\deriv}[2]{\frac{\ud #1}{\ud #2}}
\newcommand{\parderiv}[2]{\frac{\partial #1}{\partial #2}}

\newcommand{\simu}{_{\mathrm{sim}}}
\newcommand{\unit}{\text{unit}}
\newcommand{\whole}{\text{whole}}
\newcommand{\IDAE}{\text{IDAE}}
\newcommand{\nul}{\text{null}}

\newtheorem{observacion}{Remark}[section]
\newtheorem{definicion}{Definition}[section]
\newtheorem{lema}{Lemma}[section]
\newtheorem{teorema}{Theorem}[section]
\newtheorem{corolario}{Corollary}[section]

\tcolorboxenvironment{observacion}{enhanced jigsaw, breakable, before skip=1em, after skip=1em, colframe=gray}
\tcolorboxenvironment{definicion}{enhanced jigsaw, breakable, before skip=1em, after skip=1em, colframe=gray}
\tcolorboxenvironment{lema}{enhanced jigsaw, breakable, before skip=1em, after skip=1em, colframe=gray}
\tcolorboxenvironment{teorema}{enhanced jigsaw, breakable, before skip=1em, after skip=1em, colframe=gray}
\tcolorboxenvironment{corolario}{enhanced jigsaw, breakable, before skip=1em, after skip=1em, colframe=gray}
\tcolorboxenvironment{proof}{blanker, breakable, left=5mm, before skip=1em, after skip=1em, borderline west={1mm}{0pt}{gray}, parbox=false}
\renewcommand{\qedsymbol}{\textit{QED}.}


\newcommand{\givennames}[1]{#1}
\newcommand{\familynames}[1]{\textsc{#1}}

\newcommand{\documenttitle}{
	Steady-state theory of interdigitated array of electrodes in confined spaces:
	Case of pure diffusion and reversible electrode reactions
}

\title{\documenttitle}

\author{
	\givennames{Cristian F.} \familynames{Guajardo Yévenes} \\
	Biological Engineering Program and \\
	Pilot Plant Development and Training Institute \\
	King Mongkut's University of Technology Thonburi, Thailand \\
	\texttt{cristian.gua@kmutt.ac.th}
	\and
	\givennames{Werasak} \familynames{Surareungchai} \\
	School of Bioresources and Technology and \\
	Nanoscience \& Nanotechnology Graduate Program \\
	King Mongkut's University of Technology Thonburi, Thailand \\
	\texttt{werasak.sur@kmutt.ac.th}
}

\date{March 6, 2019}


\newcommand{\sep}{, }
\newcommand{\resumen}{\input{text-abstract.tex}}
\newcommand{\pclaves}{\input{text-keywords.tex}}
\hypersetup{
	pdftitle={\documenttitle},
	pdfauthor={Cristian F. GUAJARDO YÉVENES and Werasak SURAREUNGCHAI},
	pdfsubject={\input{text-abstract.tex}},
	pdfkeywords={\input{text-keywords.tex}}
}


\numberwithin{equation}{section}

\maketitle

\section*{Abstract}


\input{text-abstract}
\par \vspace{\baselineskip} \noindent
\emph{Keywords:} \input{text-keywords}

\input{text-grabs}

\tableofcontents


\input{text-introduction}
\input{text-theory}
\input{text-results_discussion}
\input{text-conclusion}
\input{text-acknowledgements}

\bibliographystyle{unsrtnat}
\bibliography{references}


\clearpage

\setcounter{page}{1}
\renewcommand*{\thepage}{S\arabic{page}}  

\setcounter{section}{0}
\renewcommand*{\thesection}{S\arabic{section}}  
\renewcommand*{\theHsection}{\thesection}  

\setcounter{figure}{0}
\renewcommand*{\thefigure}{S\arabic{figure}}  
\renewcommand*{\theHfigure}{\thefigure}  

\setcounter{table}{0}
\renewcommand*{\thetable}{S\arabic{table}}  
\renewcommand*{\theHtable}{\thetable}  

\begin{center}
	\LARGE
	Supplementary information \\[0.5\baselineskip]
	\Large
	\documenttitle \\[0.5\baselineskip]
	\normalsize
	\givennames{Cristian F.} \familynames{Guajardo Yévenes} and
	\givennames{Werasak} \familynames{Surareungchai} \\[0.5\baselineskip]
	\footnotesize
	\emph{King Mongkut's University of Technology Thonburi, 49 Soi Thianthale 25, Thanon Bangkhunthian Chaithale, Bangkok 10150, Thailand}
	\normalsize\vspace{3em}
\end{center}

\input{text-additional_proofs}
\input{text-numerical_calculations}

\end{document}

%% file: text-abstract.tex
Analytical equations were found for interdigitated electrodes,
which considered reversible electrode reactions and pure diffusion within confined spaces.
A conformal transformation, obtained by the use of Jacobian elliptic functions,
was applied to solve the diffusion equation in steady state.
The obtained steady-state current depends on the ratio of elliptic integrals of the first kind,
in which their moduli are functions of the relative dimensions of the cell.
The current is smaller for shallower cells, but approaches similiar values to
those of semi-infinite geometries when the cell is sufficiently tall.
Approximations using trigonometric and hyperbolic expressions were also found
for the steady-state current in the cases of shallow and tall cells respectively.

%% file: text-keywords.tex
confined cell\sep
interdigitated array of electrodes\sep
diffusion equation\sep
steady state\sep
voltammogram\sep
current density\sep
limiting current

%% file: text-grabs.tex

\section*{Graphical abstract}

\begin{figure}[h]
	\centering
	\includegraphics{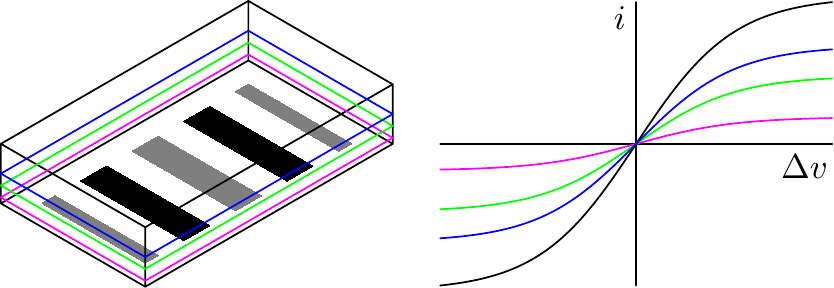}
\end{figure}

%% file: text-introduction.tex

\clearpage
\section{Introduction}

The use of microelectrodes has been in constant increase since
approximately the beginning of 1980 \cite{Dayton1980, Ewing1981, Wightman1981},
mainly due to the development of microfabrication techniques
and because of the benefits of microelectrodes towards sensing:
reduced ohmic potential drop, fast non-faradaic time constants,
shorter diffusion times, fast establishment of steady-state signals
and increased signal-to-noise ratio \cite{Forster2007, Szunerits2007}.

There are several kinds of microelectrode configurations:
disk, cylinder, disk array, microband array\footnote{%
	composed of only anodes or only cathodes (all band electrodes at the same potential).
}, interdigitated array\footnote{%
	composed of band electrodes organized as alternating anodes and cathodes.
}, ring and recessed electrodes \cite{Aoki1993sep, Forster2007}.
Among all these configurations, interdigitated array of electrodes (IDAE)
is a popular choice and has drawn great attention,
because it can produce high currents from the redox cycling in between
closely arranged generators and collectors \cite{Aoki1993sep, Szunerits2007},
apart from all the known advantages inherited from microelectrodes.

The use of simple mathematical models allows
understanding, prediction of behavior and design of electrode configurations.
In this way, the behavior of an IDAE can be prescribed,
when its physical implementation is fabricated according to
the constraints and assumptions imposed by its mathematical model.
Among these models, pure diffusional transport and
Nernstian boundary conditions are probably the most used for IDAEs,
since they allow great simplifications.

Using this model, simulations have been performed to understand
the time dependence of the current generated at IDAE \cite{Aoki1989jul}.
Also, theoretical results have been obtained in
\cite{Aoki1988dec, Aoki1990apr, Morf2006may}
by analytically solving the diffusion equation in steady state.
From these theoretical results, the work of Aoki stands out,
due to the obtention of exact expressions for the current-potential curve and
limiting current in steady state for reversible \cite{Aoki1988dec} and
irreversible \cite{Aoki1990apr} electrode reactions.

These theoretical results consider that the IDAE is subject to semi-infinite geometry,
that is, they consider that the IDAE is in contact with a large amount of solution around it,
so that the diffusion layer is much smaller than the total size of the electrochemical cell.
In practical terms, this is equivalent to say that
the ratio between the ``height of the cell'' and
the ``center-to-center separation between adjacent electrodes'' is very large.

This semi-infinite condition is in general not valid for every IDAE electrochemical cell.
This can be seen particularly in the case of microfluidic devices,
where the height of the microchannel is clearly finite,
especially when using low-cost fabrication techniques:
softlithography, transparency-film masks, paper and screen-print to name a few
\cite{Duffy1998dec, Xia1998aug, Whitesides2001aug, Dungchai2009jul}.
Typical heights of microfluidic channels fabricated using softlithography
depend on the thickness of the photoresist molds,
which can range between \SIrange{1}{200}{\micro\metre} \cite{Duffy1998dec}.
The width and gap of IDAE bands fabricated using photolithography
is commonly constrained by the resolution of transparency-film masks,
which can range between \SIrange{20}{50}{\micro\metre}
for printers operating between \SIrange{3380}{5080}{dpi} \cite{Duffy1998dec, Whitesides2001aug}.
Therefore, these fabrication techniques can produce microfluidic electrochemical cells
where the ratio between ``height of the cell'' and ``center-to-center separation between adjacent electrodes''
is clearly finite and ranges between $\sim$ \numrange{0,01}{10}.
This implies that the equations obtained in \cite{Aoki1988dec, Aoki1990apr}
may not be always applicable to IDAEs operating within microfluidic channels.

Currently, and due to the lack of analytical expressions for current and current-potential curves,
experiments involving IDAE in microchannels are normally constrasted against simulations of its ideal behavior,
namely, the diffusion equation subject to Nernstian boundary conditions \cite{Goluch2009may,Kanno2014}.
Despite this fact, these studies, together with pure experimental \cite{Lewis2010mar} and theoretical \cite{Strutwolf2005feb, GuajardoYevenes2013sep} results,
have contributed to reveal the behavior of IDAE in confined spaces with stagnant solutions.
These results indicate that higher currents are obtained
when using electrochemical cells with taller microchannels.
In fact, the current approaches similar values to the case of semi-infinite cells,
predicted by \cite{Aoki1988dec, Aoki1990apr},
when the ``height of the microchannel'' $H$ is larger than
the ``center-to-center separation between adjancent electrodes'' $W$.
According to \cite{GuajardoYevenes2013sep},
finite-height microfluidic cells can be regarded as semi-infinite cells,
within 12\% error in bulk concentration, when $H/W \gtrsim 1$.


Therefore, this work intends to elucidate the ideal behavior of IDAE
in electrochemical cells within confined spaces,
and replace the use of simulations,
by obtaining analytical expressions for the current and
voltammogram as a function of the dimensions of the IDAE
and the geometry of the electrochemical cell.
This model would enable the design of electrochemical cells
that output the maximum current available given microfabrication constraints,
and contrast their actual performance with their ideal behavior.

%% file: text-theory.tex

\section{Theory}

\subsection{Definition of the problem}
\label{main:def:problema}

\begin{figure*}[t]
	\centering
	\subcaptionbox{
		\label{main:fig:celda:int_C}}{\includegraphics{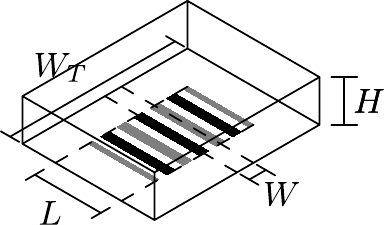}}
	\subcaptionbox{
		\label{main:fig:celda:ext_C}}{\includegraphics{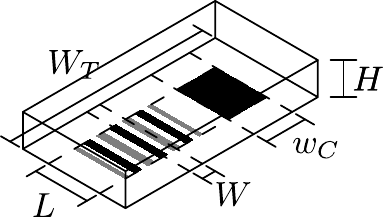}}
	\subcaptionbox{
		\label{main:fig:celda:unitaria}}{\includegraphics{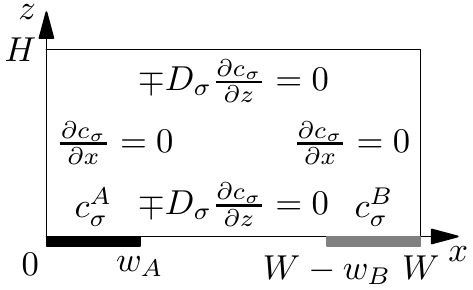}}
	\caption{
		Sketch of an interdigitated array of electrodes (IDAE)
		in a cell of finite height $H$ and total width $W_{T}$.
		(\subref{main:fig:celda:int_C}) Case of one array used as counter electrode.
		(\subref{main:fig:celda:ext_C}) Case of external counter electrode coplanar to the IDAE.
		(\subref{main:fig:celda:unitaria}) 2D unit cell of height $H$,
		width $W$ and half electrode bands $A$ and $B$.
		Fig. (\subref{main:fig:celda:int_C}) can be modeled exactly as
		an assembly of unit cells if the IDAE fits exactly in the cell.
		Fig. (\subref{main:fig:celda:int_C}) and (\subref{main:fig:celda:ext_C})
		can be modeled approximately as an assembly of unit cells
		provided that the number of electrode bands and the length $L$ are large enough.
	}
	\label{main:fig:celda}
\end{figure*}

\subsubsection{Description of the cell}
\label{main:sec:description}


Consider an electrochemical cell with an IDAE configuration, as shown in Fig. \ref{main:fig:celda}.
The cell has a finite height $H$, total width $W_{T}$
and it is surrounded by walls which behave as perfect insulators.


The IDAE is located at the floor of the cell
and it is composed of two arrays of bands, $A$ (black) and $B$ (gray),
of $N_{A}$ and $N_{B}$ bands respectively.
Each band $A$ and $B$ has respectively a width of $2w_{A}$ and $2w_{B}$,
a common length $L$, and they are placed alternatingly
such that two consecutive bands have a center-to-center separation of $W$.
Besides the IDAE, the cell may include a counter electrode $C$
of width $w_{C}$ and length $L$,
which is assumed to lay on the same plane as the IDAE.


Inside the electrochemical cell
there is an oxidated species $O$ and a reduced species $R$,
with diffusion coefficients $D_{O}$ and $D_{R}$ respectively.
Both species are transported solely by diffusion
and react at the surface of the electrodes according to
\begin{equation}
	\ce{$O$ + $n_{e}$ e- <=> $R$}
\end{equation}
where $n_{e}$ corresponds to the number of exchanged electrons.
Here it is assumed that the charge transfer on the electrodes
follows reversible electrode reactions,
and therefore Nernst equation holds even when current flows \cite[Eq. (5:8:6)]{Oldham1994}.

The IDAE is assumed to operate in dual mode, that is,
different potentials are applied to each array of bands.
If one array is potentiostated and its complementary array performs as counter electrode,
then it is said that the IDAE operates with \emph{internal counter electrode}.
On the other hand, if both arrays are potentiostaded independently,
an additional counter electrode would be required,
and the IDAE is said to operate with \emph{external counter electrode}.

\subsubsection{Properties of initial and final conditions}


If the common length $L$ of the electrodes is large enough, then
the concentration profile of the electrochemical species $\sigma$
doesn't change along the depth of the cell, and therefore,
it can be reduced to two dimensions $c_{\sigma}(x,z,t)$.
Here $x$ and $z$ are the horizontal and vertical coordinates respectively,
and $t$ corresponds to time.

In this case, the concentration profile of species $\sigma \in \{O, R\}$
is given initially by $c_{\sigma}(x,z,0^{-}) = c_{\sigma,i}(x,z)$,
and it is assumed to come from a previous steady state.
Under this condition, the initial profile has an average
(along any horizontal line spaning the whole cell)
that is independent of $z$ \cite[Remark 2.2]{GuajardoYevenes2013sep}
\begin{subequations}
	\label{main:eqn:ci:properties}
	\begin{gather}
		\bar{c}_{\sigma,i}^{\whole}
		:= \frac{1}{W_{T}} \int_{\whole} c_{\sigma,i}(x,z) \ud{x}
		\\
		\intertext{%
			due to the fact that a coplanar counter electrode is included
			in the whole cell.
			Similarly, the weighted sum of initial concentrations
			is independent of $(x,z)$
			\cite[Remark 2.1]{GuajardoYevenes2013sep} and equals
		}
		\label{main:eqn:ci:total}
		D_{O} c_{O,i}(x,z) + D_{R} c_{R,i}(x,z)
		= D_{O} \bar{c}_{O,i}^{\whole} + D_{R} \bar{c}_{R,i}^{\whole}
	\end{gather}
\end{subequations}

After the cell is potentiostated,
the concentration profile is shifted out from its initial state,
entering a new steady state $c_{\sigma}(x,z,+\infty) = c_{\sigma,f}(x,z)$
after a sufficiently long time compared with the characteristic time of the cell.
In this final steady state, the average of concentration
(along any horizontal line spaning the whole cell) remains independent of $z$
and equals its initial counterpart $\bar{c}_{\sigma,i}^{\whole}$
\cite[Remark 2.2]{GuajardoYevenes2013sep}
\begin{subequations}
	\label{main:eqn:cf:properties}
	\begin{gather}
		\label{main:eqn:cf:average}
		\frac{1}{W_{T}} \int_{\whole} c_{\sigma,f}(x,z) \ud{x}
		= \bar{c}_{\sigma,i}^{\whole}
		\\
		\intertext{%
			due to the fact that the whole cell includes a coplanar counter electrode.
			Similarly, the weighted sum of concentrations
			remains independent of $(x,z)$ and equals its initial counterpart
			\cite[Remark 2.1]{GuajardoYevenes2013sep} %
		}
		\label{main:eqn:cf:total}
		D_{O} c_{O,f}(x,z) + D_{R} c_{R,f}(x,z)
		= D_{O} \bar{c}_{O,i}^{\whole} + D_{R} \bar{c}_{R,i}^{\whole}
	\end{gather}
\end{subequations}

\subsubsection{Properties of boundary conditions}

The final concentration $c_{\sigma,f}^{E}$ of species $\sigma \in \{O,R\}$,
at the surface $E \in \{A,B,C\}$ of each band and
the external counter electrode (if present), is governed by Nernst equation
\begin{equation}
	\label{main:eqn:eta}
	\eta_{f}^{E}
	:= \ln\! \left( \frac{c_{O,f}^{E}}{c_{R,f}^{E}} \right)
	= \frac{F n_{e}}{RT} (V_{f}^{E} - {V^{o}}')
\end{equation}
However, since the weighted sum of concentrations on the surface of $E$
is related to the initial concentrations, due to Eq. \eqref{main:eqn:cf:total}
\begin{equation}
	\label{main:eqn:cfE:total}
	D_{O} c_{O,f}^{E} + D_{R} c_{R,f}^{E}
	= D_{O} \bar{c}_{O,i}^{\whole} + D_{R} \bar{c}_{R,i}^{\whole}
\end{equation}
then the Nernst equation can be decoupled,
leading to uniform concentrations on all electrodes $E \in \{A,B,C\}$
\begin{subequations}
	\label{main:eqn:nernst}
	\begin{align}
		c_{O,f}^{E} &=
		\frac{
			D_{O} \bar{c}_{O,i}^{\whole} + D_{R} \bar{c}_{R,i}^{\whole}
		}{
			D_{O} + D_{R} \e^{-\eta_{f}^{E}}
		} =
		\frac{
			D_{O} \bar{c}_{O,i}^{\whole} + D_{R} \bar{c}_{R,i}^{\whole}
		}{
			D_{O} \bar{c}_{O,i}^{\whole} + D_{R} \bar{c}_{R,i}^{\whole}
			\e^{-(\eta_{f}^{E} - \eta_{\nul})}
		} \bar{c}_{O,i}^{\whole}
		\\
		c_{R,f}^{E} &=
		\frac{
			D_{R} \bar{c}_{R,i}^{\whole} + D_{O} \bar{c}_{O,i}^{\whole}
		}{
			D_{R} + D_{O} \e^{\eta_{f}^{E}}
		} =
		\frac{
			D_{R} \bar{c}_{R,i}^{\whole} + D_{O} \bar{c}_{O,i}^{\whole}
		}{
			D_{R} \bar{c}_{R,i}^{\whole} + D_{O} \bar{c}_{O,i}^{\whole}
			\e^{(\eta_{f}^{E} - \eta_{\nul})}
		} \bar{c}_{R,i}^{\whole} 
	\end{align}
\end{subequations}
Here, $\eta_{f}^{E}$ and $V_{f}^{E}$ are the normalized
and applied potentials at the electrode $E$ in the final steady state,
${V^{o}}'$ is the formal potential of the redox couple, $F$ is the Faraday constant,
$R$ is the universal gas constant and $T$ is the temperature of the system.

Note that if the final concentrations on the electrodes satisfy
$c_{\sigma,f}^{A} = c_{\sigma,f}^{B} = c_{\sigma,f}^{C} = \bar{c}_{\sigma,i}^{\whole}$,
then there is no gradient of concentration generated inside the cell,
and therefore the current in the cell must equal zero.
This corresponds to the case when the null potential $V_{\nul}$
is applied to the electrodes $\eta_{f}^{E} = \eta_{\nul}$
\begin{equation}
	\label{main:eqn:null}
	\eta_{\nul}
	= \ln\left( \frac{\bar{c}_{O,i}^{\whole}}{\bar{c}_{R,i}^{\whole}} \right)
	= \frac{n_{e} F}{RT}(V_{\nul} - {V^{o}}')
\end{equation}

\subsubsection{Problem reduced to diffusion in the unit cell}


If the IDAE fits exactly within the cell,
such that the bands at both ends of the IDAE have half width,
then the cell in Fig. \ref{main:fig:celda:int_C} can be modeled exactly
as an assembly of two-dimensional unit cells,
like the one shown in Fig. \ref{main:fig:celda:unitaria}.
In case the IDAE doesn't fit exactly in the cell,
the configurations in Figs. \ref{main:fig:celda:int_C},\subref{main:fig:celda:ext_C}
still can be regarded as an assembly of two-dimensional unit cells
provided the following conditions:
(i) The length $L$ is large enough, so that the problem still can be reduced to two dimensions.
(ii) The number of bands is so large that the edge effects at the ends of the IDAE are negligible,
and it is still possible to consider symmetry boundary conditions for the unit cell
\cite[271]{Aoki1988dec}.

Under these conditions, the final concentration $c_{\sigma,f}(x,z)$
of the electrochemical species $\sigma \in \{O,R\}$
can be reduced to a problem of steady-state two-dimensional diffusion%
\footnote{
	for $\pm$ or $\mp$, the upper sign corresponds to $\sigma=O$,
	and the lower sign, to $\sigma=R$.
}
inside a representative unit cell

\begin{subequations}
	\label{main:eqn:pde}
	\begin{align}
		\parderiv{^2 c_{\sigma,f}}{x^2}(x,z)
		+ \parderiv{^2 c_{\sigma,f}}{z^2}(x,z) 
		&= 0
		\\
		\label{main:eqn:pde:cb:izda-dcha}
		\parderiv{c_{\sigma,f}}{x}(0,z)
		=\parderiv{c_{\sigma,f}}{x}(W,z) &= 0,\, \forall z \in [0,H]
		\\
		\label{main:eqn:pde:cb:arriba}
		\mp D_{\sigma} \parderiv{c_{\sigma,f}}{z}(x,H) &= 0,\, \forall x \in [0,W]
		\\
		\label{main:eqn:pde:cb:abajo:gap}
		\mp D_{\sigma} \parderiv{c_{\sigma,f}}{z}(x,0) &= 0,\, \forall x \notin A \cup B 
		\\
		\label{main:eqn:pde:cb:abajo:A}
		c_{\sigma,f}(x,0) &= c_{\sigma,f}^{A},\, \forall x \in A,
		\\
		\label{main:eqn:pde:cb:abajo:B}
		c_{\sigma,f}(x,0) &= c_{\sigma,f}^{B},\, \forall x \in B
	\end{align}
\end{subequations}
where $D_{\sigma}$ is the diffusion coefficient of the species $\sigma \in \{O,R\}$,
Eqs. (\ref{main:eqn:pde:cb:izda-dcha}--\ref{main:eqn:pde:cb:abajo:gap}) are symmetry/insulation boundary conditions,
and $c_{\sigma,f}^{A}$ and $c_{\sigma,f}^{B}$ are the final concentrations
at the surface of the bands $A$ and $B$ respectively, see Eqs. \eqref{main:eqn:nernst}.

\subsection{Exact solution in steady state}

\subsubsection{Transformation of the unit-cell domain}

Alternated concentration and insulation boundary conditions at the bottom of the unit cell,
Eqs. (\ref{main:eqn:pde:cb:abajo:gap}--\ref{main:eqn:pde:cb:abajo:B}),
make it difficult to obtain an analytical solution to Eqs. \eqref{main:eqn:pde}.
Nevertheless, through domain transformations,
it is possible to arrange these alternated boundary conditions,
so they can be placed at different walls in a transformed cell.

One convenient way to obtain such transformation is through complex conformal mappings,
since they leave the diffusion equation (in steady state),
as well as the concentration and non-flux/insulation boundary conditions,
invariant under domain changes \cite[\S5.7]{Driscoll2002}.
In particular, the complex Jacobian elliptic functions $\sn()$ and $\cd()$ are of interest,
since they conformally map a square domain into the upper half-plane
\cite[\S 2.5]{Driscoll2002} \cite[\S\dlmf{22.18.}{ii}]{dlmf}.
Also the Möbius functions are important,
since they are able to reorganize the upper half-plane,
by mapping into itself \cite[\S2.3]{Driscoll2002}.

One can now regard the concentration $c_{\sigma,f}(x,z)$
as a function $c_{\sigma,f}(\bm{r})$ of a complex variable
$\bm{r} = x + \bm{i} z$, where $\bm{i}^{2} = -1$,
and find a complex function $\bm{\rho} = T_{r}^{\rho}(\bm{r})$
that can transform the unit cell from the IDAE domain $\bm{r} = x + \bm{i}z$
to a parallel-plates domain $\bm{\rho} = \xi + \bm{i}\zeta$
as shown in Fig. \ref{main:fig:transformacion}.

\begin{figure*}[t]
	\centering
	\includegraphics{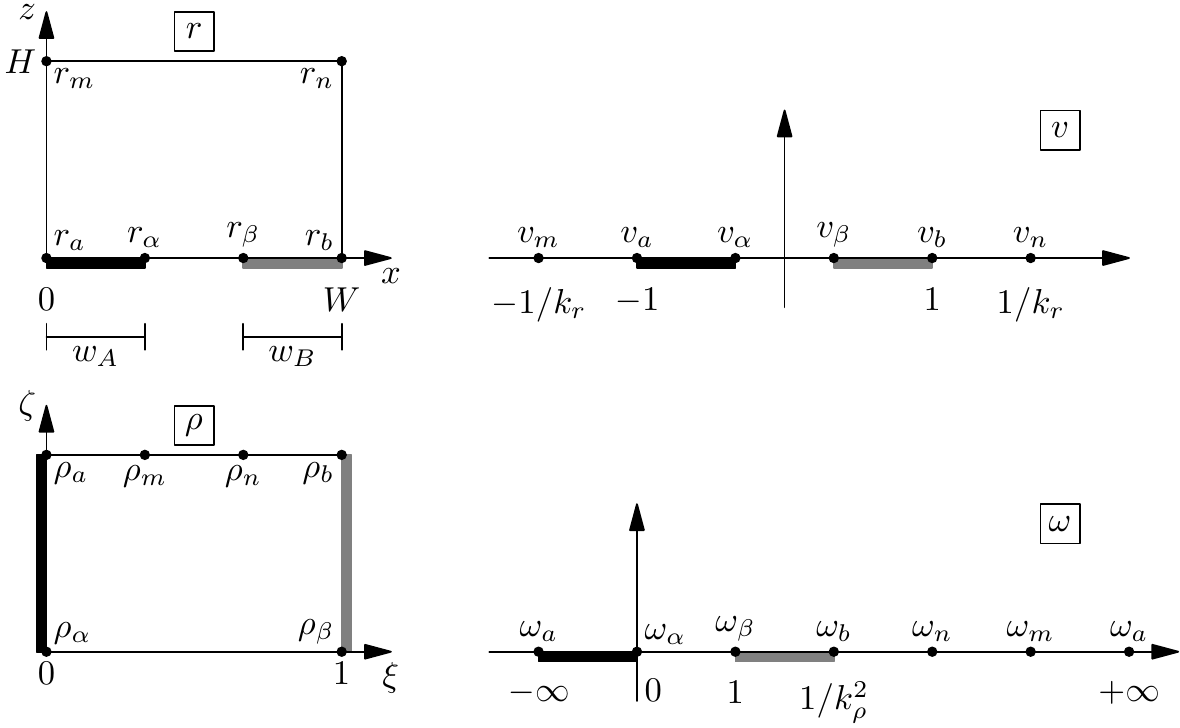}
	\caption{
		Complex transformation of the IDAE domain $\bm{r}=(x,z)$
		into the conformal parallel-plates domain $\bm{\rho}=(\xi,\zeta)$,
		by using the auxiliary complex domains $\bm{v}$ and $\bm{\omega}$.
		First, the transformation $\bm{r} \to \bm{v}$
		maps the interior of the unit cell into the upper half plane
		by using a Jacobian elliptic function.
		Later, the transformation $\bm{v} \to \bm{\omega}$
		reorganizes the structure of the upper half plane
		through a Möbius transformation.
		Finally, the transformation $\bm{\omega} \to \bm{\rho}$
		maps the upper half plane to the interior of the conformal parallel-plates cell,
		by using a composition of squared root and a inverse Jacobian elliptic function.
	}
	\label{main:fig:transformacion}
\end{figure*}

\begin{lema}
	\label{main:lem:T}
	The domain $\bm{r}=(x,z)$ of the IDAE unit cell can be conformally transformed
	into a parallel-plates configuration in the domain $\bm{\rho}=(\xi,\zeta)$,
	through the use of a complex transformation $\bm{\rho} = T_{r}^{\rho}(\bm{r})$
	(see Fig. \ref{main:fig:transformacion}).
	This domain transformation is given by
	\begin{subequations}
		\label{main:eqn:T}
		\begin{align}
			\bm{\rho} = T_{r}^{\rho}(\bm{r})
			&= T_{\omega}^{\rho} \circ T_{v}^{\omega} \circ T_{r}^{v}(\bm{r})
			\\
			\label{main:eqn:omega-rho}
			\bm{\rho} = T_{\omega}^{\rho}(\bm{\omega})
			&= \frac{1}{K(k_{\rho})} \arcsn(\sqrt{\bm{\omega}}, k_{\rho})
			\\
			\label{main:eqn:v-omega}
			\bm{\omega} = T_{v}^{\omega}(\bm{v})
			&= \frac{(\bm{v}-\bm{v}_{\alpha})}{(\bm{v}-\bm{v}_{a})}
			\frac{(\bm{v}_{\beta}-\bm{v}_{a})}{(\bm{v}_{\beta}-\bm{v}_{\alpha})}
			\\
			\label{main:eqn:r-v}
			\bm{v} = T_{r}^{v}(\bm{r})
			&= -\cd\!\left(K(k_{r}) \frac{2\bm{r}}{W}, k_{r}\right)
		\end{align}
	\end{subequations}
	where the moduli and their complements are
	\begin{subequations}
		\label{main:eqn:kr-krho}
		\begin{align}
			\label{main:eqn:krho}
			k_{\rho}^{2} &=
			\frac{
				2(\bm{v}_{\beta}-\bm{v}_{\alpha})
			}{
				(1-\bm{v}_{\alpha})(1+\bm{v}_{\beta})
			},
			&
			{k_{\rho}'}^{2} &=
			\frac{(1+\bm{v}_{\alpha})}{(1-\bm{v}_{\alpha})}
			\frac{(1-\bm{v}_{\beta})}{(1+\bm{v}_{\beta})}
			\\
			\label{main:eqn:kr}
			k_{r} &= Q^{-1}(\e^{-\pi 2H/W}),
			&
			k_{r}' &= Q^{-1}(\e^{-\pi W/2H})
		\end{align}
	\end{subequations}
	and the following points on the boundary of the IDAE domain $\bm{r}$
	\begin{subequations}
		\begin{align}
			\bm{r}_{a} &= (0, 0), & \bm{r}_{\alpha} &= (w_{A}, 0) \\
			\bm{r}_{b} &= (W, 0), & \bm{r}_{\beta}  &= (W-w_{B}, 0)
		\end{align}
	\end{subequations}
	are transformed to the auxiliary domain $\bm{v}$ as
	\begin{subequations}
		\label{main:eqn:vAB}
		\begin{align}
			\label{main:eqn:vAB:A}
			\bm{v}_{a} = T_{r}^{v}(\bm{r}_{a}) &= -1, 
			& 
			\bm{v}_{\alpha} = T_{r}^{v}(\bm{r}_{\alpha}) &= 
			-\cd\!\left(K(k_{r}) \frac{2 w_{A}}{W}, k_{r}\right) 
			\\
			\label{main:eqn:vAB:B}
			\bm{v}_{b} = T_{r}^{v}(\bm{r}_{b}) &= \phantom{-}1, 
			& 
			\bm{v}_{\beta} = T_{r}^{v}(\bm{r}_{\beta}) &= 
			\phantom{-}\cd\!\left(K(k_{r}) \frac{2 w_{B}}{W}, k_{r}\right)
		\end{align}
	\end{subequations}
	and to the parallel-plates domain $\bm{\rho}$ as
	\begin{subequations}
		\label{main:eqn:rhoAB}
		\begin{align}
			\bm{\rho}_{a} = T_{r}^{\rho}(\bm{r}_{a}) 
			&= \left(0, \frac{K'(k_{\rho})}{K(k_{\rho})}\right),
			&
			\bm{\rho}_{\alpha} = T_{r}^{\rho}(\bm{r}_{\alpha}) &= (0,0)
			\\
			\bm{\rho}_{b} = T_{r}^{\rho}(\bm{r}_{b}) 
			&= \left(1, \frac{K'(k_{\rho})}{K(k_{\rho})}\right),
			&
			\bm{\rho}_{\beta} = T_{r}^{\rho}(\bm{r}_{\beta}) &= (1,0)
		\end{align}
	\end{subequations}
\end{lema}


\noindent For a detailed construction of this transformation see
\emph{Supplementary information} \S\ref{si:proof:T}.

Several special functions have been used for the definition of the conformal transformation $T_{r}^{\rho}$:
$\cd(\bm{u},k)$ and $\arcsn(\bm{u},k)$ correspond to Jacobian elliptic functions
\citex{Eqs. (\dlmf[E]{22.2.}{8}) and (\dlmf[E]{22.15.}{12})}{dlmf}
analogous to their circular counterparts $\cos()$ and $\arcsin()$ respectively,
$K(k)$ and $K'(k)$ correspond to the \emph{complete elliptic integral of the first kind}
and its associated function respectively
\citex{Eqs. (\dlmf[E]{19.2.}{4}), (\dlmf[E]{19.2.}{8}) and (\dlmf[E]{19.2.}{9})}{dlmf},
and $Q(k)$ corresponds to the \emph{elliptic nome function} \citex{Eq. (\dlmf[E]{22.2.}{1})}{dlmf}.

The conformal transformation in Lemma \ref{main:lem:T} agrees with that
in \cite[Eqs. (9), (10), (14), (15) and (19)]{Aoki1988dec}%
\footnote{
	Note that in \cite{Aoki1988dec}
	the parameter $p=k_{\rho}^{2}$ is used instead of the modulus $k_{\rho}$.
}
for semi-infinite geometries,
since $k_{r} \to 0^{+}$, $\cd(\cdot, k_{r}) \to \cos(\cdot)$
and $K(k_{r}) \to \pi/2$ when $H \to +\infty$.

\subsubsection{Concentration profile}

If the concentration profile is now written in terms of
the conformal parallel-plates domain $(\xi,\zeta) = T_{r}^{\rho}(x,z)$
by using Lema \ref{main:lem:T}
\begin{equation}
	\gamma_{\sigma,f}(\xi,\zeta)
	= \gamma_{\sigma,f} \circ T_{r}^{\rho}(x,z)
	= c_{\sigma,f}(x,z)
\end{equation}
then the Laplacian operator in the IDAE and the parallel-plates domains
are related \cite[\S2.1 Problem 7, Eq. (5.4.17)]{Ablowitz2003apr},
\cite[Eq. (5.20)]{Driscoll2002}, \cite[Eq. (6.3)]{Olver2018mar}
\begin{equation}
	\parderiv{^{2}c_{\sigma,f}}{x^{2}}(x,z)
	+ \parderiv{^{2}c_{\sigma,f}}{z^{2}}(x,z)
	= \left|\parderiv{\bm{r}}{\bm{\rho}}(\xi,\zeta)\right|^{-2}
	\left[
		\parderiv{^{2} \gamma_{\sigma,f}}{\xi^2}(\xi,\zeta) 
		+ \parderiv{^{2} \gamma_{\sigma,f}}{\zeta^2}(\xi,\zeta)
	\right]
\end{equation}
meaning that the diffusion equation in steady state
is invariant under conformal domain transformations.
The same is true for insulation/symmetry boundary conditions, that is,
they remain invariant under conformal transformations,
since the angles between the iso-concentration lines and the flux lines are preserved.

Therefore, the steady-state diffusion problem in Eqs. \eqref{main:eqn:pde}
is transformed to the conformal parallel-plates domain $\bm{\rho}=(\xi,\zeta)$
\begin{subequations}
	\label{main:eqn:pde:parallel}
	\begin{align}
		\parderiv{^2 \gamma_{\sigma,f}}{\xi^2}(\xi,\zeta) 
		+ \parderiv{^2 \gamma_{\sigma,f}}{\zeta^2}(\xi,\zeta)
		&= 0
		\\
		\label{main:eqn:pde:parallel:simetria}
		\parderiv{\gamma_{\sigma,f}}{\zeta}(\xi,0) 
		= \parderiv{\gamma_{\sigma,f}}{\zeta}(\xi,\zeta_{a})
		&= 0
		\\
		\gamma_{\sigma,f}(0,\zeta) &= c_{\sigma,f}^{A}
		\\
		\gamma_{\sigma,f}(1,\zeta) &= c_{\sigma,f}^{B}
	\end{align}
\end{subequations}
where $\zeta_{a} = \Im\bm{\rho}_{a} = K'(k_{\rho})/K(k_{\rho})$
corresponds to the imaginary part of $\bm{\rho}_{a}$ defined in Eq. \eqref{main:eqn:rhoAB}.
This leads to the following solution for the concentration profile in steady state

\begin{teorema}
	\label{main:teo:cf}
	If the IDAE can be modeled as an assembly of unit cells and
	the final concentrations on the bands $A$ and $B$ are uniform,
	such as in \S\ref{main:def:problema},
	then the concentration profile in the final steady state is given by
	\begin{equation}
		\label{main:eqn:cf}
		c_{\sigma,f}(x,z) = c_{\sigma,f}^{A} + [c_{\sigma,f}^{B} - c_{\sigma,f}^{A}] \xi(x,z)
	\end{equation}
	where $c_{\sigma,f}^{E}$ is the steady-state concentration of
	the species $\sigma \in \{O,R\}$ on each band $E \in \{A,B\}$,
	defined in Eq. (\ref{main:eqn:nernst}),
	and $\xi(x,z) = \Re\, T_{r}^{\rho}(x,z)$ corresponds to
	the real part of the conformal transformation $T_{r}^{\rho}(x,z)$,
	defined in Eq. (\ref{main:eqn:T}).
\end{teorema}

\begin{proof}
	Due to the symmetry boundary conditions in Eq. \eqref{main:eqn:pde:parallel:simetria},
	the parallel-plates cell can be extended infinitely along the $\zeta$-axis.
	The concentration profile $\gamma_{\sigma,f}(\xi,\zeta)$ for
	the infinite parallel-plates cell doesn't depend on the $\zeta$-coordinate,
	therefore the steady-state solution of Eq. \eqref{main:eqn:pde:parallel}
	is given by a linear interpolation of the concentration at each of the plates
	\begin{equation}
		\gamma_{\sigma,f}(\xi,\zeta)
		= c_{\sigma,f}^{A} + [c_{\sigma,f}^{B} - c_{\sigma,f}^{A}] \xi
	\end{equation}
	The desired result is finally obtained by applying
	the conformal transformation $(\xi,\zeta) = T_{r}^{\rho}(x,z)$
	in Lema \ref{main:lem:T} to return to the IDAE domain.
\end{proof}

Note that the expression for the concentration profile
given by Eqs. (\ref{main:eqn:cf}) and (\ref{main:eqn:nernst})
agrees with that in \cite[Eqs. (20) and (21)]{Aoki1988dec}%
\footnote{
	Note that \cite[$H$ in Eq. (21)]{Aoki1988dec} has a typo,
	and the $+$ sign, at the middle of the expression,
	should be replaced by a $-$ sign.
} for semi-infinite geometries,
which depends directly on the real part of the conformal transformation $T_{r}^{\rho}$ in Lemma \ref{main:lem:T}.

\subsubsection{Current density}


\begin{corolario}
	\label{main:cor:jf}
	Under the conditions of Theorem \ref{main:teo:cf},
	the current density $j_{f}(x)$ in the final steady state is given by
	\begin{equation}
		\label{main:eqn:jf}
		\mp \frac{j_{f}(x)}{F n_e D_{\sigma} [c_{\sigma,f}^{B} - c_{\sigma,f}^{A}]}
		= \parderiv{\xi}{z}(x,0)
		= -\Im \parderiv{\bm{\rho}}{\bm{r}}(x) 
	\end{equation}
	where $\Im\{\}$ corresponds to the imaginary part and the complex derivative of $\bm{\rho} = T_{r}^{\rho}(\bm{r})$ is given by
	\begin{equation}
		\label{main:eqn:dTdr}
		\parderiv{\bm{\rho}}{\bm{r}}(\bm{r}) 
		= \bm{i}\, \frac{k_{r}'}{W} \frac{K(k_{r})}{K(k_{\rho})}
		\frac{
			(1-\bm{v}_{\alpha})^{1/2} (1+\bm{v}_{\beta})^{1/2}
		}{
			(\bm{v}-\bm{v}_{\alpha})^{1/2} (\bm{v}-\bm{v}_{\beta})^{1/2}
		}
		\nd\!\left(K(k_{r}) \frac{2\bm{r}}{W}, k_{r}\right)
	\end{equation}
	of which its parameters are defined in Eqs. \eqref{main:eqn:kr-krho} and \eqref{main:eqn:vAB}, and $\nd()$ corresponds to a Jacobian elliptic function, defined in \citex{Eq. (\dlmf[E]{22.2.}{6})}{dlmf}.
\end{corolario}

\begin{proof}
	Applying Fick's law to Eq. \eqref{main:eqn:cf} leads to
	\begin{equation}
		\label{main:eqn:jf:pre}
		j_{f}(x) = \mp F n_e D_{\sigma} \parderiv{c_{\sigma,f}}{z}(x,0)
		= \mp F n_e D_{\sigma} [c_{\sigma,f}^{B} - c_{\sigma,f}^{A}]
		\parderiv{\xi}{z}(x,0)
	\end{equation}
	Later, using the Cauchy-Riemann equations
	\cite[Theorem 3.2]{Olver2018mar}
	\begin{equation}
		\label{main:eqn:dxidz}
		\parderiv{\xi}{z} = -\parderiv{\zeta}{x} 
		= -\Im \parderiv{\bm{\rho}}{x} = -\Im \parderiv{\bm{\rho}}{\bm{r}} 
	\end{equation}
	leads to Eq. \eqref{main:eqn:jf}, which ends the main proof.
	The rest of the proof concerns about obtaining
	the complex derivative of $T_{r}^{\rho}$ in Eq. \eqref{main:eqn:dTdr},
	which is detailed in \emph{Supplementary information} \S\ref{si:proof:dTdr}.
\end{proof}

The expression for the current density in Eq. \eqref{main:eqn:jf}
agrees with that in \cite[Eqs. (26) and (19)]{Aoki1988dec}
for semi-infinite geometries.
Both results depend directly on the imaginary part of
the complex derivative $\partial\bm{\rho}(x)/\partial\bm{r}$ and
the product \cite[$c^{*}H$ given by Eqs. (6) and (21)]{Aoki1988dec}%
\footnote{
	Note that \cite[$H$ in Eq. (21)]{Aoki1988dec} has a typo,
	and the $+$ sign, at the middle of the expression,
	should be replaced by a $-$ sign.
}
corresponds to $[c_{\sigma,f}^B - c_{\sigma,f}^{A}]$.

The expression for $\Im\partial\bm{\rho}(x)/\partial\bm{r}$ in Eq. \eqref{main:eqn:dTdr}
also agrees with that in \cite[Eqs. (27) and (17)]{Aoki1988dec}
for semi-infinite geometries,
which can be seen with the help of the identity
\begin{equation}
	\cos(\alpha + \beta) + \cos(\alpha - \beta)
	= (1 + \cos(2\alpha))^{1/2} (1 + \cos(2\beta))^{1/2}
\end{equation}
and because $k_{r} \to 0^{+}$, $k_{r}' \to 1^{-}$, $K(k_{r}) \to \pi/2$,
$\cd(\cdot, k_{r}) \to \cos(\cdot)$, $\nd(\cdot, k_{r}) \to 1$,
when $H \to +\infty$.

\subsubsection{Current per band}


\begin{corolario}
	\label{main:cor:if}
	Under the conditions of Theorem \ref{main:teo:cf},
	the current in the final steady state $i_{f}^{E}$
	flowing through one band $E \in \{A,B\}$ is
	\begin{equation}
		\label{main:eqn:if}
		\pm \frac{
			i_{f}^{A}/L
		}{
			F n_{e} D_{\sigma} [c_{\sigma,f}^{A} - c_{\sigma,f}^{B}]
		}
		= \pm \frac{
			i_{f}^{B}/L
		}{
			F n_{e} D_{\sigma} [c_{\sigma,f}^{B} - c_{\sigma,f}^{A}]
		}
		= 2 \frac{K'(k_{\rho})}{K(k_{\rho})}
	\end{equation}
	where the modulus $k_{\rho}$ and its complement $k_{\rho}'$ are defined in Eq. (\ref{main:eqn:krho}).
\end{corolario}

\begin{proof}
	The current $i_{f}^{E}$ can be obtained by integrating the flux through one band $E \in \{A,B\}$
	\begin{equation}
		i_{f}^{E} 
		= \mp \int_{E} F n_{e} D_{\sigma} \parderiv{c_{\sigma,f}}{z}(x,0) L \ud{x}
		= \mp \int_{E} F n_{e} D_{\sigma} [c_{\sigma,f}^{B} - c_{\sigma,f}^{A}] \parderiv{\xi}{z}(x,0) L \ud{x}
	\end{equation}
	Using the Cauchy-Riemann identities in Eq. \eqref{main:eqn:dxidz}
	the current can be further simplified
	\begin{equation}
		\label{main:eqn:int_dxi7dz_dx}
		\mp \frac{
			i_{f}^{E}/L
		}{
			F n_{e} D_{\sigma} [c_{\sigma,f}^{B} - c_{\sigma,f}^{A}]
		}
		= \int_{E} \parderiv{\xi}{z}(x,0) \ud{x} 
		= -\Im \int_{E} \parderiv{\bm{\rho}}{x}(x,0) \ud{x}
		= -\Im \int_{E} \partial\bm{\rho}(\bm{r})
	\end{equation}
	Due to symmetry this integral can be taken in a half band
	\begin{equation}
		\label{main:eqn:int_drho}
		-\Im \int_{A} \partial\bm{\rho}(\bm{r})
		= -2\,\Im\bm{\rho}(\bm{r})\Big|_{\bm{r}_{a}}^{\bm{r}_{\alpha}},\quad
		-\Im \int_{B} \partial\bm{\rho}(\bm{r})
		= -2\,\Im\bm{\rho}(\bm{r})\Big|_{\bm{r}_{\beta}}^{\bm{r}_{b}}
	\end{equation}
	Since $\bm{\rho} = T_{r}^{\rho}(\bm{r})$ and due to Eq. (\ref{main:eqn:rhoAB}),
	the result in Eq. (\ref{main:eqn:if}) can be obtained.
\end{proof}

Note that Eq. (\ref{main:eqn:if}) agrees with the result
in \cite[Eqs. (28)]{Aoki1988dec}%
\footnote{
	Note that in \cite{Aoki1988dec} the parameter $p=k_{\rho}^{2}$
	is used instead of the modulus $k_{\rho}$.
} for semi-infinite geometries,
where the product \cite[$c^{*}H$ given by Eqs. (6) and (21)]{Aoki1988dec}%
\footnote{
	Note that \cite[$H$ in Eq. (21)]{Aoki1988dec} has a typo,
	and the $+$ sign, at the middle of the expression,
	should be replaced by a $-$ sign.
}
corresponds to $[c_{\sigma,f}^B - c_{\sigma,f}^{A}]$.

\subsubsection{Voltammogram}

The voltammogram in steady state corresponds just to the expression of current,
obtained in Eq. \eqref{main:eqn:if}, as the potential applied to the bands is scanned.
During the scan, all parameters of the expression for the current remain unchanged,
except for the difference of concentrations $c_{\sigma,f}^{A} - c_{\sigma,f}^{B}$,
which varies according to the potential due to Eqs. \eqref{main:eqn:nernst}.

Therefore, the shape of the voltammogram is proportional to
the shape of $c_{\sigma,f}^{A} - c_{\sigma,f}^{B}$ as the potential is scanned.
This difference is analyzed in two cases:
when the IDAE has an external and an internal counter electrode.

\paragraph{Case of external counter electrode}

This case is the simplest to analyze,
but its experimental setup is relatively complex,
since it requires a bipotentiostat connected to the IDAE
and an external counter electrode.

\begin{teorema}
	\label{main:teo:cE-cE':extC}
	Consider the electrodes $A$ and $B$ with an external counter electrode,
	which undergo reversible electrode reations
	and satisfy Eq. \eqref{main:eqn:cfE:total}
	(like in \S\ref{main:def:problema} with Fig. \ref{main:fig:celda:ext_C}).
	If the potential at $E \in \{A, B\}$ is scanned and
	its complementary electrode $E'$ is fixed to an extreme potential,
	such that $c_{\sigma,f}^{E'} = 0$,
	then the difference of final concentrations is given by
	\begin{equation}
		\label{main:eqn:cE-cE':extC}
		[c_{\sigma,f}^{E} - c_{\sigma,f}^{E'}]
		= c_{\sigma,f}^{E}
		= \frac{
			D_{\sigma} \bar{c}_{\sigma,i}^{\whole}
			+ D_{\sigma'} \bar{c}_{\sigma',i}^{\whole}
		}{
			D_{\sigma} \bar{c}_{\sigma,i}^{\whole}
			+ D_{\sigma'} \bar{c}_{\sigma',i}^{\whole}
			\e^{\mp (\eta_{f}^{E} - \eta_{\nul})}
		} \bar{c}_{\sigma,i}^{\whole}
	\end{equation}
where $\sigma'$ is the complementary redox especies of $\sigma \in \{O, R\}$,
$D_{\sigma'}$ and $D_{\sigma}$ are their diffusion coefficients,
$\bar{c}_{\sigma,i}^{\whole}$ and $\bar{c}_{\sigma,'i}^{\whole}$
are the average concentrations of the whole cell in the initial steady state,
$\eta_{f}^{E}$ is the normalized potential applied to $E$,
defined in Eq. \eqref{main:eqn:eta},
and $\eta_{\nul}$ is the normalized null potential in Eq. \eqref{main:eqn:null}.
\end{teorema}

\begin{proof}
	Eq. \eqref{main:eqn:cE-cE':extC} is obtained directly from
	Eq. (\ref{main:eqn:nernst}), when $c_{\sigma,f}^{E'} = 0$.
\end{proof}

\paragraph{Case of internal counter electrode}

In this case the experimental setup is simpler,
since it requires a conventional potentiostat connected to both arrays of the IDAE.
Moreover, when the internal counter electrode serves as reference,
no potentiostat would be necessary and a voltage source with a sensitive ammeter would suffice as instrumentation \cite[end of p. 33]{Rahimi2009aug}.

However, the analysis is not direct, since
the potential at the counter electrode bands
is controlled automatically by the potentiostat,
and therefore it is unknown before performing the experiment.
Moreover, when the internal counter electrode acts as reference,
even the potential at the working electrode is unknown and only the voltage
(difference of potential) between working and counter electrodes would be known.

One solution is to consider the case of bands of equal width $2w_{A} = 2w_{B}$.
Here the average in Eq. \eqref{main:eqn:cf:average} is reduced to
$\bar{c}_{\sigma,i}^{\whole} = (c_{\sigma,f}^{A} + c_{\sigma,f}^{B})/2$
\cite[Eq. (15)]{Morf2006may}, due to symmetry, or equivalently
\begin{equation}
	\label{main:eqn:cfE:average}
	[c_{\sigma,f}^{A} - \bar{c}_{\sigma,i}^{\whole}] = -[c_{\sigma,f}^{B} - \bar{c}_{\sigma,i}^{\whole}]
\end{equation}
This allows an \emph{a priori} estimation of
the concentration at the counter electrode bands.
However, this also restricts the magnitude of the concentrations on the bands,
since they must be non-negative, and therefore,
they cannot decrease indefinitely below the average $\bar{c}_{\sigma,i}^{\whole}$
\begin{equation}
	- D_{\sigma} \bar{c}_{\sigma,i}^{\whole}
	\leq D_{\sigma} [c_{\sigma,f}^{A} - \bar{c}_{\sigma,i}^{\whole}]
	= - D_{\sigma} [c_{\sigma,f}^{B} - \bar{c}_{\sigma,i}^{\whole}]
	\leq D_{\sigma} \bar{c}_{\sigma,i}^{\whole}
\end{equation}
Due to Eq. \eqref{main:eqn:cf:total}, a similar situation occurs with the complementary species,
and has been graphically illustrated in \cite[Fig. 2]{GuajardoYevenes2013sep}.

\begin{lema}
	\label{main:lem:cE-cbar=cE'-cbar}
	Consider the electrodes $A$ and $B$,
	which undergo reversible electrode reactions
	and satisfy Eqs. \eqref{main:eqn:cfE:total} and \eqref{main:eqn:cfE:average}
	(like in \S\ref{main:def:problema} with Fig. \ref{main:fig:celda:int_C}
	and bands of equal width).
	If the electrode $E \in \{A,B\}$ and its complementary electrode $E'$
	perform as working and counter electrodes respectively, 
	then the difference of final concentrations of species $\sigma \in \{O,R\}$ is given by
	\begin{equation}
		\label{main:eqn:cE':intC}
		[c_{\sigma,f}^{E} - c_{\sigma,f}^{E'}]
		=  2[c_{\sigma,f}^{E} - \bar{c}_{\sigma,i}^{\whole}]
		= -2[c_{\sigma,f}^{E'} - \bar{c}_{\sigma,i}^{\whole}]
	\end{equation}
	
	Nevertheless, this difference cannot reach its ideal maximum,
	obtainable from Eq. \eqref{main:eqn:cE-cE':intC}, being limited from above and below by
	\begin{equation}
		\label{main:eqn:cE-cE':cotas}
		-2 D_{\lambda} \bar{c}_{\lambda,i}^{\whole}
		\leq D_{\sigma} [c_{\sigma,f}^{E} - c_{\sigma,f}^{E'}] 
		\leq 2 D_{\lambda} \bar{c}_{\lambda,i}^{\whole}
	\end{equation}
	where the limiting species $\lambda \in \{O,R\}$ is such that $D_{\lambda} \bar{c}_{\lambda,i}^{\whole} = \min(D_{O} \bar{c}_{O,i}^{\whole}, D_{R} \bar{c}_{R,i}^{\whole})$.
	This last expression determines the limiting current of the cell.
\end{lema}

\noindent See \emph{Supplementary information} \S\ref{si:proof:cE-cbar} for a detailed proof.

\begin{observacion}
	Note that, in the case of external counter electrode,
	the difference of concentrations is bounded by
	\begin{equation}
		0 \leq D_{\sigma} [c_{\sigma,f}^{E} - c_{\sigma,f}^{E'}] \leq
		D_{\sigma} \bar{c}_{\sigma,f}^{\whole} + D_{\sigma'} \bar{c}_{\sigma',f}^{\whole}
	\end{equation}
	due to Eq. \eqref{main:eqn:cE-cE':extC} when $\eta_{f}^{E} \to \pm \infty$.
	This determines the limiting current in the case of external counter electrode.
\end{observacion}

The result in Lemma \ref{main:lem:cE-cbar=cE'-cbar} allows us to estimate
the unknown concentration on the counter electrode,
thus facilitating an analytical expresion for the voltammogram

\begin{teorema}
	Under the assumptions of Lemma \ref{main:lem:cE-cbar=cE'-cbar},
	the difference of concentrations (voltammogram) in terms of the normalized potential is given by
	\begin{subequations}
		\label{main:eqn:voltammogram:intC}
		\begin{equation}
			\label{main:eqn:cE-cE':intC}
			[c_{\sigma,f}^{E} - c_{\sigma,f}^{E'}]
			= \frac{
				D_{\sigma'} \bar{c}_{\sigma',i}^{\whole}
				- D_{\sigma'} \bar{c}_{\sigma',i}^{\whole}
				\e^{\mp (\eta_{f}^{E} - \eta_{\nul})}
			}{
				D_{\sigma} \bar{c}_{\sigma,i}^{\whole}
				+ D_{\sigma'} \bar{c}_{\sigma',i}^{\whole}
				\e^{\mp (\eta_{f}^{E} - \eta_{\nul})}
			} \cdot 2\bar{c}_{\sigma,i}^{\whole}
		\end{equation}
		or in terms of the normalized voltage 
		\begin{equation}
			\label{main:eqn:etaE-etaE':intC}
			\begin{split}
				\pm (\eta_{f}^{E} - \eta_{f}^{E'})
				&= 2 \arctanh\left(
					\frac{
						c_{\sigma,f}^{E} - c_{\sigma,f}^{E'}
					}{
						2\bar{c}_{\sigma,i}^{\whole}
					}
				\right)
				\\
				& + 2 \arctanh\left(
					\frac{
						D_{\sigma} \bar{c}_{\sigma,i}^{\whole}
					}{
						D_{\sigma'} \bar{c}_{\sigma',i}^{\whole}
					}
					\cdot \frac{
						c_{\sigma,f}^{E} - c_{\sigma,f}^{E'}
					}{
						2\bar{c}_{\sigma,i}^{\whole}
					}
				\right)
			\end{split}
		\end{equation}
	\end{subequations}
	where $\sigma'$ is the complementary redox especies of $\sigma$,
	$D_{\sigma'}$ and $D_{\sigma}$ are their diffusion coeficients,
	$\bar{c}_{\sigma,i}^{\whole}$ and $\bar{c}_{\sigma,'i}^{\whole}$ are the average concentrations of the whole cell in the initial steady state,
	$\eta_{f}^{E}$ and $\eta_{f}^{E'}$ are the normalized potentials
	applied to the electrodes $E$ and $E'$, given in Eq. (\ref{main:eqn:eta}),
	and $\eta_{\nul}$ is the normalized null potential in Eq. \eqref{main:eqn:null}.
\end{teorema}

\begin{proof}
	Eq. \eqref{main:eqn:cE-cE':intC} is obtained when substracting
	Eq. (\ref{main:eqn:nernst}) with $\bar{c}_{\sigma,i}^{\whole}$,
	and later by applying $[c_{\sigma,f}^{E} - c_{\sigma,f}^{E'}] =
	2 [c_{\sigma,f}^{E} - \bar{c}_{\sigma,f}^{\whole}]$.
	The proof for Eq. \eqref{main:eqn:etaE-etaE':intC} can be found in
	\emph{Supplementary information} \S\ref{si:proof:etaE-etaE'}.
\end{proof}

\subsection{Approximations for shallow and tall cells}

For calculating the steady-state current through the cell,
one must evaluate the ratio $K'(k_{\rho})/K(k_{\rho})$ in Eq. \eqref{main:eqn:if},
which depends on several elliptic functions as seen in Eqs. 
\eqref{main:eqn:kr-krho} and \eqref{main:eqn:vAB}.
Currently, commercial and \emph{free and open source software} (FOSS) are available to aid in such calculations:
\textsf{Mathematica}, \textsf{Sage} and \textsf{SciPy}%
\footnote{
	Jacobian elliptic functions, the nome function and their inverses
	are available through the library \textsf{mpmath}
}
to name a few examples \cite[\S\dlmf{22.}{22}]{dlmf}.
However, standard scientific calculators and standard office software
(with \textsf{MS Office} and \textsf{LibreOffice} as common examples)
are not able to compute such special functions,
and therefore, approximations using trigonometric/hyperbolic functions are needed.

One convenient way to find such approximations is by using
the \emph{nome function} $q = Q(k)$ \cite[\S VI.3 Eq. (16)]{Nehari1952}, \citex{Eqs. (\dlmf[E]{19.2.}{9}) and (\dlmf[E]{22.2.}{1})}{dlmf}
\begin{equation}
	\label{elipticas:eqn:nomo}
	\frac{\ln Q(k)}{-\pi} = \frac{-\pi}{\ln Q(k')} = \frac{K'(k)}{K(k)}
\end{equation}
as a way to compute the ratio%
\footnote{
	this ratio is closely related to the \emph{lattice parameter}
	$\tau = \bm{i} K'(k)/K(k)$ and the \emph{nome}
	$q = \exp(\bm{i} \pi \tau)$
	\citex{\S\dlmf{20.}{1}, \S\dlmf{22.}{1}, and Eqs. (\dlmf[E]{22.2.}{1}) and (\dlmf[E]{22.2.}{12})}{dlmf}.
} $K'(k)/K(k)$.
This is because the Taylor expansion of the nome function $Q(k)$ is known
and converges relatively fast
\cite[below Eq. (12)]{Kneser1927} \citex{Eq. (\dlmf[E]{19.5.}{5})}{dlmf}

\begin{equation}
	\label{main:eqn:nomo:taylor}
	Q(k) =
	\frac{k^{2}}{16} + 8\left( \frac{k^{2}}{16} \right)^{2}
	+  84 \left( \frac{k^{2}}{16} \right)^{3}
	+ 992 \left( \frac{k^{2}}{16} \right)^{4} + O(k^{10})
\end{equation}
Therefore, for a sufficiently small modulus $k$ or $k'$,
the ratio $K'(k)/K(k)$ can be approximated with enough accuracy by using only the first term of the previous series
\begin{subequations}
	\label{main:eqn:K'K}
	\begin{align}
		\frac{K'(k)}{K(k)} \bigg|_{k \approx 0}
		&=  -\frac{1}{\pi} \ln Q(k) \bigg|_{k \approx 0}
		\approx -\frac{1}{\pi} \ln\!\left( \frac{k^{2}}{16} \right)
		\label{main:eqn:K'K:k0}
		\\
		\frac{K'(k)}{K(k)} \bigg|_{k' \approx 0}
		&=  -\pi [\ln Q(k')]^{-1} \bigg|_{k' \approx 0}
		\approx -\pi \left[ \ln\!\left( \frac{{k'}^{2}}{16} \right) \right]^{-1}
		\label{main:eqn:K'K:k'0}
	\end{align}
\end{subequations}

Also the following alternative representations of the moduli $k_{\rho}$ and $k_{\rho}'$ are useful for finding the desired approximations

\begin{lema}
	\label{main:lem:krho:alt}
	The modulus $k_{\rho}$ and the complementary modulus $k_{\rho}'$ have alternative representations to those given in Eq. (\ref{main:eqn:krho}).
	For $k_{\rho}$ the proposed alternative representation depends on the gap between consecutive bands $g = W - 2w_{E}$,
	when the width of both electrode bands is equal $2w_{E} = 2w_{A} = 2w_{B}$
	\begin{equation}
		\label{main:eqn:krho:alt}
		k_{\rho}^{2}
		=
		\frac{
			\displaystyle
			4\sn\!\left( K(k_{r}) \frac{g}{W}, k_{r} \right)
		}{
			\displaystyle
			\left[1 + \sn\!\left( K(k_{r}) \frac{g}{W}, k_{r} \right) \right]^{2}
		}
	\end{equation}

	For $k_{\rho}'$ the proposed alternative representation is given directly by
	\begin{equation}
		\label{main:eqn:krho':alt}
		{k_{\rho}'}^{2}
		=
		{k_{r}'}^{4}\,
		\frac{
			\displaystyle
			\sd\!\left( K(k_{r}) \frac{w_{A}}{W}, k_{r} \right)^{2}
		}{
			\displaystyle
			\cn\!\left( K(k_{r}) \frac{w_{A}}{W}, k_{r} \right)^{2}
		}
		\frac{
			\displaystyle
			\sd\!\left( K(k_{r}) \frac{w_{B}}{W}, k_{r} \right)^{2}
		}{
			\displaystyle
			\cn\!\left( K(k_{r}) \frac{w_{B}}{W}, k_{r} \right)^{2}
		}
	\end{equation}
\end{lema}

\begin{proof}
	Eq. (\ref{main:eqn:krho:alt}) is obtained from Eqs. (\ref{main:eqn:krho}) and (\ref{main:eqn:vAB}) when $2w_{A} = 2w_{B} = 2w_{E}$,
	later substituting $\cd()$ by $\sn()$
	\begin{equation}
		\cd(\bm{u},k) = \sn(\bm{u} - K(k) + 2K(k), k) = \sn(K(k) - \bm{u})
	\end{equation}
	(which is obtained from translation by quarter/half period
	\citex{Table \dlmf[T]{22.4.}{3}}{dlmf} and
	$\sn(-\bm{z},k) = -\sn(\bm{z},k)$ \citex{Table \dlmf[T]{22.6.}{1}}{dlmf})
	and finally by substituting $2w_{E} = W - g$.
	
	Similarly, Eq. (\ref{main:eqn:krho':alt}) is obtained from Eqs. (\ref{main:eqn:krho}) and (\ref{main:eqn:vAB}),
	and later using the identity \cite[Eqs. (1.10) and (4.1)]{Carlson2004nov}
	\begin{equation}
		\frac{1 - \cd(2\bm{u},k)}{1 + \cd(2\bm{u},k)}
		= {k'}^{2} \frac{\sd(\bm{u},k)^{2}}{\cn(\bm{u},k)^{2}}
	\end{equation}
\end{proof}

The last step remaining is to find trigonometric/hyperbolic approximations
of the moduli $k_{\rho}$ and $k_{\rho}'$,
so they can be plugged into Eqs. (\ref{main:eqn:K'K}).
This is shown below for the case of tall electrochemical cells
with small and large electrode bands.

\subsubsection{Case of tall cells}

In this case, the Jacobian elliptic functions $\sn()$ and $\sd()$ in Lemma \ref{main:lem:krho:alt}
can be approximated by their trigonometric counterpart $\sin()$.
In the same manner, $\cn()$ can be approximated by $\cos()$.
See \citex{Eqs. (\dlmf[E]{22.10.}{4})--(\dlmf[E]{22.10.}{6})}{dlmf} or \cite[Eqs. (10.1)--(10.3)]{Fenton1982jul}.

\begin{teorema}
	In case of tall electrochemical cells ($H$ is large),
	the modulus $k_{\rho}$ and the complementary modulus $k_{\rho}'$ are given by
	\begin{subequations}
		\label{main:eqn:moduli:Hinf}
		\begin{align}
			\bigg. k_{\rho}^{2} \bigg|_{
				\scriptsize \shortstack{
					$H \to +\infty$ \\
					$g \approx 0$
				}
			}
			&= \frac{
				4\sin(\pi g/2W)
			}{
				[1 + \sin(\pi g/2W)]^{2}
			}
			\approx 4 \left( \frac{\pi}{2}\frac{g}{W} \right)
			\label{main:eqn:krho:Hinf}
			\\
			\bigg. {k_{\rho}'}^{2} \bigg|_{
				\scriptsize \shortstack{
					$H \to +\infty$ \\
					$2w_{A} \approx 0$ \\
					$2w_{B} \approx 0$
				}
			}
			&= \tan\! \left( \frac{\pi}{2} \frac{w_{A}}{W} \right)^{2}
			\tan\! \left( \frac{\pi}{2} \frac{w_{B}}{W} \right)^{2}
			\approx \left( \frac{\pi}{2} \frac{w_{A}}{W} \right)^{2}
			\left( \frac{\pi}{2} \frac{w_{B}}{W} \right)^{2}
			\label{main:eqn:krho':Hinf}
		\end{align}
	\end{subequations}
	Therefore, the ratio $K'(k_{\rho})/K(k_{\rho})$
	for large and small electrodes is given respectively by
	\begin{subequations}
		\label{main:eqn:K'K:Hinf}
		\begin{align}
			\left. \frac{K'(k_{\rho})}{K(k_{\rho})} \right|_{
				\scriptsize \shortstack{
					$H \to +\infty$ \\
					$g \approx 0$
				}
			}
			&\approx -\frac{1}{\pi} \ln\left( \frac{\pi}{8} \frac{g}{W} \right)
			\label{main:eqn:K'K:Hinf:wEinf}
			\\
			\left. \frac{K'(k_{\rho})}{K(k_{\rho})} \right|_{
				\scriptsize \shortstack{
					$H \to +\infty$ \\
					$2w_{A} \approx 0$ \\
					$2w_{B} \approx 0$
				}
			}
			&\approx -\pi \left[
				2\ln\left( \frac{\pi}{4} \frac{w_{A}}{W} \right)
				+ 2\ln\left( \frac{\pi}{4} \frac{w_{B}}{W} \right)
			\right]^{-1}
			\label{main:eqn:K'K:Hinf:wE0}
		\end{align}
	\end{subequations}
	Here $g = W - 2w_{E}$ corresponds to the gap between consecutive bands,
	when they have equal width $2w_{E} = 2w_{A} = 2w_{B}$.
\end{teorema}

The results in Eqs. \eqref{main:eqn:krho:Hinf} and \eqref{main:eqn:K'K:Hinf:wEinf}
were first obtained by \cite[Eq. (32)]{Aoki1988dec},
for the case of large electrodes.
Almost two decades later, Eqs. \eqref{main:eqn:krho':Hinf} and \eqref{main:eqn:K'K:Hinf:wE0}
were obtained by \cite[Eqs. (2), (6) and (7)]{Morf2006may},
for the case of small electrodes.
See \emph{Supplementary information} \S\ref{si:proof:K'K:Hinf} for a proof of both results.
Also see Table \ref{main:tab:if_approx} in \emph{Results and discussion}
\S\ref{main:approx} to know the width of the electrode bands for which the approximations hold.


\subsubsection{Case of shallow cells}

Similarly, hyperbolic approximations \cite[Eqs. (11.1)--(11.3)]{Fenton1982jul}
of Jacobian elliptic functions in Lemma \ref{main:lem:krho:alt}
can be used for shallow electrochemical cells with small and large electrode bands.

\begin{teorema}
	In case of shallow electrochemical cells ($H$ is small),
	the modulus $k_{\rho}$ and the complementary modulus $k_{\rho}'$ are given by
	\begin{subequations}
		\label{main:eqn:moduli:H0}
		\begin{align}
			\bigg. k_{\rho}^{2} \bigg|_{
				\scriptsize \shortstack{
					$H \approx 0$ \\
					$g \approx 0$
				}
			}
			&\approx
			4 \coth\!\left( \frac{\pi}{4} \frac{W}{H} - \ln\sqrt{2} \right) 
			\tanh\!\left( \frac{\pi}{4} \frac{g}{W} \frac{W}{H} \right)
			\label{main:eqn:krho:H0}
			\\
			\bigg. {k_{\rho}'}^{2} \bigg|_{
				\scriptsize \shortstack{
					$H \approx 0$}}
			&\approx
			16 \e^{-\pi W/H} \sinh\!\left(
				\frac{\pi}{2} \frac{w_{A}}{W} \frac{W}{H}
			\right)^{2}
			\sinh\!\left( \frac{\pi}{2} \frac{w_{B}}{W} \frac{W}{H} \right)^{2}
			\label{main:eqn:krho':H0}
		\end{align}
	\end{subequations}
	Therefore, the normalized current in Eq. (\ref{main:eqn:if})
	for large and small electrodes is given respectively by
	\begin{subequations}
		\label{main:eqn:K'K:H0}
		\begin{align}
			\left. \frac{K'(k_{\rho})}{K(k_{\rho})} \right|_{
				\scriptsize \shortstack{
					$H \approx 0$ \\
					$g \approx 0$
				}
			}
			&\approx
			-\frac{1}{\pi} \left[
				\ln\tanh\!\left( \frac{\pi}{4} \frac{g}{H} \right)
				- \ln\tanh\!\left( \frac{\pi}{4} \frac{W}{H} - \ln\sqrt{2} \right) 
				- \ln 4
			\right]
			\label{main:eqn:K'K:H0:wEinf}
			\\
			\left. \frac{K'(k_{\rho})}{K(k_{\rho})} \right|_{
				\scriptsize \shortstack{
					$H \approx 0$ \\
					$2w_{A} \approx 0$ \\
					$2w_{B} \approx 0$
				}
			}
			&\approx
			-\pi \left[
				2 \ln\sinh\! \left( \frac{\pi}{2} \frac{w_{A}}{H} \right)
				+ 2 \ln\sinh\! \left( \frac{\pi}{2} \frac{w_{B}}{H} \right) 
				- \pi \frac{W}{H}
			\right]^{-1}
			\label{main:eqn:K'K:H0:wE0}
		\end{align}
	\end{subequations}
	Here the gap between consecutive bands equals $g = W - 2w$,
	when they have equal width $2w = 2w_{A} = 2w_{B}$.
\end{teorema}

See \emph{Supplementary information} \S\ref{si:proof:K'K:H0} for a proof.
Also see Table \ref{main:tab:if_approx} in \emph{Results and discussion}
\S\ref{main:approx} to know the cell dimensions for which the approximations hold.

%% file: text-results_discussion.tex

\section{Results and discussion}

The scripts\footnote{%
	All python scripts are available online at \url{\elektrodo}.%
} for obtaining the plots and simulations in this section
were written in \href{https://www.python.org/}{\textsf{Python}}
using the \href{https://scipy.org/}{\textsf{SciPy}} stack \cite{Jones2016jan}.
In particular, Jacobian elliptic functions
(required for plotting the theoretical concentration profile, current density and current)
were obtained from the library \href{http://mpmath.org/}{\textsf{mpmath}} \cite{Johansson2015dec}.
The simulations of the concentration profile and current were computed
using the finite element library \href{http://www.ctcms.nist.gov/fipy/}{\textsf{FiPy}} \cite{Guyer2009may}.

\subsection{Concentration profile}

\begin{figure}[t]
	\centering
	\subcaptionbox{
		$H/W = 1$.
		\label{main:fig:xi:HW10}
	}{
		\includegraphics{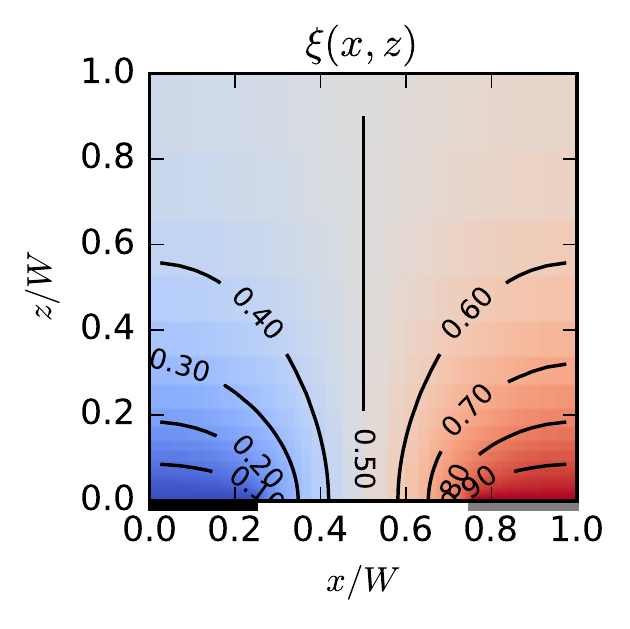}
	}
	\subcaptionbox{
		$H/W \approx \num{0,3}$.
		\label{main:fig:xi:HW03}
	}{
		\includegraphics{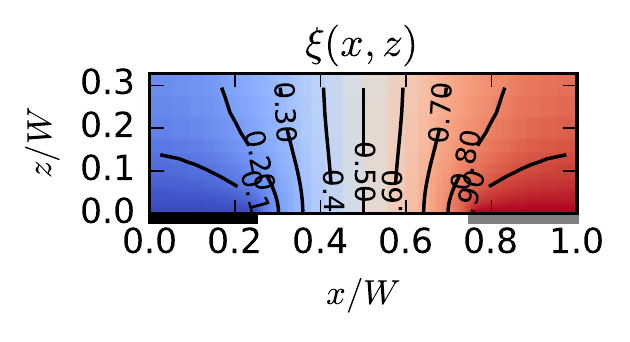}
	}
	\caption{
		Normalized concentration
		$\xi(x,z) = [c_{\sigma,f}(x,z) - c_{\sigma,f}^{A}]/[c_{\sigma,f}^{B} - c_{\sigma,f}^{A}]$
		in the final steady state Eq. \eqref{main:eqn:cf}
		(iso-concentration lines shown in black).
		Plots consider electrodes of equal size $2w_{A} = 2w_{B} = \num{0,5}W$
		and different aspect ratios $H/W$ for the unit cell.
		The absolute error between theoretical concentration and its simulated counterpart was not greater than $\approx\num{0,003}$ for this two cases.
		(See \emph{Supplementary information} \S\ref{si:sec:xi} for error plots).
	}
	\label{main:fig:xi}
\end{figure}

The concentration profile in Eq. \eqref{main:eqn:cf}
was constrasted against simulations of Eqs. \eqref{main:eqn:pde} for
different values of aspect ratio $H/W \in \mathopen] 0, \num{1,25} \mathclose[$
and electrode widths $2w_{A}/W = 2w_{B}/W \in \mathopen] 0, 1 \mathclose[$.
In 43 of 45 analyzed cases, the maximum absolute error was not greater than
$\approx \num{0,0045}$ (which corresponds to two decimal places of accuracy),
confirming the correctness of the expression obtained in Eq. \eqref{main:eqn:cf}.
See Table \ref{si:tab:exhaustive:xi} in \emph{Supplementary information} \S\ref{si:sec:exhaustive} for more details.

Fig. \ref{main:fig:xi} shows the behavior of the concentration profile
under the cases of high and low aspect ratios $H/W$
(simulation errors and other details are shown in \emph{Supplementary information} \S\ref{si:sec:xi}).
In case of high aspect ratio $H/W$ (Fig. \ref{main:fig:xi:HW10}),
the concentration becomes uniform far from the surface of the electrodes
(near the roof of the cell)
and full radial diffusion is present,
thus mimicking the behavior of semi-infinite configurations
\cite[Figs. 6 and 9]{Strutwolf2005feb} \cite[p. 6406 and Fig. S3]{Heo2013}
\cite[Fig. 39]{Heo2014jun} \cite[309]{Kanno2014}.
Moreover, according to \cite[Theorem 2.6]{GuajardoYevenes2013sep},
$H/W \gtrsim 1$ corresponds to the condition under which
the concentration of a finite-height cell and
its semi-infinite counterpart have similar profiles,
which agrees with Fig. \ref{main:fig:xi:HW10}.

In the case of low aspect ratio $H/W$ (Figs. \ref{main:fig:xi:HW03}),
the concentration profile doesn't exhibit a region of uniform concentration,
and radial diffusion is truncated by the low roof the cell.
Similar results obtained in
\cite[Fig. 6]{Strutwolf2005feb} \cite[\S3.2]{GuajardoYevenes2013sep}
\cite[p.6406, Fig. S3]{Heo2013} \cite[Fig. 39]{Heo2014jun}
\cite[p.309]{Kanno2014} confirm these facts.
Moreover, when the aspect ratio is low enough
(compare supplementary Figs. \ref{si:fig:xi_sim:HW05} and \ref{si:fig:xi_sim:HW03}),
the vertical gradient of concentration is almost inexistent,
being the horizontal gradient the only one that is clearly appreciable,
which resembles that of a parallel-plates configuration
in the gap between consecutive electrode bands.

\subsection{Current density}

\begin{figure}
	\centering
	\subcaptionbox{Selected values of $H/W$.}{\includegraphics{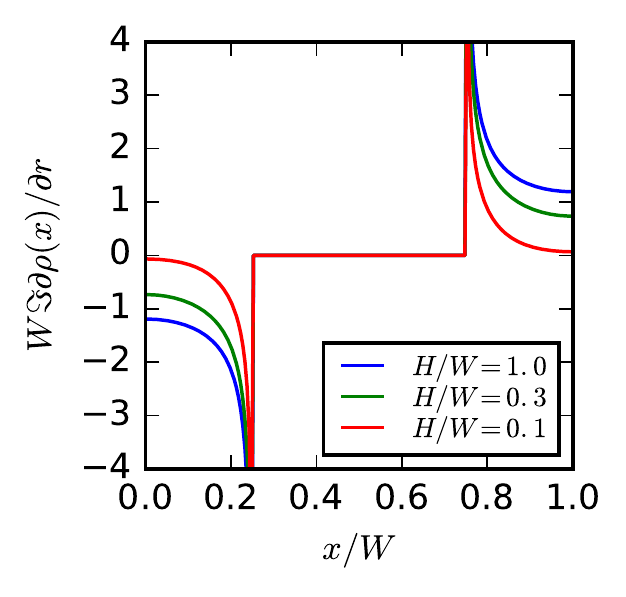}}
	\subcaptionbox{Range of values $H/W \in [0,1]$.}{\includegraphics{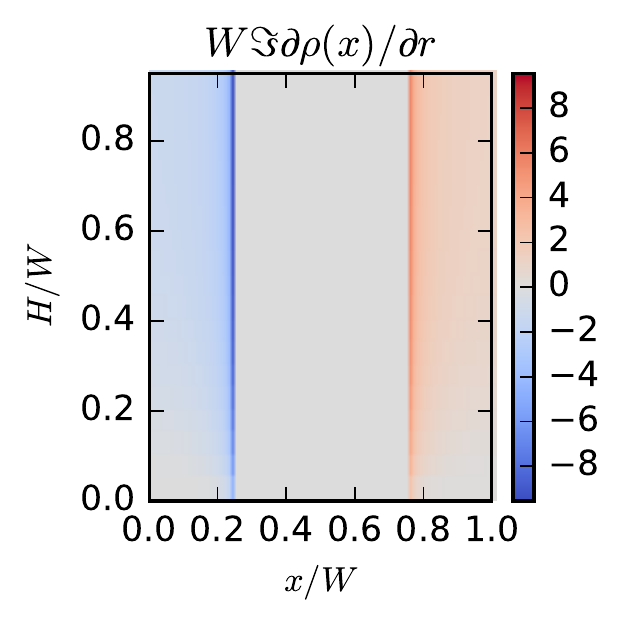}}
	\caption{
		Plots of the normalized current density
		$W \Im \partial\bm{\rho}(x)/ \partial\bm{r} = \pm W j_{f}(x)/(F n_{e} D_{\sigma} [c_{\sigma,f}^{B} - c_{\sigma,f}^{A}])$
		in Eq. \eqref{main:eqn:jf}
		for electrode bands of equal width $2w_{A} = 2w_{B} = \num{0,5}W$
		and different values of $H/W$.
	}
	\label{main:fig:im_dT7dr}
\end{figure}

The aspect ratio of the unit cell also affects the shape of the current density on the IDAE.

Fig. \ref{main:fig:im_dT7dr} shows the normalized current density,
at the bottom of the unit cell, for different aspect ratios $H/W$,
when the width of the electrode bands equals $2w_{A} = 2w_{B} = \num{0,5}W$.
As predicted by Eq. (\ref{main:eqn:dTdr}), 
the current density becomes singular at the edges of the electrode bands.
Also, it can be seen that the current density is highly non-uniform near the edges,
and becomes approximately uniform near the center of the electrode bands.

Fig. \ref{main:fig:im_dT7dr} also shows that
the magnitude of the plateau of current density varies
according to the aspect ratio $H/W$ of the unit cell.
For $H/W \gtrsim 1$ the plateau remains the highest
and the current density is approximately independent of $H/W$,
mainly because radial diffusion is completely developed
(see Fig. \ref{main:fig:xi:HW10}), as in the case of semi-infinite cells.
In the range of approximately $\num{0,3} \lesssim H/W \lesssim 1$,
the current density decreases slowly as the aspect ratio $H/W$ decreases,
due to truncation of the radial diffusion caused by the lower roof of the cell.
In the range of approximately $0 \leq H/W \lesssim \num{0,3}$,
the current density decreases quickly as $H/W$ decreases,
with plateaus approaching zero for very low $H/W$,
mainly because the vertical diffusion in the cell becomes almost nonexistent,
and only horizontal diffusion between consecutive bands takes place
(see Fig. \ref{main:fig:xi:HW03}).

\subsection{Current per band}

\begin{figure}
	\centering
	\includegraphics{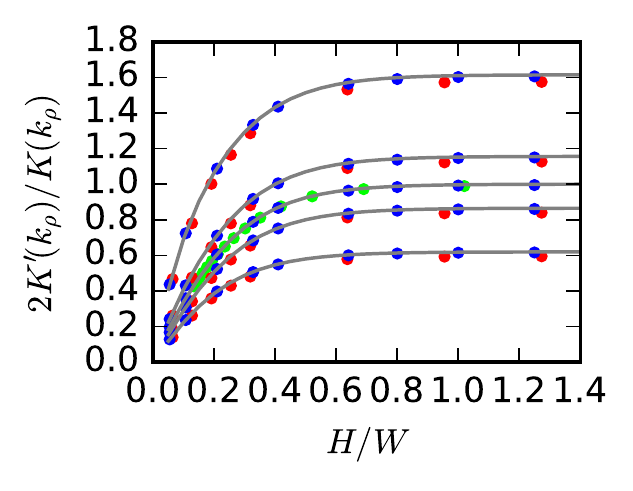}
	\caption{
		Comparison between simulated and theoretical normalized current per band
		$2 K'(k_{\rho})/K(k_{\rho}) = \pm (i_{f}^{E}/L)/(F n_{e} D_{\sigma} [c_{\sigma,f}^{E} - c_{\sigma,f}^{E'}])$
		for $2w_{A}/W = 2w_{B}/W \in \{ \numlist[list-final-separator={,}]{0,2; 0,4; 0,5; 0,6; 0,8} \}$
		(black lines in order from bottom to top).
		Black line: Theoretical expresion.
		Blue dots: Simulations in this work.
		Red dots: Simulation data from \cite[Fig. 7a]{GuajardoYevenes2013sep}.
		Green dots: Simulation data from \cite[Fig. 7a]{Strutwolf2005feb}.
		See \emph{Supplementary information} \S\ref{si:sec:exhaustive} for data.
	}
	\label{main:fig:K'Krho:simu}
\end{figure}

The normalized current in Eq. \eqref{main:eqn:if} was contrasted
against simulations of Eqs. \eqref{main:eqn:pde}
for different values of aspect ratio $H/W \in \mathopen]0, \num{1,25} \mathclose[$
and electrode widths $2w_{A}/W = 2w_{B}/W \in \mathopen] 0, 1 \mathclose[$,
see Fig. \ref{main:fig:K'Krho:simu}.
In 44 of 45 analyzed cases, the maximum absolute error was not greater than $\approx \num{0,0078}$,
thus confirming the correctness of Eq. \eqref{main:eqn:if}.
The theoretical result also agrees with simulation results
previously reported in the literature
\cite[Fig. 7a]{Strutwolf2005feb} \cite[Fig. 7a]{GuajardoYevenes2013sep},
which are also shown in Fig. \ref{main:fig:K'Krho:simu}.
For more details on the data see Tables \ref{si:tab:exhaustive:itau}, \ref{si:tab:exaustive:itau:Guajardo2013}
and \ref{si:tab:exaustive:itau:Strutwolf2005} in
\emph{Supplementary information} \S\ref{si:sec:exhaustive}.

\subsubsection{Effect of the counter electrode in the collection efficiency}

The expression in Eq. \eqref{main:eqn:if} shows that the currents at bands
of both arrays are equal in magnitude but opposite $i_{f}^{A} = -i_{f}^{B}$.
This suggests that 100\% collection efficiency is a necessary condition for this expression to hold,
which is immediately satisfied when the IDAE operates with internal counter electrode
\cite[start of p. 280]{Aoki1988dec} \cite[p. 7558 mid col. 2]{Rahimi2011}.

However, in the case of external counter electrode,
the collection efficiency may not reach 100\%
\cite[start of p. 280]{Aoki1988dec} \cite[p. 7558 mid col. 2]{Rahimi2011}.
In case the roof of the cell is low ($H/W \lesssim 1$),
the collection efficiency is close to 100\%
\cite[Fig. 5b]{Strutwolf2005feb} \cite[Fig. 5e]{Kanno2014}
and Eq. \eqref{main:eqn:if} approximately holds.
However, when the roof of the cell becomes taller ($H/W \gtrsim 1$),
the collection efficiency departs from 100\% and decreases
\cite[Fig. 5a]{Strutwolf2005feb} \cite[Fig. 5e]{Kanno2014}.
In this last case, a correction that takes into account the effect of
external counter is needed to accurately predict the current through the IDAE.
This kind of correction was done semi-empirically in \cite[Eq. (33)]{Aoki1988dec}
and later in \cite[Eqs. (13) and (20)]{Morf2006may},
both for the case of semi-infinite cells ($H \to +\infty$).

\subsubsection{Influence of the cell geometry}

\begin{figure}
	\centering
	\subcaptionbox{
		Case of bands of equal width.
		\label{main:fig:K'Krho:wAwB}
	}{\includegraphics{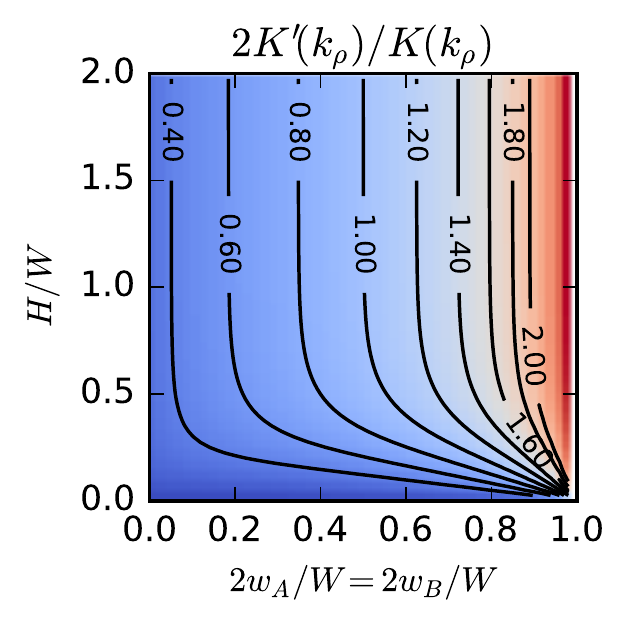}}
	\subcaptionbox{
		Case $H/W = \num{0.1}$.
		\label{main:fig:K'Krho:HW01}
	}{\includegraphics{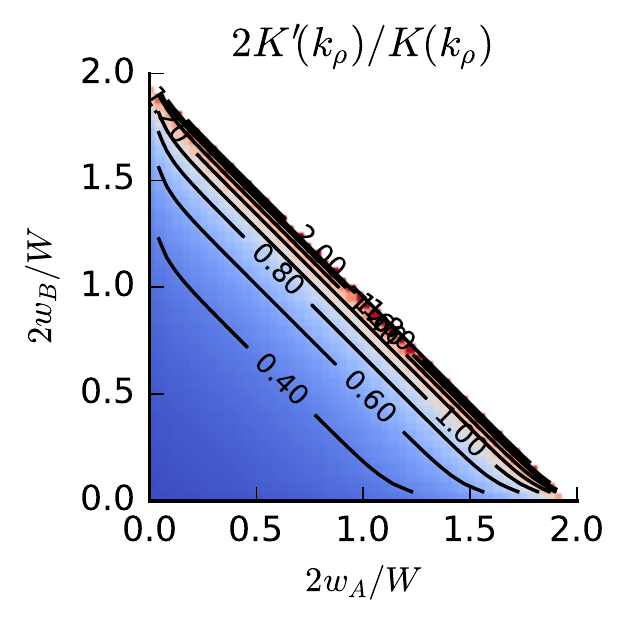}}
	
	\subcaptionbox{
		Case $H/W = \num{0.3}$.
		\label{main:fig:K'Krho:HW03}
	}{\includegraphics{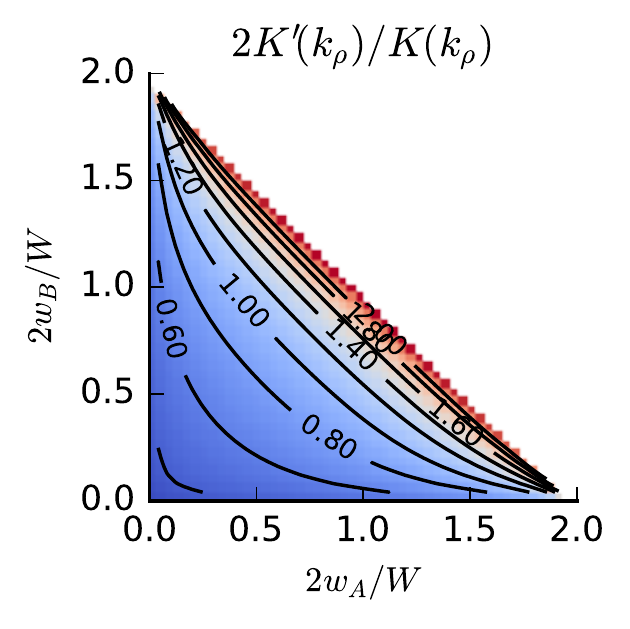}}
	\subcaptionbox{
		Case $H/W = 1$.
		\label{main:fig:K'Krho:HW10}
	}{\includegraphics{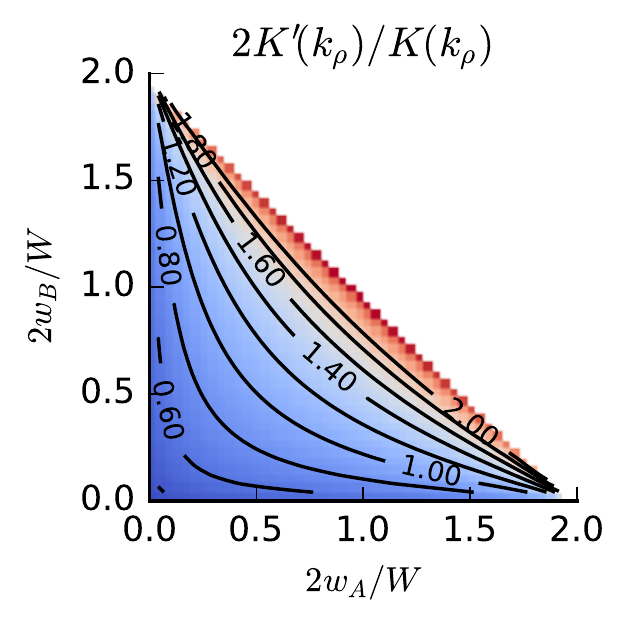}}
	\caption{
		Different slices and contour lines of the normalized current per band
		$2 K'(k_{\rho})/K(k_{\rho}) = \pm (i_{f}^{E}/L)/(F n_{e} D_{\sigma} [c_{\sigma,f}^{E} - c_{\sigma,f}^{E'}])$.
		The normalized current is a function of the relative dimensions of the cell: $2w_{A}/W$, $2w_{B}/W$ and $H/W$.
	}
	\label{main:fig:K'Krho:theo}
\end{figure}

According to Corollary \ref{main:cor:if},
the current in steady state through the IDAE depends on the ratio
$2K'(k_{\rho})/K(k_{\rho})$, which is a purely geometrical factor,
since the modulus $k_{\rho}$ depends only on
the relative dimensions of the unit cell $2w_{A}/W$, $2w_{B}/W$ and $H/W$,
as seen in Eqs. (\ref{main:eqn:kr-krho}) and (\ref{main:eqn:vAB}).
This implies that, for known electrochemical species in the cell
($D_{\sigma}$ and $n_{e}$ are known)
and for known potentials applied to the IDAE
(therefore $c_{\sigma,f}^{A}$ and $c_{\sigma,f}^{B}$ are also known),
the performance of the cell can be optimized just by adjusting
the widths of the electrode bands and the aspect ratio of the unit cell.

Visualization of the ratio $2K'(k_{\rho})/K(k_{\rho})$,
as a function of its three geometrical parameters ($2w_{A}/W$, $2w_{B}/W$ and $H/W$),
is helpful to find the right combination of cell dimensions that
allow an optimal performance.
However, this ratio is difficult to visualize,
since it corresponds to a three-dimensional scalar field.
Therefore, Fig. \ref{main:fig:K'Krho:theo} shows two-dimensional slices and
contour lines of the ratio $2K'(k_{\rho})/K(k_{\rho})$,
as a function of the relative dimensions of the cell.
These slices were chosen such that they provide enough information to understand
how the normalized current $2K'(k_{\rho})/K(k_{\rho})$ behaves
for different values of $2w_{A}/W$, $2w_{B}/W$ and $H/W$.

Figs. \ref{main:fig:K'Krho:simu} and \ref{main:fig:K'Krho:theo} show that
the current increases as both the width of the electrode bands and
the aspect ratio of the unit cell increase
(similarly occurs for the case of elevated electrodes in the simulations of
\cite[Fig. 5]{Goluch2009may} \cite[Table 1]{Heo2013} \cite[Table 2]{Heo2014jun}).
For aspect ratios $H/W \gtrsim 1$, the current becomes independent of $H/W$,
and only depends on the width of the electrode bands.
For $\num{0,3} \lesssim H/W \lesssim 1$ and fixed band widths,
the current decreases slowly as the aspect ratio $H/W$ decreases.
For $0 \leq H/W \lesssim \num{0,3}$ and fixed band widths,
the current decreases much faster as the ratio $H/W$ decreases.
This variation in the magnitude of the current, in relation to the aspect ratio $H/W$,
has the same explanation as the one given for the current density in the previous section,
and relates to the truncation of the radial diffusion by a low roof of the cell.

Also, Figs. \ref{main:fig:K'Krho:simu} and \ref{main:fig:K'Krho:wAwB}
confirm that there is a certain limit about $H/W \gtrsim 1$
where the cell is so high that diffusion is not affected by the roof of the cell,
and the current approaches that of semi-infinite diffusion,
as predicted earlier in \cite[\S3.2 and Fig. 7a]{Strutwolf2005feb}
\cite[\S3.3 and Fig. 7a]{GuajardoYevenes2013sep} \cite[p. 309 and Fig. 5]{Kanno2014}.
Similar results have been also found through simulations for
the case of elevated electrodes \cite[p.451 and Fig. 5b]{Goluch2009may}.

Finally, note that the amount of material used for fabricating the electrodes
can be also optimized for a given current.
This amount of material is proportional to $m = 2w_{A}/W + 2w_{B}/W$
and corresponds to an anti-diagonal line in the domain of Figs.
\ref{main:fig:K'Krho:HW01}, \subref{main:fig:K'Krho:HW03}
and \subref{main:fig:K'Krho:HW10}.
Therefore, the minimum amount of utilized material that
produces a desired current is obtained when $2w_{A} = 2w_{B}$.
This can be seen by selecting a desired iso-current line,
and intersecting it with the anti-diagonal $m = 2w_{A}/W + 2w_{B}/W$.
This produces two intersection points,
which are symmetric with respect to the diagonal line  $2w_{A}/W = 2w_{B}/W$.
The minimum of $m$ is obtained by moving the anti-diagonal towards the origin,
which makes the two intersection points on the iso-current line converge to a single point, thus obtaining $2w_{A} = 2w_{B}$.

\subsection{Voltammogram shape}

Consider the current\footnote{
	for $\pm$ or $\mp$, the upper sign corresponds to $\sigma=O$,
	and the lower sign, to $\sigma=R$.
} through one array of the IDAE,
consisting of $N_{E}$ bands of length $L$, see Eq. \eqref{main:eqn:if}
\begin{equation}
	\label{main:eqn:NEif}
	N_{E} i_{f}^{E}	= 
	\pm N_{E}L\; 2\frac{K'(k_{\rho})}{K(k_{\rho})}
	\cdot F n_{e} D_{\sigma} [c_{\sigma,f}^{E} - c_{\sigma,f}^{E'}]
\end{equation}
Here one can distinguish two clear factors with different roles: A geometrical
factor $N_{E}L\; 2K'(k_{\rho})/K(k_{\rho})$ and an electrochemical factor
$F n_{e} D_{\sigma} [c_{\sigma,f}^{E} - c_{\sigma,f}^{E'}]$.

\begin{figure}[t]
	\centering
	\subcaptionbox{
		\label{main:fig:cE-cE':extC}
		Case of external counter electrode.
		Solid lines: Result in Eq. \eqref{main:eqn:cE-cE':extC}.
	}{\includegraphics{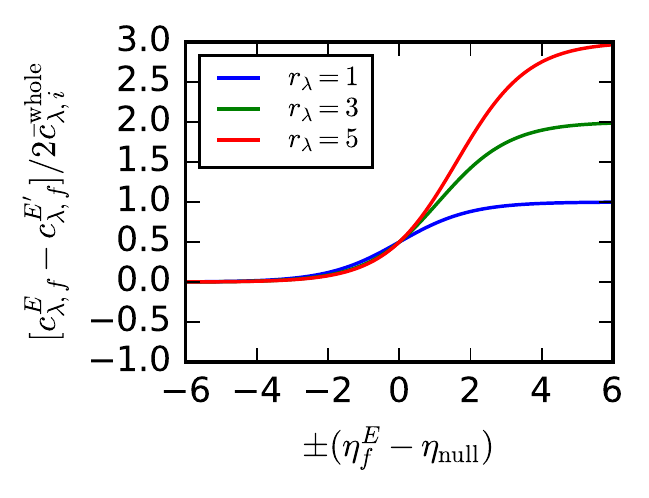}}
	\hfill
	\subcaptionbox{
		\label{main:fig:cE-cE':intC}
		Case of internal counter electrode.
		Solid lines: Result in Eq. \eqref{main:eqn:cE-cE':intC}.
		Dashed line: Saturation limit in Eq. \eqref{main:eqn:cE-cE':cotas}
		due to depletion.
	}[68mm]{\includegraphics{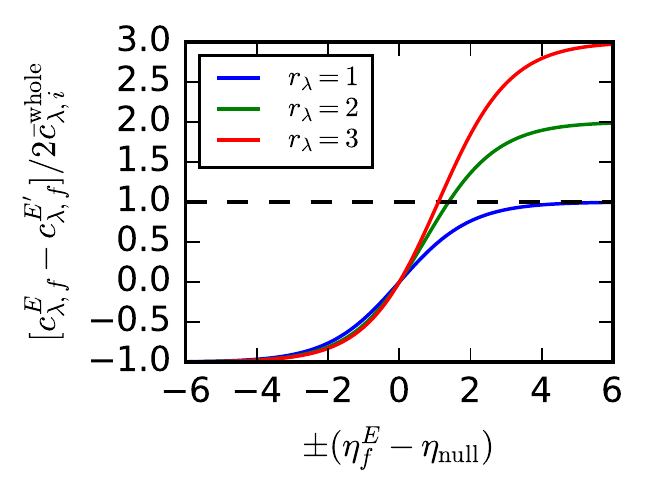}}
	\caption{
		Plot of the normalized difference of concentration
		$[c_{\lambda,f}^{E} - c_{\lambda,f}^{E'}]/2\bar{c}_{\lambda,i}^{\whole}$
		as a function of the normalized potential
		$\eta_{f}^{E} - \eta_{\mathrm{null}}
		= (n_{e} F/RT)(V_{f}^{E} - V_{\mathrm{null}})$
		for different ratios $r_{\lambda}
		= (D_{\lambda'} \bar{c}_{\lambda',i}^{\whole})
		/ (D_{\lambda} \bar{c}_{\lambda,i}^{\whole})$.
		Here $\lambda$ and $\lambda'$ correspond to the limiting species
		and its complementary species, such that
		$D_{\lambda} \bar{c}_{\lambda,i}^{\whole}
		= \min(D_{O} \bar{c}_{O,i}^{\whole}, D_{R} \bar{c}_{R,i}^{\whole})$.
	}
\end{figure}

The inherent shape of the voltammogram is solely due to the electrochemical factor,
which is proportional to the difference of concentrations
in Eqs \eqref{main:eqn:cE-cE':extC} and \eqref{main:eqn:voltammogram:intC}.
This inherent shape can be amplified or atenuated by the geometrical factor,
which is shown in Figs. \ref{main:fig:K'Krho:simu} and \ref{main:fig:K'Krho:theo},
and expressed approximately in Eqs. \eqref{main:eqn:K'K:Hinf} and \eqref{main:eqn:K'K:H0}.

Therefore, here we examine the shape of the voltammogram just by looking at
the behavior of the difference of concentrations in two cases:
With external and internal counter electrodes.

The first, and the most commonly found in the literature,
corresponds to the case where the counter electrode is external to the IDAE,
which means that each array of the IDAE is potentiostated individually
(commonly, the potential of one array is scanned,
while the complementary array is fixed to a sufficiently negative potential).

Fig. \ref{main:fig:cE-cE':extC} shows the plot of the normalized difference
of concentrations
when using an external counter electrode, see Eq. \eqref{main:eqn:cE-cE':extC}.
This normalized difference (and thus the current) is unipolar
(its is either always positive or always negative), and it increases with the
weighted sum of initial concentrations of both electrochemical species.
Moreover, the steady-state current reaches plateaus (limiting current)
that are proportional to this weighted (total, when $D_{O} = D_{R}$)
sum of concentrations  \cite[p. 7558 start of col. 2]{Rahimi2011},
when the applied potential at the scanning electrode exceeds $|\eta_{f}^{E} - \eta_{\nul}| \gtrsim 4$.

The second is the case where the counter electrode is internal to the IDAE,
which means that one of the arrays is potentiostated at will,
while the complementary array performs as counter electrode
(its potential is controlled automatically by the potentiostat).

\begin{figure}[t]
	\centering
	\subcaptionbox{
		\label{main:fig:cE-cE':intC:noR}
		Normalized difference of concentration
		$[c_{\lambda,f}^{E} - c_{\lambda,f}^{E'}]/2\bar{c}_{\lambda,i}^{\whole}$
		as a function of the normalized voltage
		$\eta_{f}^{E} - \eta_{f}^{E'}$.
		See Eq. \eqref{main:eqn:etaE-etaE':intC}.
	}{\includegraphics{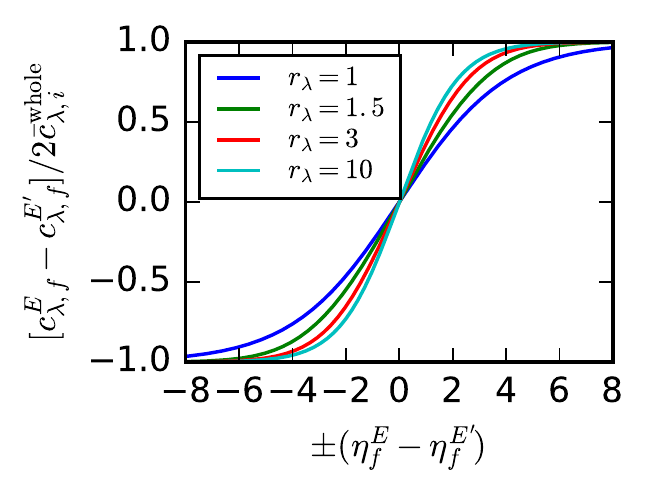}}
	\hfill
	\subcaptionbox{
		\label{main:fig:etaE=DetaE}
		Dependence of the normalized potentials $\eta_{f}^{E} - \eta_{\nul}$ (solid) and
		$\eta_{f}^{E'} - \eta_{\nul}$ (dashed) as a function of
		the normalized voltage $\eta_{f}^{E} - \eta_{f}^{E'}$.
		See Eqs. \eqref{main:eqn:voltammogram:intC}.
	}{\includegraphics{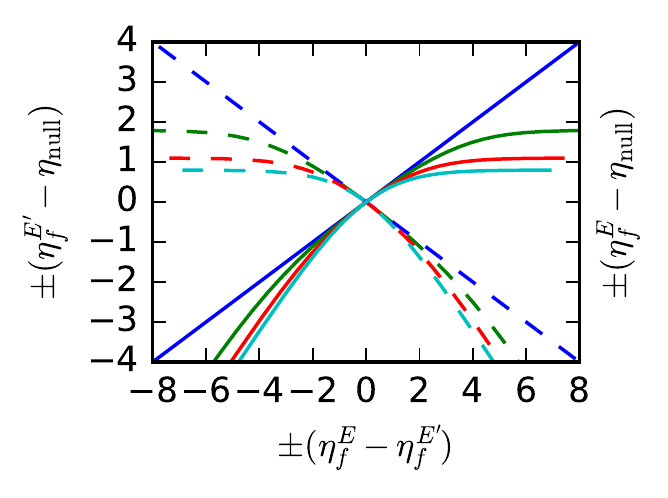}}
	\caption{
		Case when the internal counter electrode is used as reference electrode,
		analyzed for different ratios $r_{\lambda}
		= (D_{\lambda'} \bar{c}_{\lambda',i}^{\whole})
		/ (D_{\lambda} \bar{c}_{\lambda,i}^{\whole})$.
		Here $\lambda$ and $\lambda'$ are the determinant and its complementary species,
		such that $D_{\lambda} \bar{c}_{\lambda,i}^{\whole}
		= \min(D_{O} \bar{c}_{O,i}^{\whole}, D_{R} \bar{c}_{R,i}^{\whole})$,
		$\eta_{f}^{E} - \eta_{\mathrm{null}} = (F n_{e}/RT)(V_{f}^{E}
		- V_{\mathrm{null}})$ is the normalized working potential,
		$\eta_{f}^{E'} - \eta_{\mathrm{null}} = (F n_{e}/RT)(V_{f}^{E}
		- V_{\mathrm{null}})$ is the normalized counter potential,
		and $\eta_{f}^{E} - \eta_{f}^{E'}$ is the normalized voltage.
	}
\end{figure}

Fig. \ref{main:fig:cE-cE':intC} shows the plot of the normalized difference of concentrations
when using an internal counter electrode, see Eq. \eqref{main:eqn:cE-cE':intC}.
This normalized difference (and thus the current) is bipolar
(it is positive and negative in the same plot),
and presents plateaus when the applied potential is sufficiently high
$|\eta_{f}^{E} - \eta_{\nul}| \gtrsim 4$.
Ideally, the magnitude of the difference of concentrations
(steady-state current) should increase as the ratio
$r_{\lambda} = (D_{\lambda'} \bar{c}_{\lambda',i})/(D_{\lambda} \bar{c}_{\lambda,i})$
between the electrochemical species increases, but it is actually limited by
Eq. \eqref{main:eqn:cE-cE':cotas} (dashed line),
due to depletion of its limiting species $\lambda$
\cite[Fig. 2]{GuajardoYevenes2013sep}.
This causes the limiting current to be proportional to
the concentration of the limiting species, instead of the weighted
sum of concentrations of both species
\cite[p. 7558 mid col. 2]{Rahimi2011} \cite[Fig. 2]{Rahimi2010mar} \cite[Fig. 3.4]{Rahimi2009aug}.
This fact has been also discussed in \cite[\S2.3]{Morf2006may}
and \cite[\S2.3]{GuajardoYevenes2013sep}
and implies that simultaneous presence of both electrochemical species
is a necessary condition to obtain steady state currents.

\begin{figure}[t]
	\centering
	\subcaptionbox{
		\label{main:fig:fitting:extC}
		Case of external counter electrode.
		Dots: Experimental voltammogram of \SI{1}{\nano\mole\per\micro\litre}
		ferrocene (reduced sp.) with $V_{f}^{E'} - V^{R} = \SI{-0,15}{\volt}$
		from \cite[Fig. 7]{Aoki1988dec}\footnotemark for generator (green) and collector (blue).
		Lines: Curve fitting of the current at the generator using Eq. \eqref{main:eqn:model:extC}
		with $k_{0} = \SI{33,5}{\micro\ampere}$ and $k_{1} = \SI{-0,413}{\volt}$
		(green) and its reflection (blue).
	}{\includegraphics{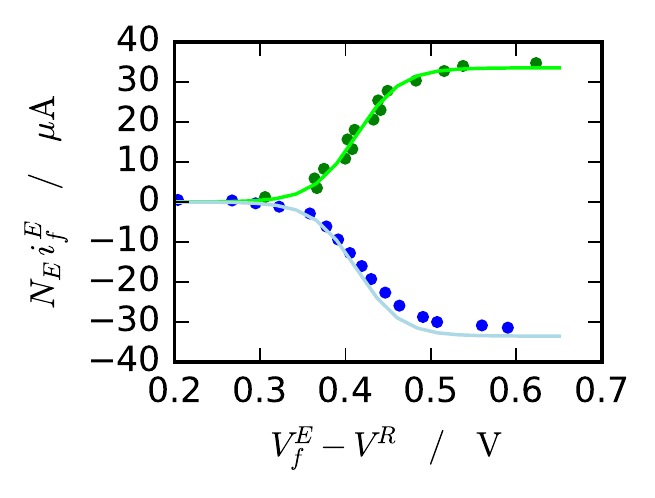}}.
	\subcaptionbox{
		\label{main:fig:fitting:intC}
		Case of internal counter electrode.
		Dots: Experimental voltammogram of \SI{0,20}{\nano\mole\per\micro\litre}
		of each ferrocyanide (reduced sp.) and ferricyanide (oxidized sp.)
		from \cite[Fig. 2]{Rahimi2011} or equivalently from
		\cite[Fig. 3.12]{Rahimi2009aug}.
		Line: Curve fitting using Eq. \eqref{main:eqn:model:intC}
		with $k_{0} = \SI{0,499}{\micro\ampere}$.
	}{\includegraphics{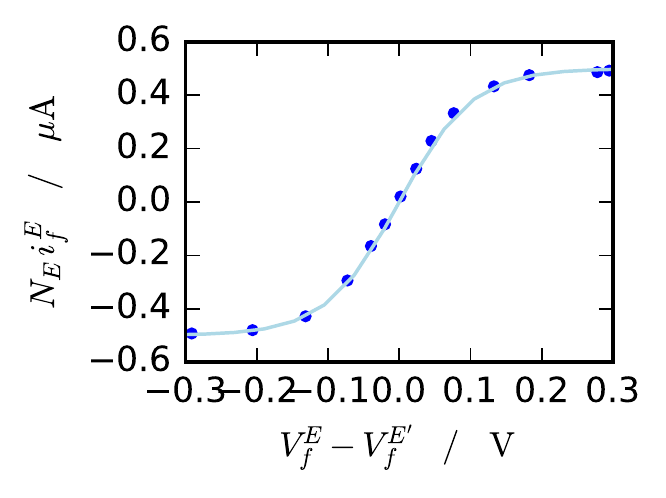}}
	\caption{
		Curve fitting of experimental voltammogram by using the proposed models
		for external and internal counter electrodes.
		$N_{E} i_{f}^{E}$ and $V_{f}^{E}$ corresponds to the total current
		and the potential applied to the working array,
		$V_{f}^{E'}$ is the potential at the counter array and
		$V^{R}$ is the potential at the reference electrode (saturated calomel electrode).
	}
	\label{main:fig:fitting}
\end{figure}

Note that, in case the internal counter electrode serves as reference electrode,
only a voltage (difference of potentials) can be applied to the cell.
Fig. \ref{main:fig:cE-cE':intC:noR} shows that
a sigmoidal shape is maintained under this condition,
of which its plateaus (limiting current) are also proportional to
the concentration of the limiting species, see Eq. \eqref{main:eqn:etaE-etaE':intC}.
Fig \ref{main:fig:etaE=DetaE} shows how the applied voltage is
distributed among the potentials of the working and counter electrodes.
When the voltage applied to the cell increases in one direction,
the potential at one electrode increases indefinitely, 
whereas the potential at the other electrode saturates in the opposite direction.
This is due to depletion of the limiting species $\lambda$
at the electrode with greatest potential (in absolute value).
This phenomenon has been observed experimentally in
\cite[Fig. 3]{Rahimi2010mar} \cite[Fig. 3.6]{Rahimi2009aug},
where figures similar to Fig. \ref{main:fig:etaE=DetaE} were obtained.
See \emph{Supplementay information} \S\ref{si:sec:voltammogram:intC:noR}
for more details on how Figs. \ref{main:fig:cE-cE':intC:noR} and
\ref{main:fig:etaE=DetaE} were generated.

\footnotetext{
	\cite[Fig. 7]{Aoki1988dec} actually used the set of electrodes (E)
	instead of (D).
	This can be seen by dividing the limiting current from
	\cite[Fig. 7]{Aoki1988dec} by the number of bands of (E)
	$\approx \SI{34}{\micro\ampere}/50 =\SI{0,68}{\micro\ampere}$,
	which corresponds to set (E) in \cite[Fig. 8]{Aoki1988dec}.
}

Finally, we show that the models in Eqs. \eqref{main:eqn:cE-cE':extC} and
\eqref{main:eqn:voltammogram:intC} are also applicable to fit experimental data.
Combining Eq. \eqref{main:eqn:NEif} with Eq. \eqref{main:eqn:cE-cE':extC},
for the case of external counter electrode, leads to\footnote{
	the upper sign corresponds to the case where
	the concentration of oxidized species at the complementary array is zero (at very negative potential), and the lower, to the case
	where the concentration of reduced species is zero
	(at very positive potential), see Corollary \ref{main:cor:if}.
}
\begin{equation}
	\label{main:eqn:model:extC}
	N_{E} i_{f}^{E}
	= \pm \frac{
		k_{0}
	}{
		\displaystyle
		1 + \exp\left(
			\mp \frac{F n_{e}}{R T} (V_{f}^{E} - V^{R} + k_{1})
		\right)
	}
\end{equation}
where the voltammogram is mirrored horizontally and vertically when changing
the polatiry at the complementary electrode \cite[\S3.3]{Wahl2018jul}.
Fig. \ref{main:fig:fitting:extC} shows that the model in Eq. \eqref{main:eqn:model:extC}
correctly fits the experimental data for the generator in \cite[Fig. 7]{Aoki1988dec},
whereas the data for the collector is slightly overestimated,
thus showing that the collection efficiency is near 100\%.

For the case of internal counter electrode with bands of equal width,
we combine Eqs. \eqref{main:eqn:NEif} with \eqref{main:eqn:etaE-etaE':intC}
when $D_{O} \bar{c}_{O,i}^{\whole} = D_{R} \bar{c}_{R,i}^{\whole}$, obtaining
\begin{equation}
	\label{main:eqn:model:intC}
	N_{E} i_{f}^{E}
	= k_{0}\tanh\left( \frac{F n_{e}}{4RT} (V_{f}^{E} - V_{f}^{E'}) \right)
\end{equation}
Fig. \ref{main:fig:fitting:intC} shows that this model also fits correctly
the experimental data in \cite[Fig. 2]{Rahimi2011} \cite[Fig. 3.12]{Rahimi2009aug},
both at the linear and the limiting current regions.
For details on the experimental data used in both curve fittings
see \emph{Supplementary information} \S\ref{si:sec:fitting}.

\begin{observacion}
	Note that the models in Eqs. \eqref{main:eqn:cE-cE':extC}
	and \eqref{main:eqn:voltammogram:intC} for the difference of concentrations
	(shape of the steady-state voltammogram), in particular Eqs.
	\eqref{main:eqn:model:extC} and \eqref{main:eqn:model:intC},
	are valid for any electrochemical cell satisfying the following conditions:
	(i) The cell has reversible electrode reactions (Nernstian boundary conditions).
	(ii) The current is proportional to the difference between concentrations at both electrodes.
	(iii) The weighted sum of concentrations at each electrode satisfies
	a relation similar to Eq. \eqref{main:eqn:cfE:total}.
	(iv) In case of internal counter electrode, the concentrations at both
	electrodes satisfy a relation similar to Eq. \eqref{main:eqn:cfE:average}.
\end{observacion}
\noindent This is because the previous conditions are
properties that depend only on the boundary conditions at the electrodes,
which are independent of the whole domain of the cell.

\subsection{Approximations for shallow and tall cells}
\label{main:approx}

The approximations of the normalized current $2K'(k_{\rho})/K(k_{\rho})$
in Eqs. \eqref{main:eqn:K'K:Hinf} for tall electrochemical cells
are already available in the literature,
and were obtained first by Aoki and colleagues,
for the case of large electrodes \cite[Eq. (32) and the equation above it]{Aoki1988dec},
and later by Morf and colleagues, for the case of small electrodes \cite[Eqs. (6) and (7)]{Morf2006may}.
Therefore, these results serve as a validation of the exact results
obtained in Corollary \ref{main:cor:if}.

\begin{figure}[t]
	\centering
	\subcaptionbox{Approximation in Eq. (\ref{main:eqn:K'K:Hinf:wEinf}) for large electrodes.}{
		\includegraphics{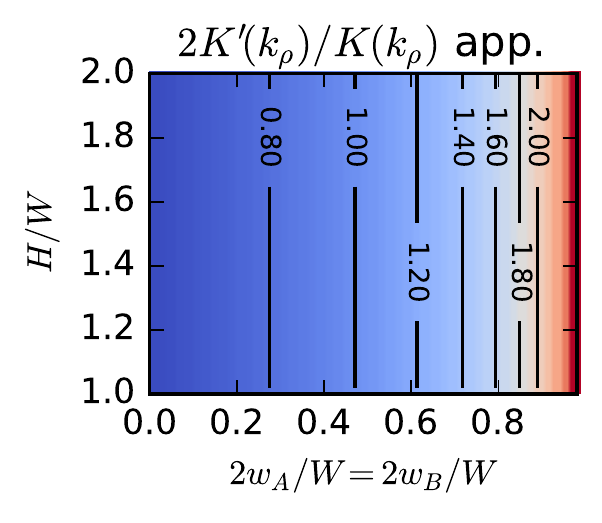}
		\includegraphics{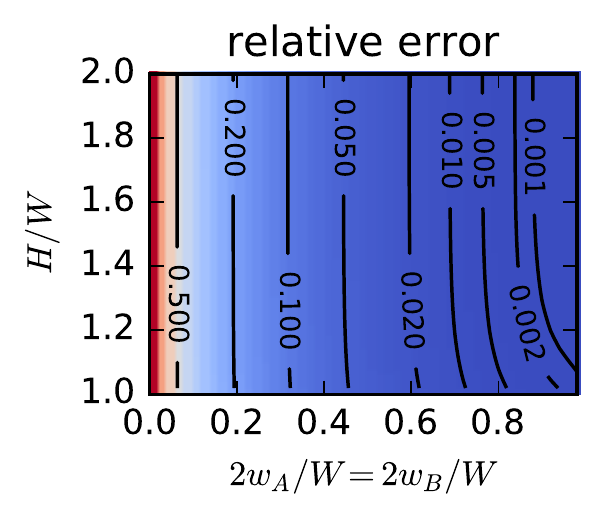}
	} \\
	\subcaptionbox{Approximation in Eq. (\ref{main:eqn:K'K:Hinf:wE0}) for small electrodes.}{
		\includegraphics{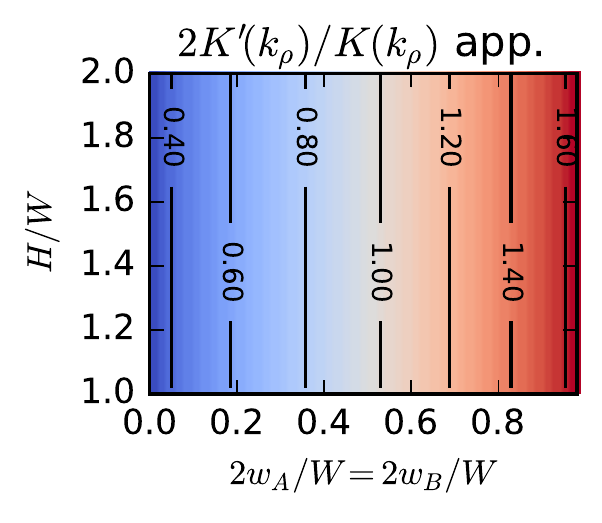}
		\includegraphics{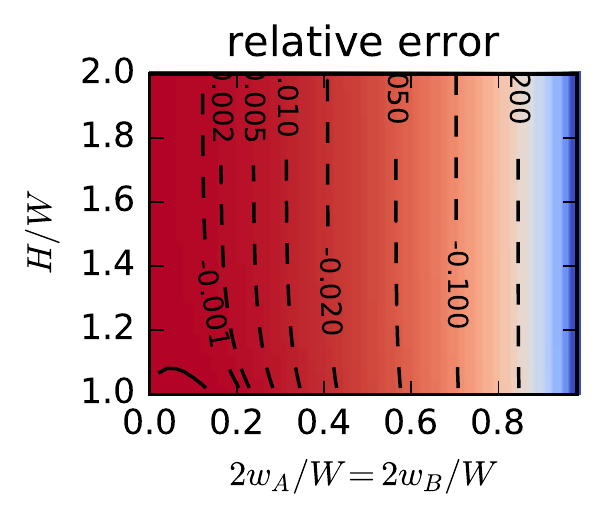}
	}
	\caption{Approximation of the normalized current $2 K'(k_{\rho})/K(k_{\rho}) = \pm (i_{f}^{E}/N_{E} L)/(F n_{e} D_{\sigma} [c_{\sigma,f}^{E} - c_{\sigma,f}^{E'}])$ for tall electrochemical cells ($H/W > 1$).}
	\label{main:fig:K'Krho:HW_tall}
\end{figure}

Fig. \ref{main:fig:K'Krho:HW_tall} shows the approximations
for tall cells $H/W > 1$ as well as their relative errors.
Here the current becomes independent of $H/W$ and
depends only on the relative size of the electrodes.
The approximation for large electrodes becomes accurate
(error less than \SI{+-5}{\percent}) for bands of relative width $> \num{0,46}$.
In the case of small electrodes, the approximation has an error
less than \SI{+-5}{\percent} for bands of relative width $< \num{0,56}$.
The regions of approximation for both sizes of electrodes overlap, thus
covering all cases of $H/W > 1$ with a relative error less than
\SI{+-5}{\percent} in the current.
See Table \ref{main:tab:if_approx} for a summary and
\emph{Supplementary information} \S\ref{si:sec:if_approx} for details on the errors.

\begin{figure}[t]
	\centering
	\subcaptionbox{Approximation in Eq. (\ref{main:eqn:K'K:H0:wEinf}) for large electrodes.}{
		\includegraphics{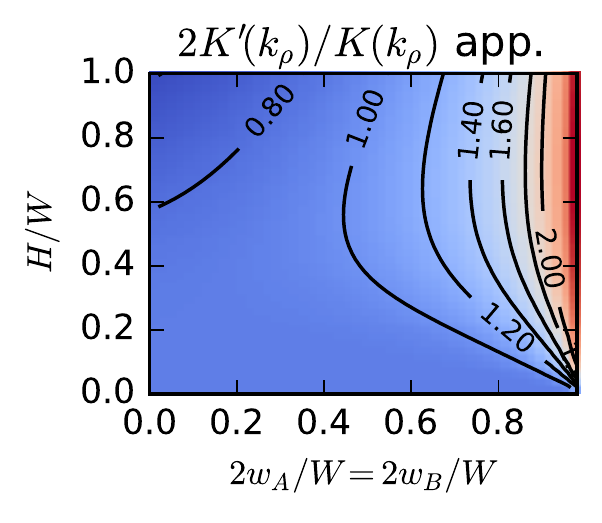}
		\includegraphics{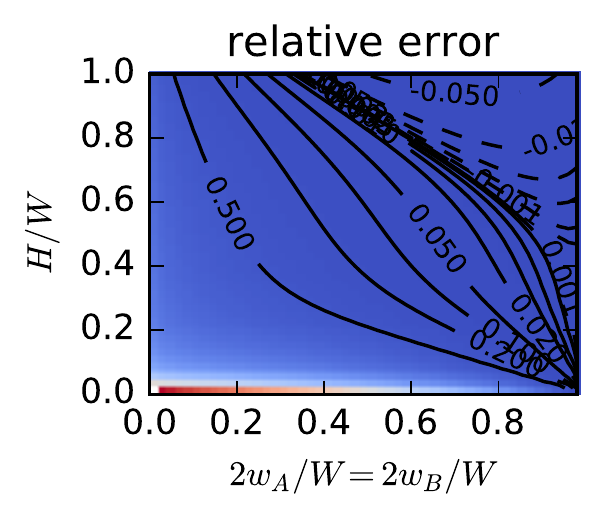}
	} \\
	\subcaptionbox{Approximation in Eq. (\ref{main:eqn:K'K:H0:wE0}) for small electrodes.}{
		\includegraphics{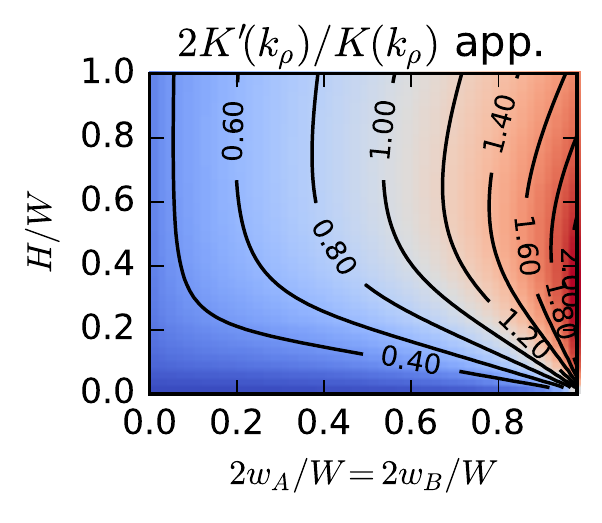}
		\includegraphics{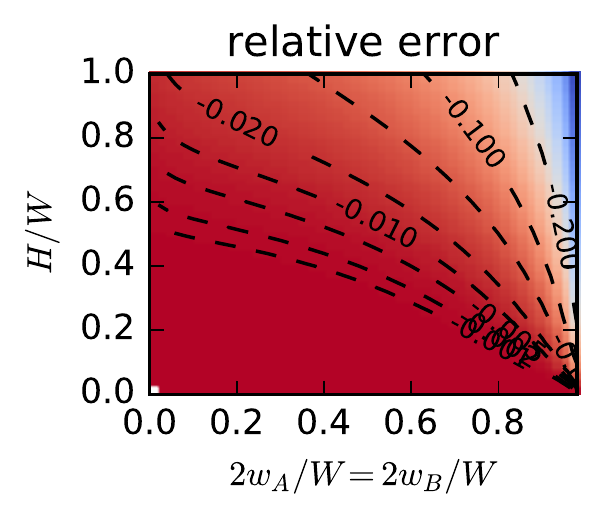}
	}
	\caption{Approximation of the normalized current $2 K'(k_{\rho})/K(k_{\rho}) = \pm (i_{f}^{E}/N_{E} L)/(F n_{e} D_{\sigma} [c_{\sigma,f}^{E} - c_{\sigma,f}^{E'}])$ for shallow electrochemical cells ($H/W \leq 1$).}
	\label{main:fig:K'Krho:HW_shallow}
\end{figure}

\begin{table}
	\centering
	\begin{tabular}{ccccc}
		\hline
		Approx. & Cell & Electrodes & Domain \\
		\hline
		Eq. (\ref{main:eqn:K'K:Hinf:wEinf}) & Tall    & Large &
		$2w_{E}/W > \num{0,46}$ \\
		Eq. (\ref{main:eqn:K'K:Hinf:wE0}) & Tall    & Small & 
		$2w_{E}/W < \num{0,56}$ \\
		Eq. (\ref{main:eqn:K'K:H0:wEinf})   & Shallow & Large & 
		$2w_{E} + (1 - \num{0,28})H \gtrsim W$ \\
		Eq. (\ref{main:eqn:K'K:H0:wE0})     & Shallow & Small &
		$2w_{E} + (1 - \num{0,36})H \lesssim W$ \\
		\hline
	\end{tabular}
	\caption{
		Regions where the approximations hold with a relative error
		less than \SI{+-5}{\percent} for tall ($H/W > 1$) and
		shallow ($H/W \leq 1$) cells.
	}
	\label{main:tab:if_approx}
\end{table}

Approximations of the normalized current $2K'(k_{\rho})/K(k_{\rho})$
in Eqs. \eqref{main:eqn:K'K:H0} for shallow electrochemical cells
are results that have not been published before.
Fig. \ref{main:fig:K'Krho:HW_shallow} shows the approximations
for shallow cells $H/W \leq 1$ as well as their relative errors.
The approximation for large electrodes becomes accurate
(error less than \SI{+-5}{\percent}) for combinations of $(2w_{E}/W, H/W)$
approximately at the right of the line given by the points
(1, 0) and (\num{0,28}, 1),
where the line that passes through the points $(1, 0)$ and $(p, 1)$ is given by
\begin{equation}
	\frac{2 w_{E}}{W} + (1 - p) \frac{H}{W} = 1
\end{equation}
In the case of small electrodes, the approximation has an error less than \SI{+-5}{\percent}
for combinations of $(2w_{E}/W, H/W)$ approximately at the left of
the line given by the points (1, 0) and (\num{0,36}, 1).
The regions of approximation for both sizes of electrodes overlap,
thus covering all cases of $H/W \leq 1$
with a relative error less than \SI{+-5}{\percent} in the current.
See Table \ref{main:tab:if_approx} for a summary and
\emph{Supplementary information} \S\ref{si:sec:if_approx} for details on the errors.

%% file: text-conclusion.tex

\section{Conclusions}

Thanks to Jacobian elliptic functions,
it is possible to transform the unit cell from the IDAE domain
into a parallel-plates domain.
In this last domain, the solution of the diffusion equation in steady state is simple,
and corresponds to a linear interpolation of the concentrations on both plates. This solution is transformed back into the IDAE domain,
leading to an analytical result for the concentration profile
that depends on elliptic functions and integrals.
Both, current density and current were derived from this concentration profile.

The results for the concentration profile, current density and current
depend geometrically on the relative dimensions of the cell, not on absolute dimensions.
Their behavior approaches that of an IDAE in a semi-infinite cell
as it becomes taller (approximately when the cell is taller than
the separation between centers of consecutive bands).

The shape obtained for the voltammogram is sigmoidal
and can be unipolar (it has either always positive or always negative currents)
or bipolar (presenting positive and negative currents)
depending on whether the IDAE is bipotentiostated using an external counter electrode,
or potentiostated using one of its arrays as internal counter electrode.
In case of using an external counter electrode,
the plateaus of current (limiting current) are proportional to
the weighted sum of initial concentrations
(total initial concentration when the diffusion coefficients are equal). Whereas, in case of using an internal counter electrode,
the plateaus of current (limiting current) are proportional to the limiting species,
which will be the species of least concentration if the diffusion coefficients are equal.

Approximations for the exact results were found.
Trigonometric and hyperbolic functions were used to approximate the cases of
tall and shallow cells respectively.
The approximations are accurate with a relative error smaller than $\pm 5\%$
with respect to the exact values.
Finally, when the approximations are used in combination,
they cover all possibilities of interest for IDAE in confined cells.

%% file: text-acknowledgements.tex

\section*{Acknowledgements}

The authors deeply appreciate the aid and comments of Dr. Mithran Somasundrum,
which helped to improve the quality of this manuscript.
The authors would like to thank the financial support provided by
\emph{King Mongkut’s University of Technology Thonburi} through the
\emph{KMUTT 55th Anniversary Commemorative Fund}, and the
\emph{Petchra Pra Jom Klao Ph.~D. scholarship} (Grant No. 28/2558)
for sponsoring {CFGY}.
Finally, the authors acknowledge 
the \emph{Higher Education Research	Promotion}
and \emph{National Research University Project of Thailand},
\emph{Office of the Higher Education Commission},
Ministry of Education, Thailand.

%% file: text-additional_proofs.tex

\section{Additional proofs}

\subsection{Transformation of the unit cell domain}
\label{si:proof:T}

\subsubsection{From \texorpdfstring{$\bm{r}$}{r} to \texorpdfstring{$\bm{v}$}{v} domain}

Considering the special values of the function $\arcsn()$
\citex{Table \dlmf[T]{22.5.}{1}}{dlmf}
\begin{subequations}
	\label{elipticas:eqn:arcsn:values}
	\begin{align}
		0 &= \arcsn(0,k) \\
		\label{main:eqn:K}
		\pm K(k) &= \arcsn(\pm 1,k) \\
		\pm K(k) + \bm{i} K'(k) &= \arcsn(\pm 1/k,k) \\
		\bm{i} K'(k) &= \arcsn(\bm{\infty},k)
	\end{align}
\end{subequations}
one can construct a function $(T_{r}^{v})^{-1}$
that maps the upper half-plane of $\bm{v}$ into the IDAE domain $\bm{r}$,
as shown in Fig. \ref{si:fig:T}.
The scale and translation of the function $\arcsn()$ must be chosen
such that $\bm{v}_{a} = -1$ is mapped to $\bm{r}_{a} = 0$ 
and $\bm{v}_{b} = 1$ is mapped to $\bm{r}_{b} = W$,
which leads to
\begin{subequations}
	\begin{gather}
		\bm{r} = (T_{r}^{v})^{-1}(\bm{v}) 
		= \frac{W}{2}
		\left[
			1 + \frac{1}{K(k_{r})}\arcsn(\bm{v},k_{r})
		\right]
		\\
		\label{main:eqn:r-v:alt}
		\bm{v} = T_{r}^{v}(\bm{r}) 
		= \sn\!\left(K(k_{r})\frac{2\bm{r}}{W} - K(k_{r}), k_{r}\right)
		= -\cd\!\left(K(k_{r})\frac{2\bm{r}}{W}, k_{r}\right)
	\end{gather}
\end{subequations}
which is corresponds to Eq. \eqref{main:eqn:r-v}
due to the quarter- and half-period properties
\citex{Table \dlmf[T]{22.4.}{3}}{dlmf}
\begin{equation}
	\label{main:eqn:sn-cd}
	\sn(\bm{u} - K(k),k) = \sn(\bm{u} + K(k) - 2K(k),k) = -\cd(\bm{u},k)
\end{equation}

The appropriate modulus $k_{r}$ can be obtained by forcing $\bm{v}_{m} = -1/k_{r}$
to be mapped to the upper-left corner $\bm{r}_{m} = \bm{i}H$
\begin{equation}
	\label{main:eqn:K1Kr}
	\bm{r}_{m} = (T_{r}^{v})^{-1}(\bm{v}_{m})
	\Leftrightarrow
	\bm{i}H = \bm{i} \frac{W}{2} \frac{K'(k_{r})}{K(k_{r})}
	\Leftrightarrow
	\frac{K'(k_{r})}{K(k_{r})} = \frac{2H}{W}
\end{equation}
which leads to Eq. \eqref{main:eqn:kr}, by applying the inverse nome function
Eq. \eqref{elipticas:eqn:nomo} \citex{Eqs. (\dlmf[E]{19.2.}{9}) and (\dlmf[E]{22.2.}{1})}{dlmf}.
Therefore, the function $T_{r}^{v}$ in Eq. (\ref{main:eqn:r-v:alt})
transforms the IDAE domain $\bm{r}$ into the upper half-plane $\bm{v}$
as shown in Fig. \ref{si:fig:T}.

Finally, the point $\bm{v}_{\alpha}$, in Eq. \eqref{main:eqn:vAB:A},
is obtained directly by evaluating Eq. (\ref{main:eqn:r-v:alt}).
Whereas the point $\bm{v}_{\beta}$, in Eq. \eqref{main:eqn:vAB:B},
is obtained by evaluating Eq. \eqref{main:eqn:r-v:alt},
and by combining the fact that $\sn()$ is an odd function and Eq. \eqref{main:eqn:sn-cd}.

\begin{figure*}
	\centering
	\includegraphics{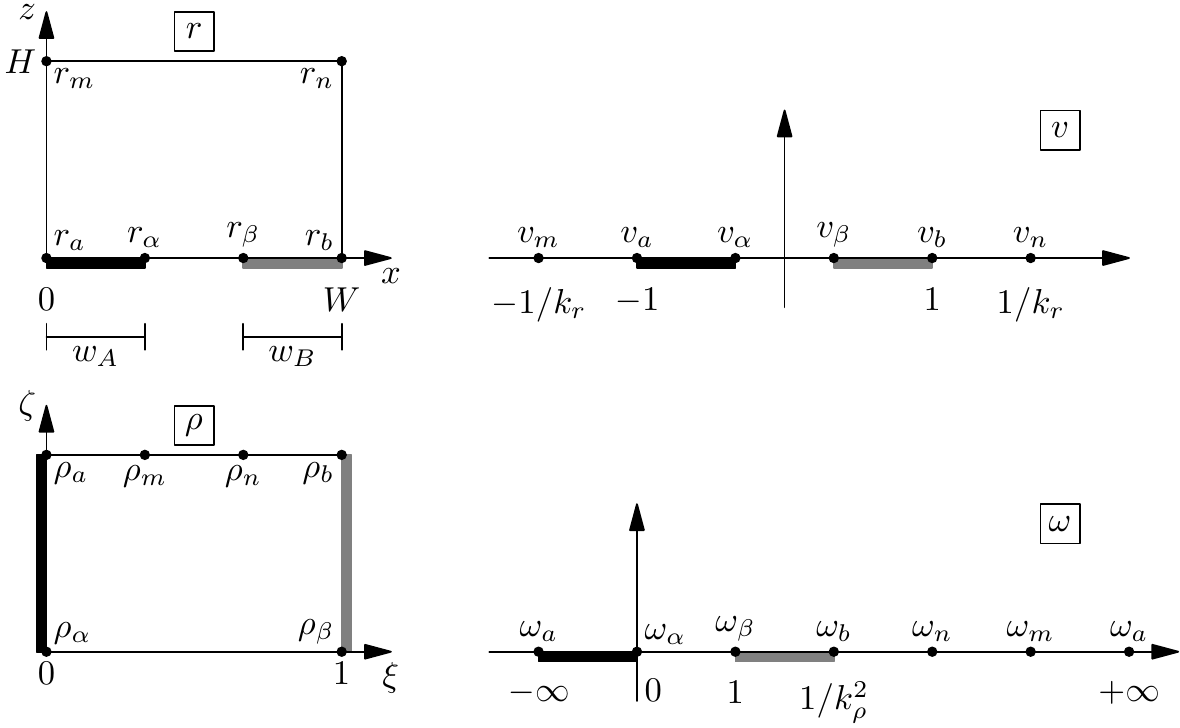}
	\caption{
		Complex transformation of the IDAE domain $\bm{r}=(x,z)$
		into the conformal parallel-plates domain $\bm{\rho}=(\xi,\zeta)$,
		by using the auxiliary complex domains $\bm{v}$ and $\bm{\omega}$.
	}
	\label{si:fig:T}
\end{figure*}

\subsubsection{From \texorpdfstring{$\bm{v}$}{v} to \texorpdfstring{$\bm{\omega}$}{ω} domain}


The function $T_{v}^{\omega}$ in Eq. \eqref{main:eqn:v-omega},
which is also shown below for convenience,
\begin{equation}
	\label{main:eqn:v-omega:alt}
	\bm{\omega} = T_{v}^{\omega}(\bm{v}) =
	\frac{(\bm{v}-\bm{v}_{\alpha})}{(\bm{v}-\bm{v}_{a})}
	\frac{(\bm{v}_{\beta}-\bm{v}_{a})}{(\bm{v}_{\beta}-\bm{v}_{\alpha})}
\end{equation}
corresponds to a Möbius function and it is constructed such that:
(i) it maps the upper half-plane of $\bm{v}$ into the upper half-plane of $\bm{\omega}$ and
(ii) it maps the half band of $A$ to the negative real axis of $\bm{\omega}$
and the half band of $B$ to the real interval $\bm{\omega} \in [1,\bm{\omega}_{b}]$,
see Fig. \ref{si:fig:T}.

Condition (ii) is immediately achieved,
since Eq. (\ref{main:eqn:v-omega:alt}) maps $\bm{v}_{a}$, $\bm{v}_{\alpha}$ and $\bm{v}_{\beta}$
into $\bm{\omega}_{a}=\bm{\infty}$, $\bm{\omega}_{\alpha}=0$ and $\bm{\omega}_{\beta}=1$.
Condition (i) can be analyzed by rewritting Eq. \eqref{main:eqn:v-omega:alt}
as a composition of translations, rotations and scalings, and inversion
\begin{equation}
	\bm{\omega} = T_{v}^{\omega}(\bm{v})
	= \left[
	1 + \frac{\bm{v}_{a} - \bm{v}_{\alpha}}{\bm{v} - \bm{v}_{a}}
	\right]
	\underbrace{\left[
		\frac{\bm{v}_{\beta}-\bm{v}_{a}}{\bm{v}_{\beta}-\bm{v}_{\alpha}}
		\right]}_{p}
\end{equation}
The fact that $\bm{v}_{a} \in \mathds{R}$, $\bm{v}_{a} - \bm{v}_{\alpha} < 0$ and $p > 0$
ensures that the upper half-plane of $\bm{v}$
is mapped to the upper half-plane of $\bm{\omega}$.
See \cite[Eq. (5.7.3)]{Ablowitz2003apr},
\cite[\S V.2 Eqs. (6), (7) and (10)]{Nehari1952}
or \cite[Examples 5.3 and 5.4]{Olver2018mar}
for more details on decomposition and mapping of Möbius functions.

\subsubsection{From \texorpdfstring{$\bm{\omega}$}{ω} to \texorpdfstring{$\bm{\rho}$}{ρ} domain}

The function $T_{\omega}^{\rho}$ in Eq. \eqref{main:eqn:omega-rho},
which is written below for convenience,
\begin{equation}
	\label{main:eqn:omega-rho:alt}
	\bm{\rho} = T_{\omega}^{\rho}(\bm{\omega})
	= \frac{1}{K(k_{\rho})} \arcsn(\sqrt{\bm{\omega}}, k_{\rho})
\end{equation}
is in charge of mapping the upper half-plane of $\bm{\omega}$
into the parallel-plates domain $\bm{\rho}$, see Fig. \ref{si:fig:T}.
This is achieved in two stages:
(i) The upper half-plane of $\bm{\omega}$ is mapped to the first quadrant of $\sqrt{\bm{\omega}}$,
such that the half band of $A$ is mapped to the positive imaginary axis of $\sqrt{\bm{\omega}}$
and the half band of $B$ is mapped to the real interval $\sqrt{\bm{\omega}} \in [1,\sqrt{\bm{\omega}_{b}}]$.
(ii) The first quadrant of $\sqrt{\bm{\omega}}$
is mapped to the parallel-plates domain $\bm{\rho}$
by using the special values of $\arcsn()$ 
in Eqs. \eqref{elipticas:eqn:arcsn:values} or \citex{Table \dlmf[T]{22.5.}{1}}{dlmf}.

The scaling and translation of the function $\arcsn()$ are chosen such that
$\sqrt{\bm{\omega}_{\alpha}}=0$ is mapped to $\bm{\rho}_{\alpha}=0$
and $\sqrt{\bm{\omega}_{\beta}}=1$ is mapped to $\bm{\rho}_{\beta}=1$,
which leads to Eq. \eqref{main:eqn:omega-rho:alt}.
The appropriate modulus $k_{\rho}$ is obtained by choosing $\sqrt{\bm{\omega}_{b}}=1/k_{\rho}$,
such that $\sqrt{\bm{\omega}_{b}}$ be mapped to the upper right corner $\bm{\rho}_{b}$ of the parallel-plates domain,
leading to Eq. \eqref{main:eqn:krho}.

Finally, the point $\bm{\rho}_{a}$, in Eqs. \eqref{main:eqn:rhoAB},
can be obtained directly by evaluating $\sqrt{\bm{\omega}_{a}} = \infty$
in Eq. \eqref{main:eqn:omega-rho:alt}, by using the properties in \eqref{elipticas:eqn:arcsn:values}.

\subsubsection{Conformality of the transformation}

The conformality of $T_{r}^{\rho}$ comes from the fact that
$T_{r}^{\rho}$ has non-zero complex derivative in the interior of the IDAE domain.
This can be seen from Eq. \eqref{main:eqn:dTdr} or \eqref{si:eqn:dTdr},
where the zeros of $\nd(K(k_{r}) 2\bm{r}/W, k_{r})$
\begin{equation}
	\bm{r} = mW + (2n+1) \bm{i} H\quad m, n \in \mathds{Z}
\end{equation}
lay only on the top vertices of the IDAE domain.
These zeros are obtained by using \citex{Table \dlmf[T]{22.4.}{2}}{dlmf} and Eq. \eqref{main:eqn:K1Kr}.

\subsection{Derivative of the domain transformation}
\label{si:proof:dTdr}

This proof concerns about obtaining the complex derivative of
the domain transformation $T_{r}^{\rho}$ in Eqs. \eqref{main:eqn:T},
which is written below for convenience
\begin{subequations}
	\label{si:eqn:T}
	\begin{align}
		\bm{\rho} = T_{r}^{\rho}(\bm{r})
		&= T_{\omega}^{\rho} \circ T_{v}^{\omega} \circ T_{r}^{v}(\bm{r})
		\\
		\bm{\rho} = T_{\omega}^{\rho}(\bm{\omega})
		&= \frac{1}{K(k_{\rho})} \arcsn(\sqrt{\bm{\omega}}, k_{\rho})
		\\
		\bm{\omega} = T_{v}^{\omega}(\bm{v})
		&= \frac{(\bm{v}-\bm{v}_{\alpha})}{(\bm{v}-\bm{v}_{a})}
		\frac{(\bm{v}_{\beta}-\bm{v}_{a})}{(\bm{v}_{\beta}-\bm{v}_{\alpha})}
		\\
		\bm{v} = T_{r}^{v}(\bm{r})
		&= -\cd\!\left(K(k_{r}) \frac{2\bm{r}}{W}, k_{r}\right)
	\end{align}
\end{subequations}

First, the complex derivative of each component of $T_{r}^{\rho}$,
in Eq. \eqref{si:eqn:T}, is taken
\begin{subequations}
	\begin{align}
		\parderiv{\bm{\rho}}{\bm{\omega}} &=
		\frac{1}{2 K(k_{\rho})}\,
		\frac{1}{
			\bm{\omega}^{1/2} (1-\bm{\omega})^{1/2} (1-k_{\rho}^{2}\bm{\omega})^{1/2}
		}
		\\
		\parderiv{\bm{\omega}}{\bm{v}} &=
		\frac{
			\textcolor{Blue1}{(\bm{v}_{\alpha}-\bm{v}_{a})}
		}{
			\textcolor{Red2}{(\bm{v}-\bm{v}_{a})^{2}}
		}
		\frac{
			\textcolor{Yellow4}{(\bm{v}_{\beta}-\bm{v}_{a})}
		}{
			\textcolor{Green4}{(\bm{v}_{\beta}-\bm{v}_{\alpha})}
		}
		\\
		\parderiv{\bm{v}}{\bm{r}} &=
		{k_{r}'}^{2} \sd\!\left(K(k_{r}) \frac{2\bm{r}}{W},k_{r}\right)
		\nd\!\left(K(k_{r}) \frac{2\bm{r}}{W},k_{r}\right)
		\frac{2}{W} K(k_{r})
	\end{align}
\end{subequations}
where
\begin{subequations}
	\begin{align}
		\bm{\omega} &=
		\frac{
			(\bm{v}-\bm{v}_{\alpha}) \textcolor{Yellow4}{(\bm{v}_{\beta}-\bm{v}_{a})}
		}{
			\textcolor{Red2}{(\bm{v}-\bm{v}_{a})} \textcolor{Green4}{(\bm{v}_{\beta}-\bm{v}_{\alpha})}
		}
		\\
		(1 - \bm{\omega}) &=
		\frac{
			(\bm{v}-\bm{v}_{\beta}) \textcolor{Blue1}{(\bm{v}_{\alpha}-\bm{v}_{a})}
		}{
			\textcolor{Red2}{(\bm{v}-\bm{v}_{a})} \textcolor{Green4}{(\bm{v}_{\beta}-\bm{v}_{\alpha})}
		} \e^{\bm{i}(\pi + 2\pi\mathds{Z})}
		\\
		(1 - k_{\rho}^{2}\bm{\omega}) &=
		\frac{
			(\bm{v}_{b}-\bm{v}) \textcolor{Blue1}{(\bm{v}_{\alpha}-\bm{v}_{a})}
		}{
			\textcolor{Red2}{(\bm{v}-\bm{v}_{a})} (\bm{v}_{b}-\bm{v}_{\alpha})
		}
	\end{align}
\end{subequations}
since $1/k_{\rho}^{2} = \bm{\omega}_{b} = T_{v}^{\omega}(\bm{v}_{b})$. 
Later, two components are combined
\begin{equation}
	\parderiv{\bm{\rho}}{\bm{\omega}} \parderiv{\bm{\omega}}{\bm{v}} =
	\frac{\e^{-\bm{i}(\pi/2 + \pi\mathds{Z})}}{2 K(k_{\rho})}\,
	\frac{
		(\bm{v}_{b}-\bm{v}_{\alpha})^{1/2} 
		\textcolor{Yellow4}{(\bm{v}_{\beta}-\bm{v}_{a})^{1/2}}
	}{
	(\bm{v}-\bm{v}_{\alpha})^{1/2} (\bm{v}-\bm{v}_{\beta})^{1/2}
	}
	\frac{1}{
		\textcolor{Red2}{(\bm{v}-\bm{v}_{a})^{1/2}}
		(\bm{v}_{b}-\bm{v})^{1/2}
	}
\end{equation}
which leads to an expression that consists of two complex branches:
$\e^{-\bm{i}(\pi/2 + \pi\mathds{Z})} = \pm \bm{i}$.
Using the identity $1-\cd(\bm{u},k)^{2} = {k'}^{2} \sd(\bm{u},k)^{2}$ from
\citex{Eq. (\dlmf[E]{22.6.}{4})}{dlmf} \cite[Eq. (1.1)]{Carlson2004nov}
and $\bm{v}_{a} = -1$ and $\bm{v}_{b} = 1$ from Fig. \ref{si:fig:T}
\begin{equation}
	(\bm{v}-\bm{v}_{a})(\bm{v}_{b}-\bm{v})
	= 1-\bm{v}^{2}
	= {k_{r}'}^{2} \sd\!\left(K(kr) \frac{2\bm{r}}{W}, kr\right)^{2}
\end{equation}
the final expression for the complex derivative can be obtained
\begin{equation}
	\label{si:eqn:dTdr}
	\parderiv{\bm{\rho}}{\bm{\omega}}
	\parderiv{\bm{\omega}}{\bm{v}}
	\parderiv{\bm{v}}{\bm{r}}
	=
	\pm \bm{i} \frac{k_{r}'}{W} \frac{K(k_{r})}{K(k_{\rho})}\,
	\frac{
		(\bm{v}_{b}-\bm{v}_{\alpha})^{1/2} (\bm{v}_{\beta}-\bm{v}_{a})^{1/2}
	}{
		(\bm{v}-\bm{v}_{\alpha})^{1/2} (\bm{v}-\bm{v}_{\beta})^{1/2}
	}
	\nd\!\left(K(kr) \frac{2\bm{r}}{W}, kr\right)
\end{equation}

In case $c_{O,f}^{B} > c_{O,f}^{A}$,
the current density at each electrode band $B$ must be positive
\begin{equation}
	j_{f}(x) =
	F n_e D_{O} [c_{O,f}^{B} - c_{O,f}^{A}]
	\Im \parderiv{\bm{\rho}}{\bm{r}}(x) 
\end{equation}
In order to achieve this,
the $+\bm{i}$ branch of the complex derivative must be chosen as the one carrying physical meaning,
such that the current density be positive on each electrode band $B$.

\subsection{Differences of concentrations when using internal counter electrode}
\label{si:proof:cE-cbar}

\subsubsection{Difference with respect to the average in the unit cell}

Inside the unit cell the average concentration satisfies \cite[Remark 2.2]{GuajardoYevenes2013sep}
\begin{subequations}
	\begin{gather}
	\int_{\unit} c_{\sigma,f}(x,z) - \bar{c}_{\sigma,i}^{\unit} \ud{x}ç
	= 0
	\\
	\intertext{%
		(either when an integer number of unit cells fits exactly in the whole cell,
		or when it doesn't, at least the number of unit cells is large enough)
		since the counter electrode is internal to the IDAE.
		Considering that the concentrations at the bands $A$ and $B$ are uniform,
		due to reversible eletrode reactions, one obtains
	}
	w_{A} [c_{\sigma,f}^{A} - \bar{c}_{\sigma,i}^{\unit}]
	+ \int_{w_{A}}^{W-w_{B}} 
	c_{\sigma,f}(x,z) - \bar{c}_{\sigma,i}^{\unit}
	\ud{x}
	+ w_{B} [c_{\sigma,f}^{B} - \bar{c}_{\sigma,i}^{\unit}]
	= 0
	\\
	\intertext{%
		If the bands have equal width $2w_{A} = 2w_{B}$,
		the integral in the gap between consecutive bands is zero,
		due to symmetry, which leads to
	}
	[c_{\sigma,f}^{A} - \bar{c}_{\sigma,i}^{\unit}]
	+ [c_{\sigma,f}^{B} - \bar{c}_{\sigma,i}^{\unit}]
	= 0
	\end{gather}
\end{subequations}

\subsubsection{Difference with respect to the average in the whole cell}

Note that the average in the unit cell may not equal that in the whole cell
\begin{subequations}
	\begin{gather}
	\int_{\whole} c_{\sigma,f}(x,z) - \bar{c}_{\sigma,i}^{\whole} \ud{x} = 0
	\\
	\intertext{this can be seen by separating the integral inside the IDAE and outside it (IDAE$'$)}
	\int_{\IDAE} c_{\sigma,f}(x,z) - \bar{c}_{\sigma,i}^{\whole} \ud{x}
	+ \int_{\IDAE'} c_{\sigma,f}(x,z) - \bar{c}_{\sigma,i}^{\whole} \ud{x}
	= 0
	\\
	2W N_{E} [\bar{c}_{\sigma,i}^{\unit} - \bar{c}_{\sigma,i}^{\whole}]
	+ \int_{\IDAE'} c_{\sigma,f}(x,z) - \bar{c}_{\sigma,i}^{\whole} \ud{x}
	= 0
	\\
	\intertext{Since the thickness of the fully-developed diffusion layer is in the order of the center-to-center separation between consecutive bands (namely $k W$), and any concentration outside the diffusion layer should equal the bulk $\bar{c}_{\sigma,i}^{\whole}$, then}
	2W N_{E} [\bar{c}_{\sigma,i}^{\unit} - \bar{c}_{\sigma,i}^{\whole}]
	+ 2 kW [c_{\sigma,f}(x_{o}, z) - \bar{c}_{\sigma,i}^{\whole}]
	= 0
	\end{gather}
\end{subequations}
where $x_{o}$ is some value within the diffusion layer,
due to the \emph{first mean value theorem for integrals}
\citex{Eq. (\dlmf[E]{1.4.}{29})}{dlmf}
(note that the choice of $k$ and $x_{o}$ depends on $z$).

Therefore, the larger the number of bands $N_{E}$, the closer the averages
$\bar{c}_{\sigma,i}^{\unit} = \bar{c}_{\sigma,i}^{\whole}$.
This leads to the following equality when the bands satisfy $2w_{A} = 2w_{B}$
\begin{equation}
	\label{si:eqn:cA-cbar=cB-cbar}
	[c_{\sigma,f}^{A} - \bar{c}_{\sigma,i}^{\whole}]
	= -[c_{\sigma,f}^{B} - \bar{c}_{\sigma,i}^{\whole}]
\end{equation}
due to the result in the previous section,
which allows to estimate the unknown concentration on the internal counter electrode.

\subsubsection{Limits for the differences of concentration}

To obtain the limits for the difference of concentration in steady state,
one considers Eq. \eqref{si:eqn:cA-cbar=cB-cbar} with non-negative concentrations on all electrodes
\begin{subequations}
	\begin{align}
	-\bar{c}_{\sigma,i}^{\whole}
	\leq [c_{\sigma,f}^{E} - \bar{c}_{\sigma,i}^{\whole}]
	&= -[c_{\sigma,f}^{E'} - \bar{c}_{\sigma,i}^{\whole}]
	\leq \bar{c}_{\sigma,i}^{\whole}
	\\ {}
	-\bar{c}_{\sigma',i}^{\whole}
	\leq [c_{\sigma',f}^{E} - \bar{c}_{\sigma',i}^{\whole}]
	&= -[c_{\sigma',f}^{E'} - \bar{c}_{\sigma',i}^{\whole}]
	\leq \bar{c}_{\sigma',i}^{\whole}
	\end{align}
\end{subequations}
where $E'$ are the complementary bands of $E \in \{A, B\}$
and $\sigma'$ is the complementary species of $\sigma \in \{O,R\}$.

The limits for the concentration of species $\sigma'$
may also affect the limits for the concentration of species $\sigma$.
This can be seen by applying the weighted sum of concentrations in Eq. \eqref{main:eqn:cf:total}
at the bands $E$ and $E'$
\begin{subequations}
	\begin{align}
	D_{\sigma} [c_{\sigma,f}^{E} - \bar{c}_{\sigma,i}^{\whole}]
	+ D_{\sigma'} [c_{\sigma',f}^{E} - \bar{c}_{\sigma',i}^{\whole}] &= 0
	\\
	D_{\sigma} [c_{\sigma,f}^{E'} - \bar{c}_{\sigma,i}^{\whole}]
	+ D_{\sigma'} [c_{\sigma',f}^{E'} - \bar{c}_{\sigma',i}^{\whole}] &= 0
	\end{align}
\end{subequations}
which agrees with \cite[Eqs. (15) and (16)]{Morf2006may} when $D_{O} = D_{R}$.

Combining the last four equations leads to
\begin{subequations}
	\begin{align}
	-D_{\sigma} \bar{c}_{\sigma,i}^{\whole}
	\leq D_{\sigma} [c_{\sigma,f}^{E} - \bar{c}_{\sigma,i}^{\whole}]
	&= -D_{\sigma} [c_{\sigma,f}^{E'} - \bar{c}_{\sigma,i}^{\whole}]
	\leq D_{\sigma} \bar{c}_{\sigma,i}^{\whole}
	\\ {}
	-D_{\sigma'} \bar{c}_{\sigma',i}^{\whole}
	\leq D_{\sigma} [c_{\sigma,f}^{E} - \bar{c}_{\sigma,i}^{\whole}]
	&= -D_{\sigma} [c_{\sigma,f}^{E'} - \bar{c}_{\sigma,i}^{\whole}]
	\leq D_{\sigma'} \bar{c}_{\sigma',i}^{\whole}
	\end{align}
\end{subequations}
which can be summarized as
\begin{equation}
	-D_{\lambda} \bar{c}_{\lambda,i}^{\whole}
	\leq D_{\sigma} [c_{\sigma,f}^{E} - \bar{c}_{\sigma,i}^{\whole}]
	= -D_{\sigma} [c_{\sigma,f}^{E'} - \bar{c}_{\sigma,i}^{\whole}]
	\leq D_{\lambda} \bar{c}_{\lambda,i}^{\whole}
\end{equation}
when defining $D_{\lambda} \bar{c}_{\lambda,i}^{\whole} = \min(D_{O} \bar{c}_{O,i}^{\whole}, D_{R} \bar{c}_{R,i}^{\whole})$.

\subsection{Difference of potential (voltage) when using internal counter electrode}
\label{si:proof:etaE-etaE'}

Consider Eq. \eqref{main:eqn:cE-cE':intC} with
$c_{\sigma,f}^{E} - c_{\sigma,f}^{E'} 
= 2[c_{\sigma,f}^{E} - \bar{c}_{\sigma,i}^{\whole}]
= -2[c_{\sigma,f}^{E'} - \bar{c}_{\sigma,i}^{\whole}]$,
which is written below for convenience
\begin{subequations}
	\begin{align}
		\frac{
			c_{\sigma,f}^{E} - \bar{c}_{\sigma,i}^{\whole}
		}{
			\bar{c}_{\sigma,i}^{\whole}
		}
		&= \frac{
			D_{\sigma'} \bar{c}_{\sigma',i}^{\whole}
			- D_{\sigma'} \bar{c}_{\sigma',i}^{\whole}
			\e^{\mp (\eta_{f}^{E} - \eta_{\nul})}
		}{
			D_{\sigma} \bar{c}_{\sigma,i}^{\whole}
			+ D_{\sigma'} \bar{c}_{\sigma',i}^{\whole}
			\e^{\mp (\eta_{f}^{E} - \eta_{\nul})}
		}
		= \frac{
			1 - \e^{\mp (\eta_{f}^{E} - \eta_{\nul})}
		}{
			r_{\sigma}^{-1} + \e^{\mp (\eta_{f}^{E} - \eta_{\nul})}
		}
		\\
		\intertext{%
			where $r_{\sigma} 
			= (D_{\sigma'} \bar{c}_{\sigma',i}^{\whole})
			/(D_{\sigma} \bar{c}_{\sigma,i}^{\whole})$.
			A similar expressions is obtained at the complementary electrode $E'$}
		-\frac{
			c_{\sigma,f}^{E} - \bar{c}_{\sigma,i}^{\whole}
		}{
			\bar{c}_{\sigma,i}^{\whole}
		}
		&= \frac{
			D_{\sigma'} \bar{c}_{\sigma',i}^{\whole}
			- D_{\sigma'} \bar{c}_{\sigma',i}^{\whole}
			\e^{\mp (\eta_{f}^{E'} - \eta_{\nul})}
		}{
			D_{\sigma} \bar{c}_{\sigma,i}^{\whole}
			+ D_{\sigma'} \bar{c}_{\sigma',i}^{\whole}
			\e^{\mp (\eta_{f}^{E'} - \eta_{\nul})}
		}
		= \frac{
			1 - \e^{\mp (\eta_{f}^{E'} - \eta_{\nul})}
		}{
			r_{\sigma}^{-1} + \e^{\mp (\eta_{f}^{E'} - \eta_{\nul})}
		}
	\end{align}
\end{subequations}

Now taking their inverses
\begin{subequations}
	\label{si:eqn:etaE}
	\begin{align}
		\e^{\mp (\eta_{f}^{E} - \eta_{\nul})}
		&= \frac{
			1 - r_{\sigma}^{-1}
			[c_{\sigma,f}^{E} - \bar{c}_{\sigma,i}^{\whole}]
			/\bar{c}_{\sigma,i}^{\whole}
		}{
			1 + \phantom{r_{\sigma}^{-1}}
			[c_{\sigma,f}^{E} - \bar{c}_{\sigma,i}^{\whole}]
			/\bar{c}_{\sigma,i}^{\whole}
		}
		\\
		\e^{\mp (\eta_{f}^{E'} - \eta_{\nul})}
		&= \frac{
			1 + r_{\sigma}^{-1}
			[c_{\sigma,f}^{E} - \bar{c}_{\sigma,i}^{\whole}]
			/\bar{c}_{\sigma,i}^{\whole}
		}{
			1 - \phantom{r_{\sigma}^{-1}}
			[c_{\sigma,f}^{E} - \bar{c}_{\sigma,i}^{\whole}]
			/\bar{c}_{\sigma,i}^{\whole}
		}
	\end{align}
\end{subequations}
and combining both expressions
\begin{equation}
	\e^{\pm (\eta_{f}^{E} - \eta_{f}^{E'})}
	= \frac{
		\displaystyle
		1 + \phantom{r_{\sigma}^{-1}}
		\frac{
			c_{\sigma,f}^{E} - \bar{c}_{\sigma,i}^{\whole}
		}{
			\bar{c}_{\sigma,i}^{\whole}
		}
	}{
		\displaystyle
		1 - r_{\sigma}^{-1}
		\frac{
			c_{\sigma,f}^{E} - \bar{c}_{\sigma,i}^{\whole}
		}{
			\bar{c}_{\sigma,i}^{\whole}
		}
	}
	\cdot \frac{
		\displaystyle
		1 + r_{\sigma}^{-1}
		\frac{
			c_{\sigma,f}^{E} - \bar{c}_{\sigma,i}^{\whole}
		}{
			\bar{c}_{\sigma,i}^{\whole}
		}
	}{
		\displaystyle
		1 - \phantom{r_{\sigma}^{-1}}
		\frac{
			c_{\sigma,f}^{E} - \bar{c}_{\sigma,i}^{\whole}
		}{
			\bar{c}_{\sigma,i}^{\whole}
		}
	}
\end{equation}
leads to
\begin{equation}
	\e^{\pm (\eta_{f}^{E} - \eta_{f}^{E'})}
	= \frac{
		\displaystyle
		1 + \frac{
			c_{\sigma,f}^{E} - \bar{c}_{\sigma,i}^{\whole}
		}{
			\bar{c}_{\sigma,i}^{\whole}
		}
	}{
		\displaystyle
		1 - \frac{
			c_{\sigma,f}^{E} - \bar{c}_{\sigma,i}^{\whole}
		}{
			\bar{c}_{\sigma,i}^{\whole}
		}
	}
	\cdot \frac{
		\displaystyle
		1 + r_{\sigma}^{-1}
		\frac{
			c_{\sigma,f}^{E} - \bar{c}_{\sigma,i}^{\whole}
		}{
			\bar{c}_{\sigma,i}^{\whole}
		}
	}{
		\displaystyle
		1 - r_{\sigma}^{-1}
		\frac{
			c_{\sigma,f}^{E} - \bar{c}_{\sigma,i}^{\whole}
		}{
			\bar{c}_{\sigma,i}^{\whole}
		}
	}
\end{equation}
after reorganizing the factors.

Finally, by using the hyperbolic identity \citex{Eq. (\dlmf[E]{4.37.}{24})}{dlmf}
\begin{equation}
	\arctanh z = \frac{1}{2} \ln\left( \frac{1 + z}{1 - z} \right)
\end{equation}
one obtains
\begin{equation}
	\begin{split}
		\pm (\eta_{f}^{E} - \eta_{f}^{E'})
		&= 2 \arctanh\left( \frac{
			c_{\sigma,f}^{E} - \bar{c}_{\sigma,f}^{\whole}
		}{
		\bar{c}_{\sigma,f}^{\whole}
		} \right)
		\\
		&+ 2 \arctanh\left( \frac{
			D_{\sigma} \bar{c}_{\sigma,f}^{\whole}
		}{
		D_{\sigma'} \bar{c}_{\sigma',f}^{\whole}
		}
		\frac{
			c_{\sigma,f}^{E} - \bar{c}_{\sigma,f}^{\whole}
		}{
		\bar{c}_{\sigma,f}^{\whole}
		} \right)
	\end{split}
\end{equation}
which leads to Eq. \eqref{main:eqn:etaE-etaE':intC} after applying the relation
in Eq. \eqref{main:eqn:cE':intC} $c_{\sigma,f}^{E} - c_{\sigma,f}^{E'}
= 2 [c_{\sigma,f}^{E} - \bar{c}_{\sigma,f}^{\whole}]$.

\subsection{Approximations for $k_{\rho}$ and $k_{\rho}'$}

\subsubsection{Case of tall cells}
\label{si:proof:K'K:Hinf}

Consider first the following approximations when the nome $q = Q(k)$
tends to zero \cite[(10.1)--(10.3)]{Fenton1982jul}
\begin{subequations}
	\label{main:eqn:jacobi:q0}
	\begin{align}
		\lim_{q \to 0^{+}} \sn(\bm{u}, k)
		= \lim_{q \to 0^{+}} \sd(\bm{u}, k)
		&= \sin\!\left( \frac{\pi}{2} \frac{\bm{u}}{K(k)} \right) 
		\\
		\lim_{q \to 0^{+}} \cn(\bm{u}, k)
		= \lim_{q \to 0^{+}} \cd(\bm{u}, k)
		&= \cos\!\left( \frac{\pi}{2} \frac{\bm{u}}{K(k)} \right)
		\\
		\lim_{q \to 0^{+}} \dn(\bm{u}, k)
		= \lim_{q \to 0^{+}} \nd(\bm{u}, k)
		&= 1
	\end{align}
\end{subequations}

Since $q_{r} = Q(k_{r}) = \exp(-\pi\, 2H/W)$ in Eq. \eqref{main:eqn:kr}
tends to zero when $H \to +\infty$,
then
\begin{subequations}
	\label{si:eqn:moduli:alt}
	\begin{align}
		k_{\rho}^{2}
		&= \frac{
			\displaystyle
			4\sn\!\left( K(k_{r}) \frac{g}{W}, k_{r} \right)
		}{
			\displaystyle
			\left[1 + \sn\!\left( K(k_{r}) \frac{g}{W}, k_{r} \right) \right]^{2}
		}
		\\
		{k_{\rho}'}^{2}
		&= {k_{r}'}^{4}\,
		\frac{
			\displaystyle
			\sd\!\left( K(k_{r}) \frac{w_{A}}{W}, k_{r} \right)^{2}
		}{
			\displaystyle
			\cn\!\left( K(k_{r}) \frac{w_{A}}{W}, k_{r} \right)^{2}
		}
		\frac{
			\displaystyle
			\sd\!\left( K(k_{r}) \frac{w_{B}}{W}, k_{r} \right)^{2}
		}{
			\displaystyle
			\cn\!\left( K(k_{r}) \frac{w_{B}}{W}, k_{r} \right)^{2}
		}
	\end{align}
\end{subequations}
can be approximated by
\begin{subequations}
	\begin{align}
		\lim_{q_{r} \to 0^{+}} k_{\rho}^{2}
		&= \frac{
			\displaystyle
			4\sin\!\left( \frac{\pi}{2} \frac{g}{W} \right)
		}{
			\displaystyle
			\left[1 + \sin\!\left( \frac{\pi}{2} \frac{g}{W} \right) \right]^{2}
		}
		\\
		\lim_{q_{r} \to 0^{+}} {k_{\rho}'}^{2}
		&=
		\tan\!\left( \frac{\pi}{2} \frac{w_{A}}{W} \right)^{2}
		\tan\!\left( \frac{\pi}{2} \frac{w_{B}}{W} \right)^{2}
	\end{align}
\end{subequations}
which leads to Eqs. \eqref{main:eqn:moduli:Hinf}.

\subsubsection{Case of shallow cells}
\label{si:proof:K'K:H0}

Consider the following approximations when the associated nome
$q' = Q(k')$ is sufficiently small \cite[Eqs. (11.1)--(11.3)]{Fenton1982jul}
\begin{subequations}
	\label{main:eqn:jacobi:q'0}
	\begin{align}
		\sn(\bm{u}, k) \bigg|_{q' \approx 0} &\approx
		(k^{2})^{-1/4} \tanh\!\left(
		\frac{\pi}{2} \frac{\bm{u}}{K'(k)}
		\right)
		\\
		\cn(\bm{u}, k) \bigg|_{q' \approx 0} &\approx
		\frac{1}{2} \left( \frac{{k'}^{2}}{k^{2} q'} \right)^{1/4} \sech\!\left( \frac{\pi}{2} \frac{\bm{u}}{K'(k)} \right)
		\\
		\dn(\bm{u}, k) \bigg|_{q' \approx 0} &\approx
		\frac{1}{2} \left( \frac{{k'}^{2}}{q'} \right)^{1/4}
		\sech\!\left( \frac{\pi}{2} \frac{\bm{u}}{K'(k)} \right)
		\\
		\intertext{also of the subsidiary function $\sd = \sn/\dn$}
		\sd(\bm{u}, k) \bigg|_{q' \approx 0} &\approx
		2 \left( \frac{q'}{k^{2} {k'}^{2}} \right)^{1/4}
		\sinh\!\left( \frac{\pi}{2} \frac{\bm{u}}{K'(k)} \right)
		\\
		\intertext{%
			and of the modulus $(k^{2})^{1/4}$
			\cite[begining of \S 7]{Fenton1982jul}
		}
		(k^{2})^{1/4} \bigg|_{q' \approx 0} &\approx
		\frac{1 - 2q'}{1 + 2q'} =
		\tanh\!\left(
			\frac{\pi}{2} \frac{K(k)}{K'(k)} - \frac{1}{2} \ln 2
		\right)
	\end{align}
\end{subequations}

Since  $q_{r}' = Q(k_{r}') =\exp(-\pi W/2H)$ in Eq. \eqref{main:eqn:kr}
is approximately zero when $H \approx 0$,
then $k_{\rho}$ and ${k_{\rho}}'$ in Eqs. \eqref{si:eqn:moduli:alt} can be approximated by
\begin{subequations}
	\begin{align}
		k_{\rho}^{2} \bigg|_{\scriptsize \shortstack{$q_{r}' \approx 0$ \\ $g \approx 0$}}
		&\approx
		4\sn\!\left( K(k_{r}) \frac{g}{W}, k_{r} \right)
		\bigg|_{q_{r}' \approx 0}
		\approx
		4\, \frac{
			\displaystyle
			\tan\!\left( 
				\frac{\pi}{2} \frac{K(k_{r})}{K'(k_{r})} \frac{g}{W}
			\right)
		}{
			\displaystyle
			\tan\!\left( 
				\frac{\pi}{2} \frac{K(k_{r})}{K'(k_{r})} - \ln\sqrt{2} 
			\right)
		}
		\\
		{k_{\rho}'}^{2} \bigg|_{q_{r}' \approx 0}
		&\approx 
		16 {q_{r}'}^{2}\,
		\sinh\!\left(
			\pi \frac{K(k_{r})}{K'(k_{r})} \frac{w_{A}}{W}
		\right)^{2}
		\sinh\!\left(
			\pi \frac{K(k_{r})}{K'(k_{r})} \frac{w_{B}}{W}
		\right)^{2}
	\end{align}
\end{subequations}
\noindent where the last equality was obtained through the identity
$2 \sinh(\bm{u}) \cosh(\bm{u}) = \sinh(2\bm{u})$, as shown below for $E \in \{A, B\}$
\begin{equation}
	{k_{r}'}^{2}\,
	\frac{
		\displaystyle
		\sd\!\left( K(k_{r}) \frac{w_{E}}{W}, k_{r} \right)^{2}
	}{
		\displaystyle
		\cn\!\left( K(k_{r}) \frac{w_{E}}{W}, k_{r} \right)^{2}
	}
	\approx
	16 {q_{r}'}\,
	\frac{
		\displaystyle
		\sinh\!\left( \frac{\pi}{2} \frac{K(k_{r})}{K'(k_{r})} \frac{w_{E}}{W} \right)^{2}
	}{
		\displaystyle
		\sech\!\left( \frac{\pi}{2}\frac{K(k_{r})}{K'(k_{r})} \frac{w_{E}}{W} \right)^{2}
	}
	=
	4 {q_{r}'}\,
	\sinh\!\left( \pi \frac{K(k_{r})}{K'(k_{r})} \frac{w_{E}}{W} \right)^{2}
\end{equation}

Finally, due to Eqs. \eqref{main:eqn:kr} and \eqref{elipticas:eqn:nomo},
that is $K(k_{r})/K'(k_{r}) = -\ln Q(k_{r}')/\pi = W/2H$,
the desired result in Eqs. (\ref{main:eqn:moduli:H0}) is obtained.

%% file: text-numerical_calculations.tex

\section{Numerical calculations}

\subsection{Simulations in steady state}

\subsubsection{Implementation}

A normalized version of the diffusion equation in Eqs. \eqref{main:eqn:pde} was used for the simulations

\begin{subequations}
	\label{si:eqn:pde:normalized}
	\begin{align}
		\parderiv{^2 \xi\simu}{x\simu^{2}}(x, z)
		+ \parderiv{^{2} \xi\simu}{z\simu^{2}}(x, z)
		&= 0
		\\
		\parderiv{\xi\simu}{x\simu}(0, z)
		= \parderiv{\xi\simu}{x\simu}(W, z)
		&= 0,\, \forall z \in [0,H]
		\\
		\parderiv{\xi\simu}{z\simu}(x, H)
		&= 0,\, \forall x \in [0,W]
		\\
		\parderiv{\xi\simu}{z\simu}(x, 0)
		&= 0,\, \forall x \notin A \cup B
		\\
		\xi\simu(x, 0) &= 0,\, \forall x \in A
		\\
		\xi\simu(x, 0) &= 1,\, \forall x \in B
		\label{si:eqn:pde:normalized:onAB}
	\end{align}
\end{subequations}
where 
\begin{equation}
	\label{si:eqn:pde:normalization}
	\xi\simu = \frac{
		c_{\sigma,f} - c_{\sigma,f}^{A}
	}{
		c_{\sigma,f}^{B} - c_{\sigma,f}^{A}
	}
	,\quad x\simu = \frac{x}{W}
	,\quad z\simu = \frac{z}{W}
\end{equation}
The width of each band electrode was taken equal $2w_{A} = 2w_{B} = 2w_{E}$ for all simulations,
and different aspect ratios $H/W$ for the unit cell were considered.

An exponential mesh was used to partition the unit cell,
in order to keep the memory usage low
while maintaining good resolution near the electrode bands
(see \cite[Chapter 7]{Britz2016} for more details on this kind of mesh).
The number of elements of the mesh is $n_{x} \times n_{z}$,
of which their growth factors are $r_{x} = r_{z} = r$,
and the width and height of its smallest element are
$\delta_{x} = \delta_{z} = \delta_{0}$.
Additionally, the mesh parameter $R$ is sometimes given,
which corresponds to the ratio between the largest $(\delta_{0}\,r^{n_{z}-1}$)
and the smallest ($\delta_{0}$) elements in the mesh.

\subsubsection{Selection of the parameters for the exponential mesh}

Here four cases were considered:
$2w_{E}/W \in \mathopen\{ \numlist[list-final-separator={, }]{0,2; 0,4; 0,6; 0,8} \mathclose\}$ with $H/W = \num{1,0}$.
The mesh for each case was succesively refined until the absolute error between the simulated $\xi_{\simu}(x,z)$
and the theoretical concentration profile $\xi(x,z)$ was less than \num{0,005}
(approximately to two decimal places of accuracy). See \emph{Refinement output of Fig. \ref{si:fig:mesh}} at the end of this section.

\begin{figure}
	\centering
	\subcaptionbox{
		$2w_{E}/W = \num{0,2}$.
		Mesh:
		$\delta_{0} = \num{0.0005}$,
		$r = \num{1.2482936464}$,
		$n_{x} \times n_{z} = 84 \times 28$.
		Error conc.:
		\num{+- 0.0039}.
	}{\includegraphics{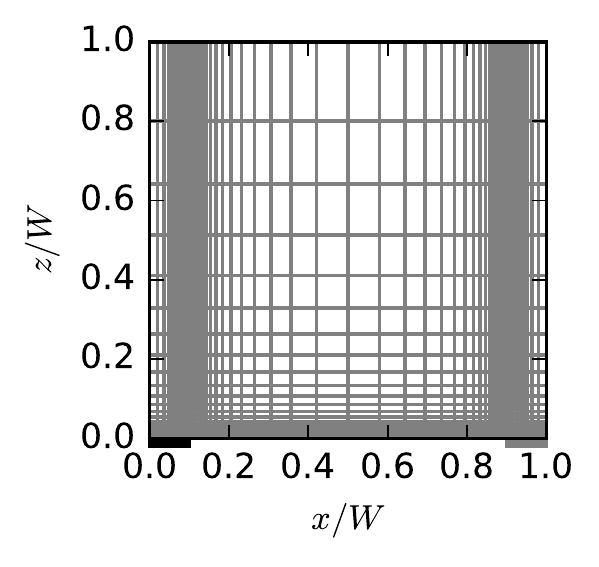}}
	\quad
	\subcaptionbox{
		$2w_{E}/W = \num{0,4}$.
		Mesh:
		$\delta_{0} = \num{0.0005}$,
		$r = \num{1.2482936464}$,
		$n_{x} \times n_{z} = 88 \times 28$.
		Error conc.:
		\num{+- 0.0040}.
	}{\includegraphics{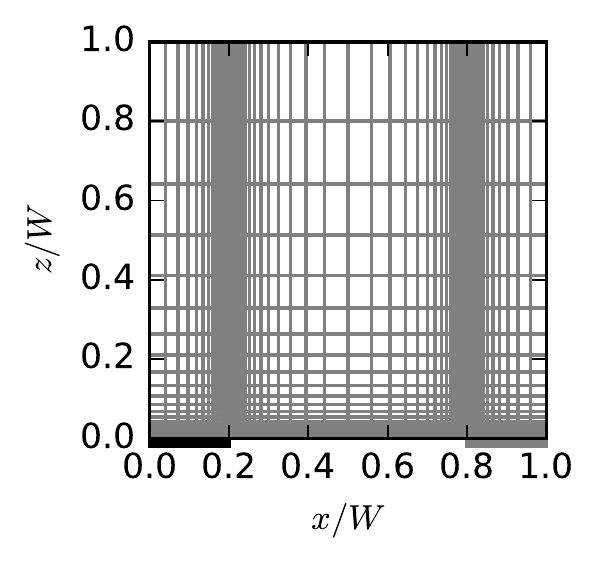}}

	\subcaptionbox{
		$2w_{E}/W = \num{0,6}$.
		Mesh:
		$\delta_{0} = \num{0.0005}$,
		$r = \num{1.2482936464}$,
		$n_{x} \times n_{z} = 88 \times 28$.
		Error conc.:
		\num{+- 0.0046}.
	}{\includegraphics{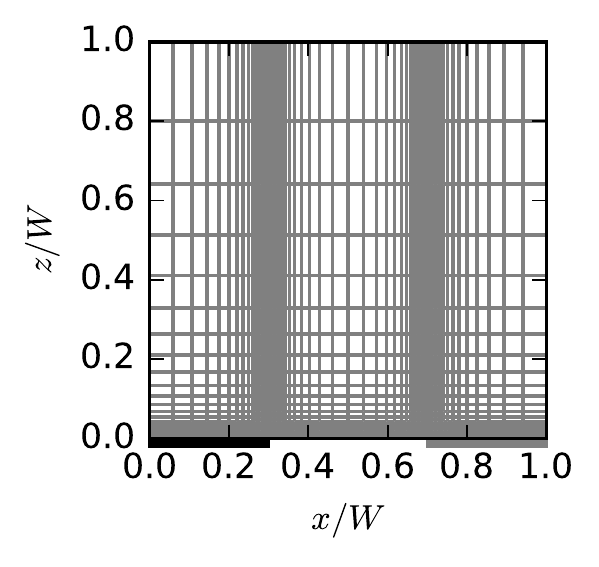}}
	\quad
	\subcaptionbox{
		$2w_{E}/W = \num{0,8}$.
		Mesh:
		$\delta_{0} = \num{0.00025}$,
		$r = \num{1.24958205268}$,
		$n_{x} \times n_{z} = 96 \times 31$.
		Error conc.:
		\num{+- 0.0045}.
	}{\includegraphics{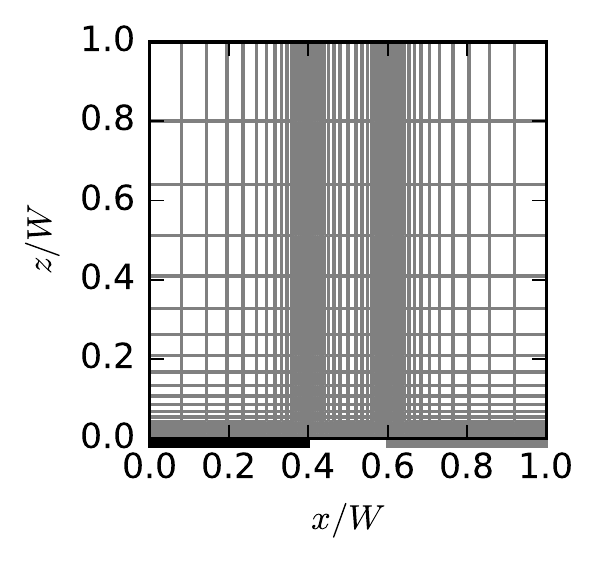}}

	\caption{
		Exponential mesh candidates.
		The size of each mesh is $n_{x} \times n_{z}$,
		the dimensions of its smallest element are $\delta_{x} = \delta_{z} = \delta_{0}$
		and the growth factors are $r_{x} = r_{z} = r$.
		The parameters of the finest mesh among the candidates,
		that is $\delta_{0} = \num{0.00025}$ and $r = \num{1.24958205268}$,
		were selected for all simulations in this text.  
	}
	\label{si:fig:mesh}
\end{figure}

The mesh candidates are shown in Fig. \ref{si:fig:mesh}.
The case for $2w_{E}/W = \num{0,8}$ is the most restrictive of all,
and therefore it requires the finest mesh to reach the desired accuracy.
This is because the current flowing in this cell is the highest,
due to the largest electrode widths,
producing the greatest gradient of concentration among all the cases.

See \emph{Refinement output of Fig. \ref{si:fig:mesh}}
for values of the normalized current $2 K'(k_{\rho})/K(k_{\rho})$.
Note that the normalized current has been writen in terms of the
lattice parameter $\tau$ for brevity, that is  $-2\bm{i}\tau = 2K'(k)/K(k)$
due to Eq. \citex{(\dlmf[E]{22.2.}{12})}{dlmf}.
The ratio $K'(k_{\rho})/K(k_{\rho})$ obtained in the simulations
was computed numerically using the expression
\begin{equation}
	\label{si:eqn:int_dxi7bz}
	-\bm{i}\tau = \frac{K'(k_{\rho})}{K(k_{\rho})}
	= \left. \int_{0}^{w_{A}} \parderiv{\xi}{z} \right|_{z=0} \ud{x}
	= \left. \int_{0}^{w_{A}/W} \parderiv{\xi_{\simu}}{z_{\simu}} \right|_{z_{\simu}=0} \ud{x_{\simu}}
\end{equation}
due to Eqs. \eqref{main:eqn:int_dxi7dz_dx}, \eqref{main:eqn:int_drho} and \eqref{main:eqn:rhoAB}.

Finally, the parameters of the finest mesh among the candidates in Fig. \ref{si:fig:mesh}:
\begin{equation}
	\label{si:eqn:mesh}
	\delta_{0} = \num{0.00025} \text{ and } r = \num{1.24958205268}
\end{equation}
were selected for all the simulations in this report.

\paragraph{Refinement output of Fig. \ref{si:fig:mesh}}
\rule{0pt}{0pt}
\begin{lstlisting}[basicstyle={\ttfamily \scriptsize}]
numpy:      1.11.0
matplotlib: 1.5.1
fipy:       3.1.3

H/W: 1.0, 2wE/W: 0.2
------------------------------
d0: 0.001
r: 1.24668231916
R: 198.673162644
(nx, nz): (72, 25)
(xi_sim - xi)_min: -0.0055702670935 at (0.8995, 0.0005)
(xi_sim - xi)_max: 0.00557026694512 at (0.1005, 0.0005)
-2itau_sim: 0.614130087454
-2itau_sim + 2itau: -0.0035839696459406859121637690753342364079264758376749102227617713482
------------------------------
d0: 0.0005
r: 1.2482936464
R: 398.613975357
(nx, nz): (84, 28)
(xi_sim - xi)_min: -0.00394025771519 at (0.89975, 0.00025)
(xi_sim - xi)_max: 0.00394025792104 at (0.10025, 0.00025)
-2itau_sim: 0.614628997862
-2itau_sim + 2itau: -0.0030850592382462330833997267239393446690348742751749102227617713482
------------------------------

H/W: 1.0, 2wE/W: 0.4
------------------------------
d0: 0.001
r: 1.24668231916
R: 198.673162644
(nx, nz): (76, 25)
(xi_sim - xi)_min: -0.00566368665462 at (0.7995, 0.0005)
(xi_sim - xi)_max: 0.00566368569237 at (0.2005, 0.0005)
-2itau_sim: 0.85762365783
-2itau_sim + 2itau: -0.0044977238833970306187808518106150001542417595123222092768377380438
------------------------------
d0: 0.0005
r: 1.2482936464
R: 398.613975357
(nx, nz): (88, 28)
(xi_sim - xi)_min: -0.00400968096584 at (0.79975, 0.00025)
(xi_sim - xi)_max: 0.00400968107264 at (0.20025, 0.00025)
-2itau_sim: 0.858133331146
-2itau_sim + 2itau: -0.0039880505674051985422039678186676025645582267974784592768377380438
------------------------------

H/W: 1.0, 2wE/W: 0.6
------------------------------
d0: 0.001
r: 1.24668231916
R: 198.673162644
(nx, nz): (76, 25)
(xi_sim - xi)_min: -0.00654672343082 at (0.6995, 0.0005)
(xi_sim - xi)_max: 0.00654672272843 at (0.3005, 0.0005)
-2itau_sim: 1.1460941062
-2itau_sim + 2itau: -0.006097664237956998989986921071485044556081808580771025801404182422
------------------------------
d0: 0.0005
r: 1.2482936464
R: 398.613975357
(nx, nz): (88, 28)
(xi_sim - xi)_min: -0.00463873205874 at (0.69975, 0.00025)
(xi_sim - xi)_max: 0.00463872975751 at (0.30025, 0.00025)
-2itau_sim: 1.14677402693
-2itau_sim + 2itau: -0.005417743508229109240048735511899354057729757799521025801404182422
------------------------------

H/W: 1.0, 2wE/W: 0.8
------------------------------
d0: 0.001
r: 1.24668231916
R: 198.673162644
(nx, nz): (72, 25)
(xi_sim - xi)_min: -0.00899524079541 at (0.5995, 0.0005)
(xi_sim - xi)_max: 0.00899524495282 at (0.4005, 0.0005)
-2itau_sim: 1.60011947532
-2itau_sim + 2itau: -0.0097043377054256020652574316060749899358774808157430304385092440169
------------------------------
d0: 0.0005
r: 1.2482936464
R: 398.613975357
(nx, nz): (84, 28)
(xi_sim - xi)_min: -0.00638554613647 at (0.59975, 0.00025)
(xi_sim - xi)_max: 0.00638554672939 at (0.40025, 0.00025)
-2itau_sim: 1.60142001651
-2itau_sim + 2itau: -0.0084037965092709655845156030424518596143747952688680304385092440169
------------------------------
d0: 0.00025
r: 1.24958205268
R: 799.729964567
(nx, nz): (96, 31)
(xi_sim - xi)_min: -0.00452419661225 at (0.599875, 0.000125)
(xi_sim - xi)_max: 0.00452419959184 at (0.400125, 0.000125)
-2itau_sim: 1.60206219202
-2itau_sim + 2itau: -0.0077616210048388358000111347807219337434794095266805304385092440169
------------------------------

z/W: [  2.50000000e-04   5.62395513e-04   9.52759340e-04   1.44055097e-03
   2.05008664e-03   2.81175147e-03   3.76351418e-03   4.95281977e-03
   6.43895469e-03   8.29600222e-03   1.06165355e-02   1.35162322e-02
   1.71396412e-02   2.16673880e-02   2.73251792e-02   3.43950535e-02
   4.32294416e-02   5.42687343e-02   6.80632364e-02   8.53005987e-02
   1.06840097e-01   1.33755468e-01   1.67388432e-01   2.09415581e-01
   2.61931951e-01   3.27555465e-01   4.09557431e-01   5.12025615e-01
   6.40068019e-01   8.00067509e-01   1.00000000e+00]

\end{lstlisting}

\subsubsection{Simulation of concentration profile in Fig. \ref{main:fig:xi}}
\label{si:sec:xi}

Here a validation is done by comparing
the theoretical normalized concentration $\xi(x,z)$ in Eq. \eqref{main:eqn:cf}
or \eqref{si:eqn:T}
\begin{equation}
	\xi(x,z) = \frac{c_{\sigma,f}(x,z) - c_{\sigma,f}^{A}}{c_{\sigma,f}^{B} -c_{\sigma,f}^{A}} = \Re\, T_{r}^{\rho}(x,z)
\end{equation}
with its simulated counterpart in Eqs. \eqref{si:eqn:pde:normalized}
and \eqref{si:eqn:pde:normalization}.

The width of each band electrode was taken equal to
$2w_{A} = 2w_{B} = \num{0,5} W$ for all simulations
and three different aspect ratios for the unit cell
$H/W \in \{\numlist[list-final-separator={,}]{\approx 0,33; \approx 0,51; 1,0}\}$ were considered.
The parameters used to generate each mesh were obtained from Eq. \eqref{si:eqn:mesh}.

The results of the simulations are shown in Fig. \ref{si:fig:xi_sim},
of which Fig. \ref{main:fig:xi} shows only a part.
In all cases the theoretical concentration $\xi(x,z)$ 
agrees with the simulations $\xi\simu(x,z)$ up to two decimal places,
since the absolute error obtained was not greater than $\approx \num{0,0030}$.
The error appears to be greater near the big-sized elements of the mesh,
but it actually reaches its extrema at the edges of each electrode band
(this is not easy to see from Fig. \ref{si:fig:xi_sim},
but was obtained as output data of the simulations, see section \emph{Simulation output} below).

\begin{figure}
	\centering
	\subcaptionbox{
		$H/W = 1$.
		Mesh: $100 \times 31$.
		Error max/min: $\approx \pm \num{0.0030}$.
		\label{si:fig:xi_sim:HW10}
	}{
		\includegraphics{fig-03-H10-conc}
		\includegraphics{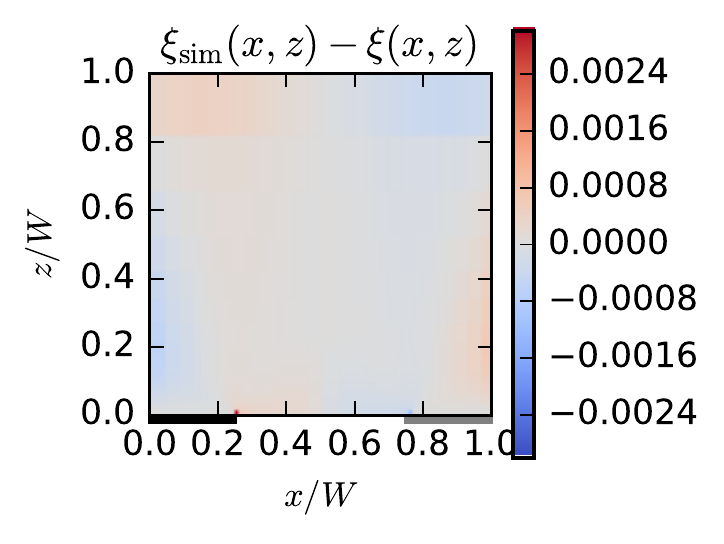}
	}
	\subcaptionbox{
		$H/W \approx \num{0,5}.$
		Mesh: $100 \times 28$.
		Error. max/min: $\approx \pm \num{0.0029}$.
		\label{si:fig:xi_sim:HW05}
	}{
		\includegraphics{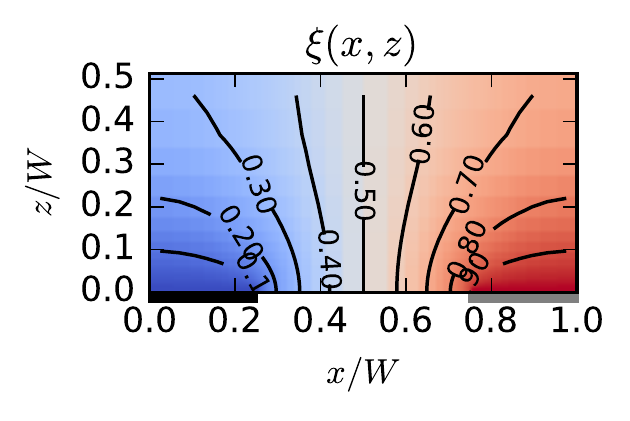}
		\includegraphics{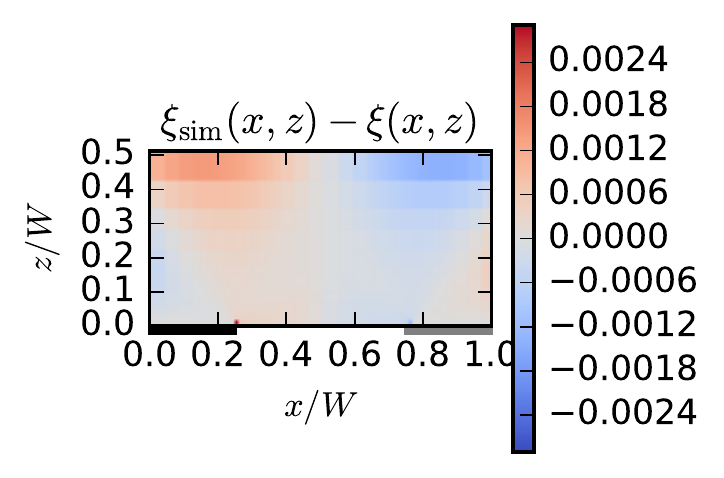}
	}
	\subcaptionbox{
		$H/W \approx \num{0,3}$.
		Mesh: $100 \times 26$.
		Error. max/min: $\approx \pm \num{0.0027}$. 
		\label{si:fig:xi_sim:HW03}
	}{
		\includegraphics{fig-03-H03-conc}
		\includegraphics{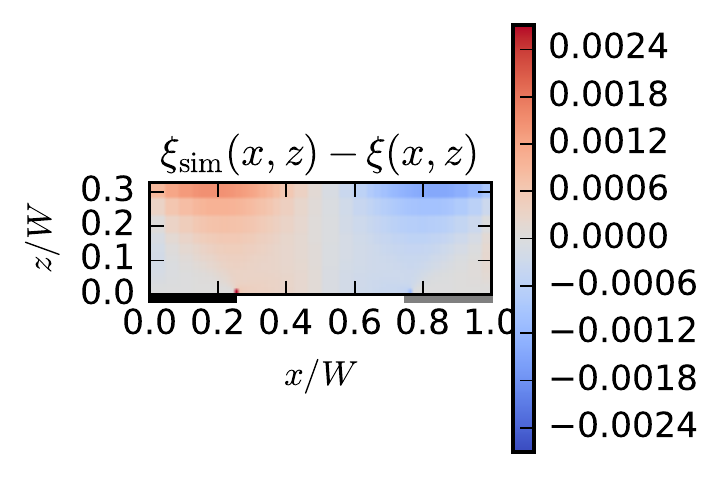}
	}
	\caption{
		Left column:
		Normalized concentration
		$\xi(x,z) = [c_{\sigma,f}(x,z) - c_{\sigma,f}^{A}]/[c_{\sigma,f}^{B} - c_{\sigma,f}^{A}]$
		in the final steady state (isoconcentration lines shown in black).
		Right column:
		Error between simulation and theoretical concentration $\xi_{\mathrm{sim}}(x,z) - \xi(x,z)$.
		Plots consider electrodes of equal size $2w_{A} = 2w_{B} = \num{0,5}W$
		and different aspect ratios $H/W$ for the unit cell.
		The simulations were obtained by numerically solving Eqs. \eqref{si:eqn:pde:normalized} in steady state.
	}
	\label{si:fig:xi_sim}
\end{figure}

\vspace{\baselineskip}
\paragraph{Simulation output}
\rule{0pt}{0pt}
\begin{lstlisting}[basicstyle={\ttfamily \scriptsize}]
numpy:      1.11.0
matplotlib: 1.5.1
fipy:       3.1.3

d0: 0.00025
r: 1.24958205268

H/W: 1.0
------------------------------
R: 799.729964567
(nx, nz): (100, 31)
(xi_sim - xi)_min: -0.00300626586786 at (0.749875, 0.000125)
(xi_sim - xi)_max: 0.00300626486245 at (0.250125, 0.000125)
-2itau_sim: 0.992253126221
-2itau_sim + 2itau: -0.00433940282946
------------------------------

H/W: 0.512025615
------------------------------
R: 409.872737037
(nx, nz): (100, 28)
(xi_sim - xi)_min: -0.00290949376575 at (0.749875, 0.000125)
(xi_sim - xi)_max: 0.00290949208295 at (0.250125, 0.000125)
-2itau_sim: 0.92562204374
-2itau_sim + 2itau: -0.00372746239555
------------------------------

H/W: 0.327555465
------------------------------
R: 262.494055937
(nx, nz): (100, 26)
(xi_sim - xi)_min: -0.0027092884472 at (0.749875, 0.000125)
(xi_sim - xi)_max: 0.00270929214004 at (0.250125, 0.000125)
-2itau_sim: 0.788457932688
-2itau_sim + 2itau: -0.00252385445816
------------------------------

\end{lstlisting}

\subsubsection{Exhaustive simulations}
\label{si:sec:exhaustive}

\begin{table}
	\centering
\begin{lstlisting}[basicstyle={\ttfamily \scriptsize}]
# mesh size (nx, nz)
#
# H/W   \    2*wE/W: [ 0.2  0.4  0.5  0.6  0.8]
0.054269  (96, 18)  (100, 18)  (100, 18)  (100, 18)  (96, 18)
0.106840  (96, 21)  (100, 21)  (100, 21)  (100, 21)  (96, 21)
0.209416  (96, 24)  (100, 24)  (100, 24)  (100, 24)  (96, 24)
0.327555  (96, 26)  (100, 26)  (100, 26)  (100, 26)  (96, 26)
0.409557  (96, 27)  (100, 27)  (100, 27)  (100, 27)  (96, 27)
0.640068  (96, 29)  (100, 29)  (100, 29)  (100, 29)  (96, 29)
0.800068  (96, 30)  (100, 30)  (100, 30)  (100, 30)  (96, 30)
1.000000  (96, 31)  (100, 31)  (100, 31)  (100, 31)  (96, 31)
1.249832  (96, 32)  (100, 32)  (100, 32)  (100, 32)  (96, 32)
\end{lstlisting}
	\caption{
		Mesh sizes $n_{x} \times n_{z}$ obtained using
		the mesh parameters in Eq. \eqref{si:eqn:mesh}
		for different values of $H/W$ and $2w_{E}/W$.
	}
	\label{si:tab:exhaustive:mesh}
\vspace{\baselineskip}
	\centering

	\subcaptionbox{
		Minimum error between normalized concentrations:
		$\min \xi_{\simu}(x,z) - \xi(x,z)$.
	}[\textwidth]{
\begin{lstlisting}[basicstyle={\ttfamily \scriptsize}]
# values of (xi_sim - xi)_min
#
# H/W   \    2*wE/W: [ 0.2  0.4  0.5  0.6  0.8]
0.054269 0.000017 -0.014094 -0.001479 -0.001883 -0.003391
0.106840 -0.001380 -0.001700 -0.001980 -0.002375 -0.003974
0.209416 -0.001959 -0.002163 -0.002430 -0.002818 -0.004316
0.327555 -0.002372 -0.002486 -0.002709 -0.003051 -0.004430
0.409557 -0.002537 -0.002621 -0.002824 -0.003141 -0.004467
0.640068 -0.002731 -0.002788 -0.002964 -0.003251 -0.004511
0.800068 -0.002770 -0.002822 -0.002994 -0.003274 -0.004520
1.000000 -0.002787 -0.002837 -0.003006 -0.003284 -0.004524
1.249832 -0.002792 -0.002842 -0.003011 -0.003287 -0.004526
\end{lstlisting}
	}

	\subcaptionbox{
		Maximum error between normalized concentrations:
		$\max \xi_{\simu}(x,z) - \xi(x,z)$.
	}[\textwidth]{
\begin{lstlisting}[basicstyle={\ttfamily \scriptsize}]
# values of (xi_sim - xi)_max
#
# H/W   \    2*wE/W: [ 0.2  0.4  0.5  0.6  0.8]
0.054269 0.999983 0.019672 0.001536 0.001876 0.003376
0.106840 0.001379 0.001700 0.001980 0.002375 0.003974
0.209416 0.001959 0.002163 0.002430 0.002818 0.004316
0.327555 0.002372 0.002486 0.002709 0.003051 0.004430
0.409557 0.002537 0.002621 0.002824 0.003141 0.004467
0.640068 0.002731 0.002788 0.002964 0.003251 0.004511
0.800068 0.002770 0.002822 0.002994 0.003274 0.004520
1.000000 0.002787 0.002837 0.003006 0.003284 0.004524
1.249832 0.002792 0.002842 0.003011 0.003287 0.004526
\end{lstlisting}
	}

	\caption{
		Error between simulated and theoretical concentrations obtained when
		simulating exhaustively for different values of $H/W$ and $2w_{E}/W$.
	}
	\label{si:tab:exhaustive:xi}
\end{table}

The normalized concentration profile $\xi(x,z)$ and current
$2 K'(k_{\rho})/K(k_{\rho})$ were exhaustively simulated
for all combinations of the following electrode widths
\begin{align}
	\frac{2w_{E}}{W}
	&= \{ \numlist[list-final-separator={, }]{0,2; 0,4; 0,5; 0,6; 0,8} \}
	\\
	\intertext{and aspect ratios}
	\frac{H}{W}
	&\approx \{ \numlist[list-final-separator={, }]{0,05; 0,11; 0,21; 0,33; 0,41; 0,64; 0,80; 1,00; 1,25} \}
\end{align}
The implementation in Eqs. \eqref{si:eqn:pde:normalized} was used,
together with the meshes specified in Table \ref{si:tab:exhaustive:mesh}
and the selected mesh parameters in Eq. \eqref{si:eqn:mesh}.
These simulations were constrasted against their theoretical counterparts,
and the results are shown in Tables \ref{si:tab:exhaustive:xi},
\ref{si:tab:exhaustive:itau} and in Fig. \ref{main:fig:K'Krho:simu}.

Table \ref{si:tab:exhaustive:xi} shows that the minimum and maximum errors
between the simulated and theoretical concentration profile
were not greater than $\approx \num{0,0045}$,
except for the two first band widths when $H/W \approx \num{0,05}$.

Table \ref{si:tab:exhaustive:itau} shows the simulated values of
the normalized current and the error with respect to its theoretical counterpart,
which was not greater than $\approx \num{0,0051}$ for most of the cases
(except for the case $H/W \approx \num{0,05}$ and $2w_{E}/W = \num{0,2}$
with an error of $\approx \num{0,1279}$,
and for the last six aspect ratios $H/W$ when $2w_{E}/W = \num{0,8}$
with an error not greater than $\approx \num{0,0078}$).

\begin{table}
	\centering

	\subcaptionbox{
		Normalized current $2K'(k_{\rho})/K(k_{\rho})_{\simu}$
		obtained through simulations.
	}[\textwidth]{
\begin{lstlisting}[basicstyle={\ttfamily \scriptsize}]
# values of 2K'(k_rho)/K(k_rho)_sim
#
# H/W   \    2*wE/W: [ 0.2  0.4  0.5  0.6  0.8]
0.054269 0.127893 0.167383 0.197903 0.242036 0.436878
0.106840 0.236720 0.307210 0.358978 0.431529 0.724088
0.209416 0.397250 0.523456 0.604589 0.710782 1.087271
0.327555 0.506054 0.683877 0.788458 0.917149 1.332990
0.409557 0.549094 0.751141 0.866497 1.004927 1.435700
0.640068 0.600287 0.834122 0.963650 1.114661 1.563992
0.800068 0.610578 0.851204 0.983779 1.137483 1.590754
1.000000 0.614875 0.858382 0.992253 1.147105 1.602062
1.249832 0.616250 0.860685 0.994976 1.150200 1.605708
\end{lstlisting}
	}

	\subcaptionbox{
		Error between normalized currents:
		$2K'(k_{\rho})/K(k_{\rho})_{\simu} - 2K'(k_{\rho})/K(k_{\rho})$.
	}[\textwidth]{
\begin{lstlisting}[basicstyle={\ttfamily \scriptsize}]
# values of 2K'(k_rho)/K(k_rho)_sim - 2K'(k_rho)/K(k_rho)
#
# H/W   \    2*wE/W: [ 0.2  0.4  0.5  0.6  0.8]
0.054269 0.127893 0.000115 -0.000196 -0.000293 -0.000960
0.106840 -0.000261 -0.000389 -0.000524 -0.000750 -0.002111
0.209416 -0.000831 -0.001115 -0.001411 -0.001869 -0.004090
0.327555 -0.001570 -0.002072 -0.002524 -0.003168 -0.005837
0.409557 -0.001994 -0.002637 -0.003162 -0.003879 -0.006672
0.640068 -0.002648 -0.003513 -0.004119 -0.004891 -0.007687
0.800068 -0.002793 -0.003695 -0.004304 -0.005068 -0.007800
1.000000 -0.002839 -0.003740 -0.004339 -0.005087 -0.007762
1.249832 -0.002837 -0.003725 -0.004314 -0.005049 -0.007690
\end{lstlisting}
	}

	\caption{
		Error between simulated and theoretical currents obtained when
		simulating exhaustively for different values of $H/W$ and $2w_{E}/W$.
	}
	\label{si:tab:exhaustive:itau}
\vspace{\baselineskip}
\begin{lstlisting}[basicstyle={\ttfamily \scriptsize}]
# simulation data used in [Fig. 7a, GuajardoYevenes2013sep]
# https://doi.org/10.1016/j.jelechem.2013.07.014
# https://arxiv.org/abs/1602.06428
#
# values of 2/pi^2 K'(k_rho)/K(k_rho)_sim
# horizontal: H/W = [0.2/pi, 0.4/pi, 0.6/pi, 0.8/pi,
#                      1/pi,   2/pi,   3/pi,   4/pi, 5/pi]
# vertical: 2wE/W = [0.2, 0.4, 0.6, 0.8]
#
0.01414	0.02652	0.03628	0.04353	0.04873	0.05864	0.06007	0.06026	0.06029
0.01845	0.03458	0.04792	0.05844	0.06636	0.08236	0.08475	0.08508	0.08513
0.0265	0.04817	0.06546	0.07902	0.08932	0.11050	0.11373	0.11417	0.11424
0.04729	0.07909	0.10152	0.11807	0.13031	0.15522	0.15926	0.15954	0.15962
\end{lstlisting}
	\caption{
		Normalized currents $2/\pi^{2}\, K'(k_{\rho})/K(k_{\rho})_{\simu}$
		obtained in \cite[Fig. 7a]{GuajardoYevenes2013sep} by
		simulating exaustively for different values of $H/W$ and $2w_{E}/W$.
	}
	\label{si:tab:exaustive:itau:Guajardo2013}
\end{table}

Table \ref{si:tab:exaustive:itau:Guajardo2013} shows exhaustive simulations of
the normalized limiting current (case of internal counter electrode)
obtained in \cite[Fig. 7]{GuajardoYevenes2013sep}.
In order to compare the results in Table \ref{si:tab:exaustive:itau:Guajardo2013}
against the normalized current of this text (Eq. \eqref{main:eqn:if})
\begin{equation}
	\label{si:eqn:if_lim:intC}
	\left| \frac{
		i_{f}^{E}/L
	}{
		F n_{e} \cdot 2 D_{\lambda}\bar{c}_{\lambda,i}^{\whole}
	} \right|
	= 2\frac{K'(k_{\rho})}{K(k_{\rho})}
\end{equation}
($D_{\lambda}|c_{\lambda,f}^{E} - c_{\lambda,f}^{E'}|
= 2 D_{\lambda} \bar{c}_{\lambda,i}^{\whole}$
for the limiting current case, see Eq. \eqref{main:eqn:cE-cE':cotas})
the average flux \cite[$\bar{\phi}_{\lambda}^{\lim}$]{GuajardoYevenes2013sep}
was related to the current per electrode band
$i_{f}^{E} = F n_{e}\, \bar{\phi}_{\lambda}^{\lim}\, 2w_{E} L$,
and then replaced in the expression of \cite[Fig. 7a]{GuajardoYevenes2013sep}
\begin{equation}
	\left| \frac{
		\bar{\phi}_{\lambda}^{\lim} w_{E} 
	}{
		\pi^{2} D_{\lambda} \bar{c}_{\lambda,i}^{\whole}
	} \right|
	= \frac{1}{\pi^{2}}
	\left| \frac{
		i_{f}^{E}/L
	}{
		F n_{e} \cdot 2 D_{\lambda} \bar{c}_{\lambda,i}^{\whole}
	} \right|
	= \frac{1}{\pi^{2}} \cdot 2\frac{K'(k_{\rho})}{K(k_{\rho})}
\end{equation}
showing that both normalized currents are related by a factor of $\pi^{2} \approx 10$.
After scaling by this factor, one can see the agreements between Tables
\ref{si:tab:exhaustive:itau} and \ref{si:tab:exaustive:itau:Guajardo2013}.
See Fig. \ref{main:fig:K'Krho:simu} for a graphical representation of both tables.

\begin{table}[t]
\begin{lstlisting}[basicstyle={\ttfamily \scriptsize}]
# 16 points per curve
# simulation data obtained from [Fig. 7a, Strutwolf2005feb]
# https://doi.org/10.1002/elan.200403112
#
# for 2wE = 0.5 g       | for 2wE = g
# 2wE/H      Gss: 2K'/K | 2wE/H       Gss: 2K'/K
  0.25932504  0.7799733   0.25932504  0.99572764
  0.49023091  0.76608812  0.49023091  0.98931909
  0.72468917  0.71268358  0.72468917  0.97222964
  0.95914742  0.64218959  0.95914742  0.93164219
  1.19005329  0.57276368  1.19360568  0.87503338
  1.42095915  0.51401869  1.42451155  0.81201602
  1.65541741  0.46275033  1.6589698   0.75113485
  1.88987567  0.41895861  1.89342806  0.69666222
  2.12433393  0.38157543  2.12788632  0.64966622
  2.35523979  0.35166889  2.35879218  0.60694259
  2.58969805  0.32603471  2.59325044  0.56849132
  2.82415631  0.30253672  2.8277087   0.53431242
  3.05861456  0.28331108  3.06216696  0.50226969
  3.28952043  0.26515354  3.29307282  0.47449933
  3.52397869  0.25020027  3.52753108  0.44779706
  3.75843694  0.23631509  3.76198934  0.42643525
\end{lstlisting}
	\caption{
		Normalized currents $2K'(k_{\rho})/K(k_{\rho})_{\simu}$
		obtained in \cite[Fig. 7a]{Strutwolf2005feb} by 
		simulating for different values of $2w_{E}/H$ and
		$2w_{E}/g \in \mathopen\{\num{0,5},\, 1 \mathclose\}$.
	}
	\label{si:tab:exaustive:itau:Strutwolf2005}
\end{table}

Table \ref{si:tab:exaustive:itau:Strutwolf2005} shows simulations of
the normalized limiting current (case of external counter electrode)
obtained in \cite[Fig. 7a]{Strutwolf2005feb}.
In order to compare against the normalized current of this text
(Eq. \eqref{main:eqn:if})
\begin{equation}
	\label{si:eqn:if_lim:extC}
	\left| \frac{
		i_{f}^{E}/L
	}{
		F n_{e} \cdot [D_{\sigma} \bar{c}_{\sigma,i}^{\whole} + D_{\sigma'} \bar{c}_{\sigma',i}^{\whole}]
	} \right|
	= 2\frac{K'(k_{\rho})}{K(k_{\rho})}
\end{equation}
($D_{\sigma}|c_{\sigma,f}^{E} - c_{\sigma,f}^{E'}|
= D_{\sigma} \bar{c}_{\sigma,i}^{\whole}
+ D_{\sigma'} \bar{c}_{\sigma',i}^{\whole}$,
see Eq. \eqref{main:eqn:cE-cE':extC} when $\pm \eta^{E} \to +\infty$
for the limiting current case)
the expressions for \cite[$G_{g}$ and $G_{c}$ in p. 171]{Strutwolf2005feb}
were rewritten using the notation of this text
\begin{equation}
	|G_{E}| = \left| \frac{
		N_{E} i_{f}^{E}
	}{
		L N_{E} n_{e} F [D_{\sigma} \bar{c}_{\sigma,i}^{\whole} + D_{\sigma'} \bar{c}_{\sigma',i}^{\whole}]
	} \right|
	= 2\frac{K'(k_{\rho})}{K(k_{\rho})}
\end{equation}
were the normalized current at the generator $G_{g}$ and collector $G_{c}$
have equal magnitude (but opposite sign) in steady state $|G_{ss}| = |G_{g}| = |G_{c}| =: |G_{E}|$.
For plotting the data of \cite[Fig. 7a]{Strutwolf2005feb} in Fig. \ref{main:fig:K'Krho:simu},
the horizontal variable \cite[$w_{e}/h_{c}$ in Fig. 7a]{Strutwolf2005feb}
was also rewritten according to the notation of this text,
that is $2w_{E}/H$ in Table \ref{si:tab:exaustive:itau:Strutwolf2005},
and later related to $H/W$ by
\begin{equation}
	\frac{H}{W} = \left[ 
		\frac{2w_{E}}{H} \cdot \left( \frac{g}{2w_{E}} + \frac{2w_{E}}{2w_{E}} \right)
	\right]^{-1}
\end{equation}
where $g$ is the gap between electrodes of equal width,
such that $W = 2w_{E} + g$.
Only the data for the case $2w_{E} = g$ ($2w_{E}/W =\num{0,5}$)
was plotted in Fig. \ref{main:fig:K'Krho:simu},
which agrees with its theoretical counterpart.

\subsection{Direct evaluations and curve fitting}

\subsubsection{Voltammogram and potentials in Figs. \ref{main:fig:cE-cE':intC:noR} and \ref{main:fig:etaE=DetaE}}
\label{si:sec:voltammogram:intC:noR}

Fig. \ref{main:fig:cE-cE':intC:noR} was generated from Eq. \eqref{main:eqn:etaE-etaE':intC}
with $r_{\lambda} = (D_{\lambda'} \bar{c}_{\lambda',i}^{\whole})
/(D_{\lambda} \bar{c}_{\lambda,i}^{\whole})$
\begin{equation}
	\pm (\eta_{f}^{E} - \eta_{f}^{E'})
	= 2 \arctanh\left(
		\frac{
			c_{\lambda,f}^{E} - c_{\lambda,f}^{E'}
		}{
			2\bar{c}_{\lambda,i}^{\whole}
		}
	\right)
	+ 2 \arctanh\left(
		\frac{1}{r_{\lambda}}
		\cdot \frac{
			c_{\lambda,f}^{E} - c_{\lambda,f}^{E'}
		}{
			2\bar{c}_{\lambda,i}^{\whole}
		}
	\right)
\end{equation}
by evaluating several points for $[c_{\lambda,f}^{E} - c_{\lambda,f}^{E'}] /2 \bar{c}_{\lambda,f}^{\whole} \in \mathopen]-1,1\mathclose[$,
and obtaining a list of points for $\pm (\eta_{f}^{E} - \eta_{f}^{E'})$.

The same list of points $[c_{\lambda,f}^{E} - c_{\lambda,f}^{E'}] /2 \bar{c}_{\lambda,f}^{\whole} \in \mathopen]-1,1\mathclose[$ was evaluated in
Eqs. \eqref{si:eqn:etaE}
\begin{subequations}
	\begin{align}
		\e^{\mp (\eta_{f}^{E} - \eta_{\nul})}
		&= \frac{
			1 - r_{\sigma}^{-1}
			[c_{\sigma,f}^{E} - \bar{c}_{\sigma,i}^{\whole}]
			/\bar{c}_{\sigma,i}^{\whole}
		}{
			1 + \phantom{r_{\sigma}^{-1}}
			[c_{\sigma,f}^{E} - \bar{c}_{\sigma,i}^{\whole}]
			/\bar{c}_{\sigma,i}^{\whole}
		}
		\\
		\e^{\mp (\eta_{f}^{E'} - \eta_{\nul})}
		&= \frac{
			1 + r_{\sigma}^{-1}
			[c_{\sigma,f}^{E} - \bar{c}_{\sigma,i}^{\whole}]
			/\bar{c}_{\sigma,i}^{\whole}
		}{
			1 - \phantom{r_{\sigma}^{-1}}
			[c_{\sigma,f}^{E} - \bar{c}_{\sigma,i}^{\whole}]
			/\bar{c}_{\sigma,i}^{\whole}
		}
	\end{align}
\end{subequations}
obtaining one list of points for $\mp (\eta_{f}^{E} - \eta_{\nul})$ and
another list of points for $\mp (\eta_{f}^{E'} - \eta_{\nul})$.
Both of them were plotted against the list of points for $\pm (\eta_{f}^{E} - \eta_{f}^{E'})$ resulting in Fig. \ref{main:fig:etaE=DetaE}.

\subsubsection{Curve fitting in Fig. \ref{main:fig:fitting}}
\label{si:sec:fitting}

\begin{table}
	\centering
\begin{lstlisting}[basicstyle={\ttfamily \scriptsize}]
# 16 points per curve
# experimental data obtained from [Fig. 7, Aoki1988dec]
# https://doi.org/10.1016/0022-0728(88)87003-7
#
#  E: collector            | E: generator
#  VfE - VR    NE ifE      | VfE - VR    NE ifE
#  /V          /uA         | /V          /uA
   0.2037931    0.51457195   0.62310345  34.67668488
   0.26724138   0.37340619   0.53758621  33.9708561 
   0.29482759  -0.33242259   0.51551724  32.7003643 
   0.32241379  -1.17941712   0.48241379  30.30054645
   0.35827586  -2.87340619   0.44931034  27.75956284
   0.37758621  -6.12021858   0.43827586  25.35974499
   0.39137931  -9.36703097   0.44103448  22.95992714
   0.40517241 -12.75500911   0.43275862  20.56010929
   0.41896552 -16.00182149   0.41068966  18.01912568
   0.43       -19.24863388   0.40241379  15.61930783
   0.44655172 -22.63661202   0.40793103  13.21948998
   0.46310345 -25.88342441   0.39965517  10.81967213
   0.49068966 -28.70673953   0.37482759   8.27868852
   0.50724138 -29.97723133   0.3637931    5.87887067
   0.55965517 -30.82422587   0.36655172   3.47905282
   0.59       -31.38888889   0.30586207   1.22040073
\end{lstlisting}
	\caption{
		Voltammogram data of \SI{1}{\nano\mole\per\micro\litre}
		ferrocene with $V_{f}^{E'} - V^{R} = \SI{-0,15}{\volt}$
		from \cite[Fig. 7]{Aoki1988dec} for generator and collector.
		Here $V^{R}$ is the potential of the reference electrode
		(saturated calomel electrode).
	}
	\label{si:tab:Aoki1988}
\vspace{\baselineskip}
	\centering
\begin{lstlisting}[basicstyle={\ttfamily \scriptsize}]
# 16 points per curve
# experimental data obtained from [Fig. 2, Rahimi2011]
# https://doi.org/10.1021/ac2012703
# or equivalently from [Fig. 3.12, Rahimi2009]
# https://hdl.handle.net/10012/4555
#
# VfE - VfE'  NE ifE
# /V          /uA
 -0.39557823  0.5061263 
 -0.35578231  0.49877474
 -0.2914966   0.49255419
 -0.20612245  0.4801131 
 -0.13163265  0.42808671
 -0.07278912  0.29406221
 -0.03979592  0.16512724
 -0.02006803  0.08369463
  0.00170068 -0.02035815
  0.02380952 -0.12441093
  0.0452381  -0.22846371
  0.07653061 -0.33251649
  0.1329932  -0.43317625
  0.18265306 -0.47502356
  0.27823129 -0.48633365
  0.29489796 -0.49198869
\end{lstlisting}
	\caption{
		Voltammogram data of \SI{0,20}{\nano\mole\per\micro\litre}
		of each ferrocyanide and ferricyanide from \cite[Fig. 2]{Rahimi2011},
		or equivalently from \cite[Fig. 3.12]{Rahimi2009aug}.
	}
	\label{si:tab:Rahimi2011}
\end{table}

The data used for curve fitting in Fig. \ref{main:fig:fitting:extC}
was obtained from \cite[Fig. 7]{Aoki1988dec},
from which 16 data points for each curve (generator and collector) were extracted.
The extracted data is shown in Table \ref{si:tab:Aoki1988}.

Similarly, the data used for curve fitting in Fig. \ref{main:fig:fitting:intC}
was obtained from \cite[Fig. 2]{Rahimi2011}, or equivalently \cite[3.12]{Rahimi2009aug},
from which 16 data points were extracted.
The extracted data is shown in Table \ref{si:tab:Rahimi2011}.

\subsubsection{Approximated normalized current in Figs. \ref{main:fig:K'Krho:HW_tall} and \ref{main:fig:K'Krho:HW_shallow}.}
\label{si:sec:if_approx}

The relative error between the approximated and exact currents,
$\tilde{i}_{f}^{E}$ and $i_{f}^{E}$, corresponds to that between its normalized counterparts
\begin{equation}
	\text{relative error} =
	\frac{\tilde{i}_{f}^{E}}{i_{f}^{E}} - 1 =
	\left[2 \frac{K'(k_{\rho})}{K(k_{\rho})} \right]_{\text{app.}}
	\left[2 \frac{K'(k_{\rho})}{K(k_{\rho})} \right]^{-1} - 1
\end{equation}
For this reason, it suffices for Figs. \ref{main:fig:K'Krho:HW_tall}
and \ref{main:fig:K'Krho:HW_shallow} to show the approximated version of
the normalized current $2 K'(k_{\rho})/K(k_{\rho})$,
together with its relative error with respect to the exact value,
for tall and shallow cells respectively.

The combinations of $(2w_{E}/H,\, H/W)$ that produce a relative error
less than $\num{+-0,05} = \SI{+-5}{\percent}$ were considered to define
regions where the approximations are valid.
These correspond to $2w_{E}/W > \num{0,46}$ (large electrodes) and
$2w_{E}/W < \num{0,56}$ (small electrodes) for the case of tall cells $H/W > 1$.
In the case of shallow cells $H/W \leq 1$, these regions of validity correspond
to the right of the line connecting $(1,0) \to (\num{0,28},1)$
(large electrodes) and to the left of the line connecting
$(1,0) \to (\num{0,36},1)$ (small electrodes).
See \emph{Output of Figs. \ref{main:fig:K'Krho:HW_tall} and
\ref{main:fig:K'Krho:HW_shallow}} at the following subsection for numerical results.
The combinations of $(2w_{E}/H,\, H/W)$ that produce a relative error
of $\num{+-0,005} = \SI{+-0,5}{\percent}$ are also given as a reference in
subsection \emph{Output of Figs. \ref{main:fig:K'Krho:HW_tall} and \ref{main:fig:K'Krho:HW_shallow}}.

When expressing the ratio $K'(k)/K(k)$ (which is involved in the relative error)
in terms of the nome function $Q(k)$, or equivalently $Q(k')$
\begin{equation}
	\frac{K'(k)}{K(k)} = \frac{\ln Q(k)}{-\pi} = \frac{-\pi}{\ln Q(k')}
\end{equation}
one can realize that the approximations in Eqs. \eqref{main:eqn:K'K:Hinf} and \eqref{main:eqn:K'K:H0}
\begin{equation}
	\left[\frac{K'(k_{\rho})}{K(k_{\rho})} \right]_{\text{app.}} :=
	\begin{cases}
		\frac{
			\displaystyle \ln\tilde{Q}(\tilde{k}_{\rho})
		}{\displaystyle -\pi}, & \text{for } k_{\rho} \approx 0
		\\[1em]
		\frac{\displaystyle -\pi}{
			\displaystyle \ln\tilde{Q}(\tilde{k}_{\rho}')
		}, & \text{for } k_{\rho}' \approx 0
	\end{cases}
\end{equation}
are due to a two-step approximation process:
First, approximation of the nome function by
\begin{equation}
	\tilde{Q}(k) := \frac{k^{2}}{16}
\end{equation}
and later, the approximations of the moduli $\tilde{k}_{\rho}$
and $\tilde{k}_{\rho}'$, which are given in Eqs. \eqref{main:eqn:moduli:Hinf} and \eqref{main:eqn:moduli:H0} in the main text.

Therefore, it is possible to analyze the propagation of the relative error,
for $k_{\rho} \approx 0$, by looking at the following total ratio
\begin{equation}
	\label{si:eqn:ratio:k_rho}
	\left[\frac{K'(k_{\rho})}{K(k_{\rho})} \right]_{\text{app.}}
	\left[\frac{K'(k_{\rho})}{K(k_{\rho})} \right]^{-1} =
	\frac{\ln\tilde{Q}(\tilde{k}_{\rho})}{\ln Q(k_{\rho})} =
	\frac{\ln\tilde{Q}(\tilde{k}_{\rho})}{\ln\tilde{Q}(k_{\rho})}
	\frac{\ln\tilde{Q}(k_{\rho})}{\ln Q(k_{\rho})}
\end{equation}
where a relative error of $\num{+-0,05} = \SI{+-5}{\percent}$ corresponds to
the total ratios \num{0,95} and \num{1,05}.
The ratios $\ln\tilde{Q}(k_{\rho})/\ln Q(k_{\rho})$ and
$\ln\tilde{Q}(\tilde{k}_{\rho})/\ln\tilde{Q}(k_{\rho})$ show
the individual contributions of the approximated nome $\tilde{Q}$ and
the approximated modulus $\tilde{k}_{\rho}$ to
the total ratio $\ln\tilde{Q}(\tilde{k}_{\rho})/\ln Q(k_{\rho})$,
and therefore, to the relative error.

\begin{figure}[t]
	\centering
	\includegraphics{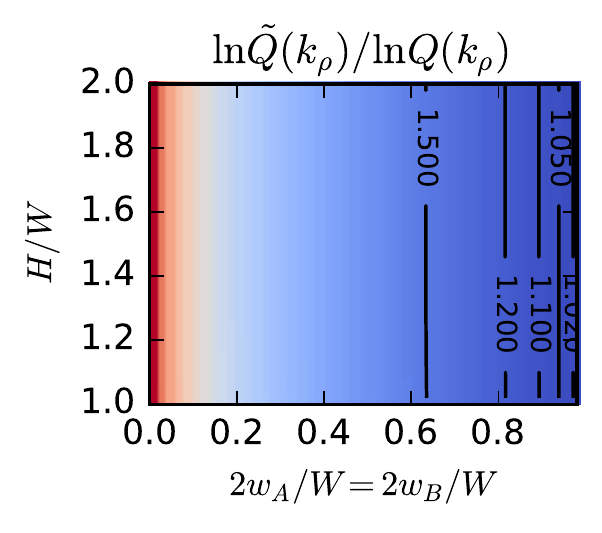} \quad
	\includegraphics{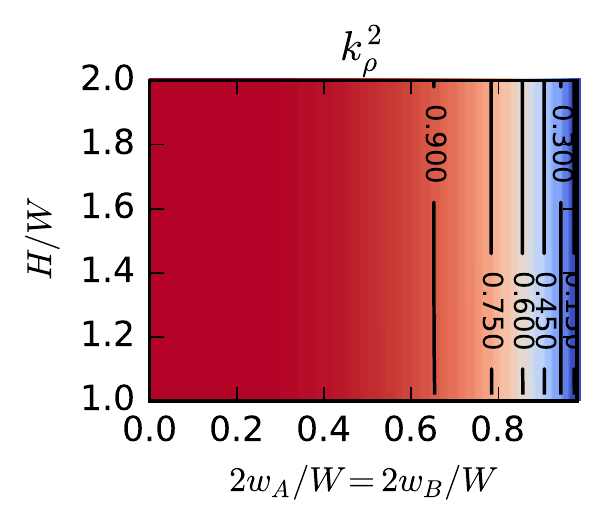}

	\includegraphics{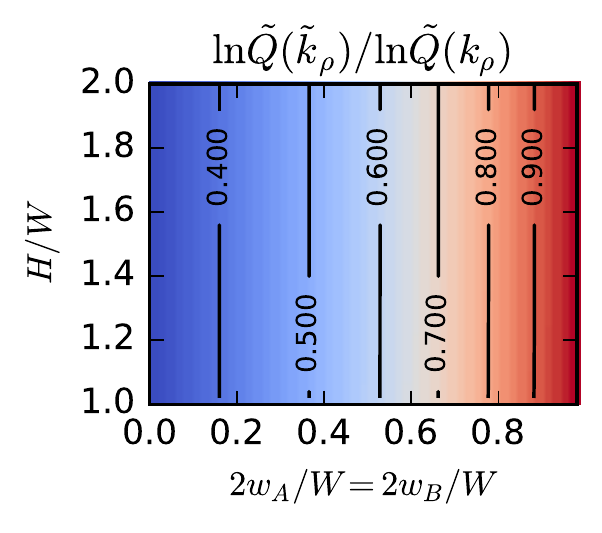} \quad
	\includegraphics{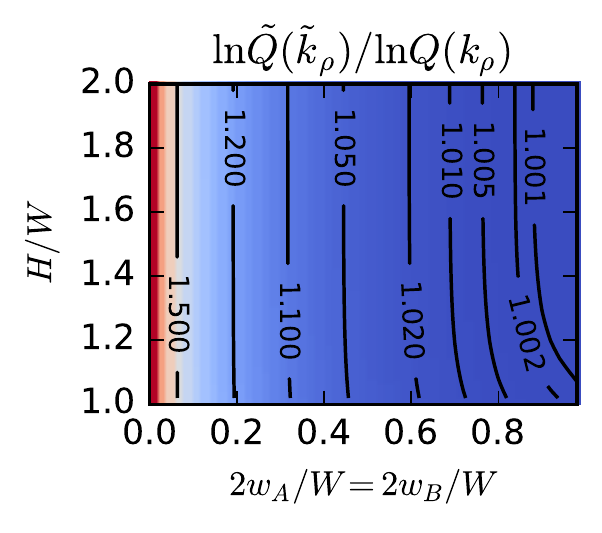}
	\caption{
		Propapation of the ratio of approximation
		$\ln\tilde{Q}(\tilde{k}_{\rho})/\ln Q(k_{\rho})$
		for the case of tall cells and large electrodes.
		Contour lines of selected values are shown in black.
	}
	\label{si:fig:ratios:tall:large}
\end{figure}

This is shown in Fig. \ref{si:fig:ratios:tall:large} for the case of tall cells and large electrodes.
The first ratio $\ln\tilde{Q}(k_{\rho})/\ln Q(k_{\rho}) < \num{1,05}$
(relative error of \SI{5}{\percent}) for widths $2w_{E}/W > \num{0,94}$
($k_{\rho}^{2} \lesssim \num{0,314}$),
which corresponds to very extreme cases of wide electrodes.
This agrees with the relative error of $\approx \SI{4}{\percent}$
in \cite[before Eq. (32)]{Aoki1988dec} when approximating only the nome function
$\tilde{Q}(k_{\rho})$ for gap widths
$g/W < \num{0,059} \Leftrightarrow 2w_{E}/W > \num{0,941}$
($k_{\rho}^{2} < \num{0,310}$).
Fortunately, the contribution due to the approximation of
the modulus $\tilde{k}_{\rho}$ makes the second ratio
$\ln\tilde{Q}(\tilde{k}_{\rho})/\ln\tilde{Q}(k_{\rho}) < 1$
in almost the whole domain $(2w_{E}/W, H/W)$.
This helps to extend the region where the total ratio
$\ln\tilde{Q}(\tilde{k}_{\rho})/\ln Q(k_{\rho}) < \num{1,05}$
(relative error less than \SI{5}{\percent}),
which corresponds to more reasonable electrode widths $2w_{E}/W < \num{0,46}$.

A similar analysis of the propagation of the relative error can be done for $k_{\rho}' \approx 0$ when considering the following total ratio
\begin{equation}
	\label{si:eqn:ratio:k1_rho}
	\left[\frac{K'(k_{\rho})}{K(k_{\rho})} \right]_{\text{app.}}
	\left[\frac{K'(k_{\rho})}{K(k_{\rho})} \right]^{-1} =
	\frac{\ln Q(k_{\rho}')}{\ln\tilde{Q}(\tilde{k}_{\rho}')} =
	\frac{\ln\tilde{Q}(k_{\rho}')}{\ln\tilde{Q}(\tilde{k}_{\rho}')}
	\frac{\ln Q(k_{\rho}')}{\ln\tilde{Q}(k_{\rho}')}
\end{equation}
where the ratios $\ln Q(k_{\rho}')/\ln\tilde{Q}(k_{\rho}')$ and
$\ln\tilde{Q}(k_{\rho}')/\ln\tilde{Q}(\tilde{k}_{\rho}')$ show
the individual contributions of the approximated nome $\tilde{Q}$ and
the approximated complementary modulus $\tilde{k}_{\rho}'$ to
the total ratio $\ln Q(k_{\rho}')/\ln\tilde{Q}(\tilde{k}_{\rho}')$,
and therefore, to the relative error.

\begin{figure}[t]
	\centering
	\includegraphics{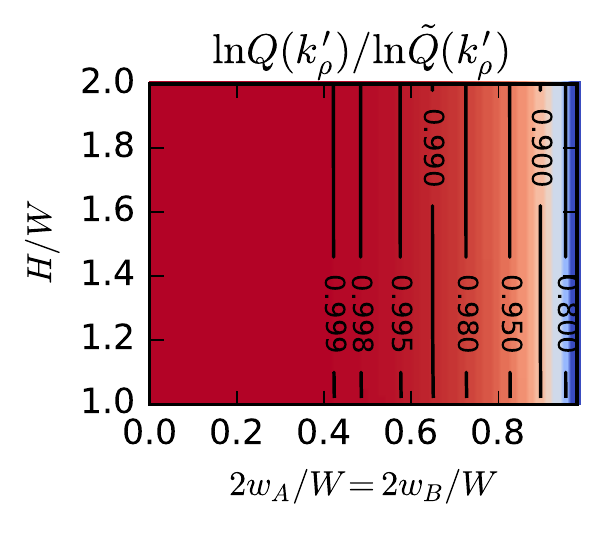} \quad
	\includegraphics{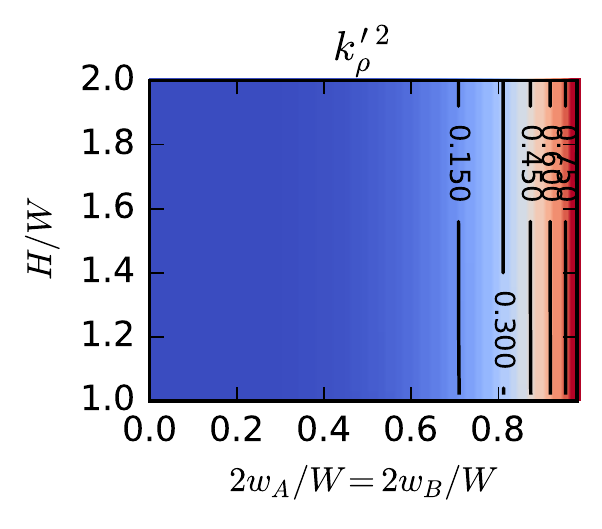}

	\includegraphics{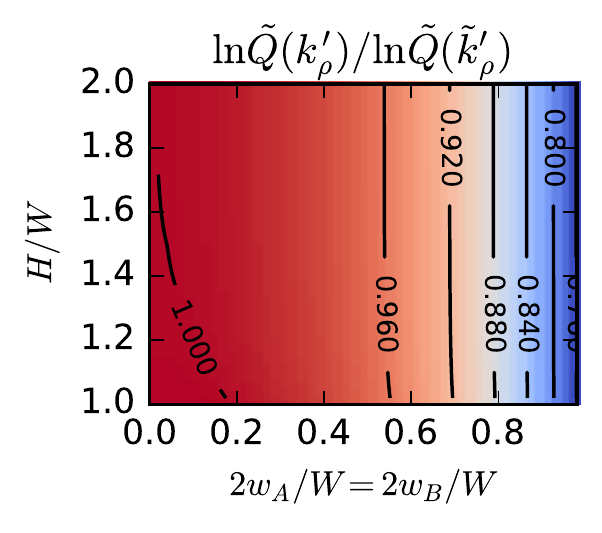} \quad
	\includegraphics{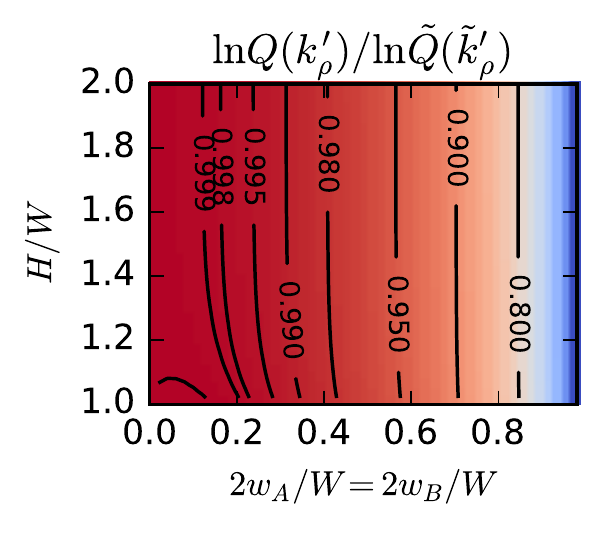}
	\caption{
		Propapation of the ratio of approximation
		$\ln\tilde{Q}(\tilde{k}_{\rho})/\ln Q(k_{\rho})$
		for the case of tall cells and small electrodes.
		Contour lines of selected values are shown in black.
	}
	\label{si:fig:ratios:tall:small}
\end{figure}

This is shown in Fig. \ref{si:fig:ratios:tall:small} for the case of tall cells and small electrodes.
Unlike the previous case, the first ratio
$\ln Q(k_{\rho}')/\ln\tilde{Q}(k_{\rho}') > \num{0,95}$
(relative error of \SI{-5}{\percent})
for widths $2w_{E}/W < \num{0,82}$ (${k_{\rho}'}^{2} \lesssim \num{0,315}$),
which corresponds to a large portion of the parameter domain.
This agrees with the statement in \cite[before Eq. (5)]{Morf2006may},
which says that the approximation converges rapidly for gap widths
$g > \num{0,1}\cdot 2w_{E} \Leftrightarrow 2w_{E}/W < \num{0,909}$
(no quantitative error is specified).
Unfortunately, the contribution due to the approximation of
the complementary modulus $\tilde{k}_{\rho}'$ makes the second ratio
$\ln\tilde{Q}(k_{\rho}')/\ln\tilde{Q}(\tilde{k}_{\rho}') < 1$
in most of the domain $(2w_{E}/W, H/W)$.
This approximation helps to simplify the expression for the complementary
modulus in Eq. \eqref{main:eqn:krho':Hinf} at the expense of shrinking
the domain where $\ln Q(k_{\rho}')/\ln\tilde{Q}(\tilde{k}_{\rho}') > \num{0,95}$
(relative error of approximation less than \SI{-5}{\percent}),
which corresponds to electrode widths $2w_{E}/W \lesssim \num{0,56}$.

The propagation of the relative error for cases of shallow cells
have not been reported previously in the literature and
they are given below as a reference.

\begin{figure}
	\centering
	\includegraphics{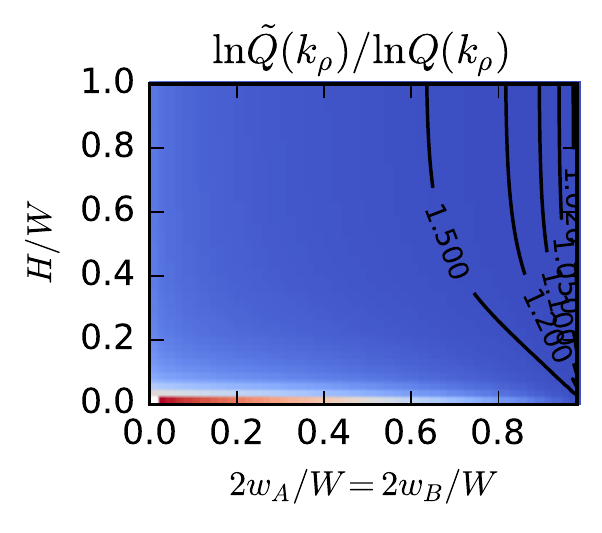} \quad
	\includegraphics{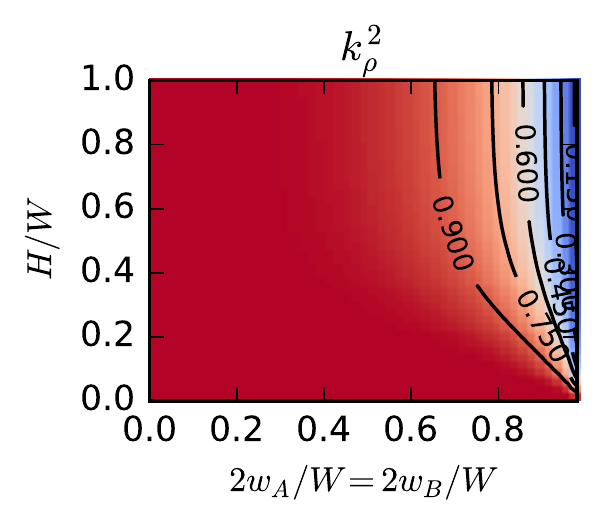}

	\includegraphics{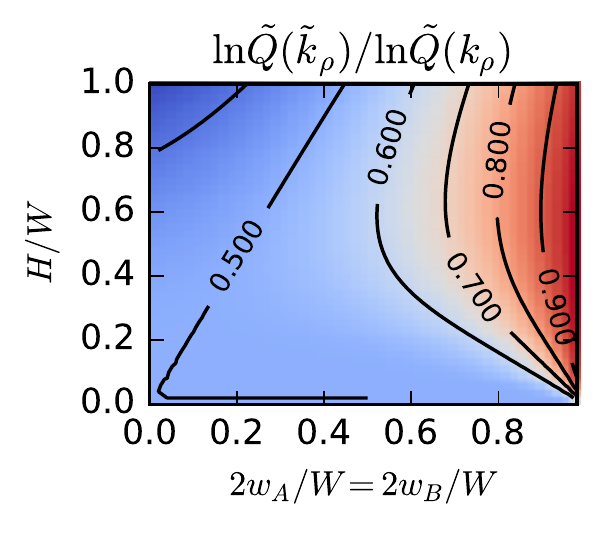} \quad
	\includegraphics{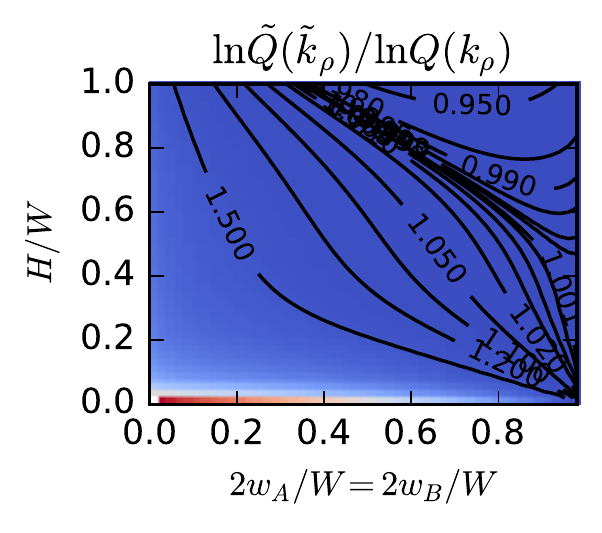}
	\caption{
		Propapation of the ratio of approximation
		$\ln\tilde{Q}(\tilde{k}_{\rho})/\ln Q(k_{\rho})$
		for the case of shallow cells and large electrodes.
		Contour lines of selected values are shown in black.
	}
	\label{si:fig:ratios:shallow:large}
\end{figure}

Fig. \ref{si:fig:ratios:shallow:large} shows the propagation of the total ratio,
in Eq. \eqref{si:eqn:ratio:k_rho}, for the case of shallow cells and large electrodes.
The first ratio $\ln\tilde{Q}(k_{\rho})/\ln Q(k_{\rho}) < \num{1,05}$
(relative error of \SI{5}{\percent}) for combinations of
$(2w_{E}/W, H/W)$ roughly to the right of the line joining the points
$(1,0) \to (\num{0,94}, 1)$ ($k_{\rho}^{2} \lesssim \num{0,316}$),
which again correspond to very wide electrodes.
The contribution due to the approximation of
the modulus $\tilde{k}_{\rho}$ makes the second ratio
$\ln\tilde{Q}(\tilde{k}_{\rho})/\ln\tilde{Q}(k_{\rho}) < 1$
in most of the domain $(2w_{E}/W, H/W)$.
This also extends the region where the total ratio
$\num{0,95} < \ln\tilde{Q}(\tilde{k}_{\rho})/\ln Q(k_{\rho}) < \num{1,05}$
(relative error less than \SI{+-5}{\percent}),
which corresponds to combinations of $(2w_{E}/W, H/W)$ approximately
at the right of the line joining the points $(1,0) \to (\num{0,28}, 1)$.

\begin{figure}[t]
	\centering
	\includegraphics{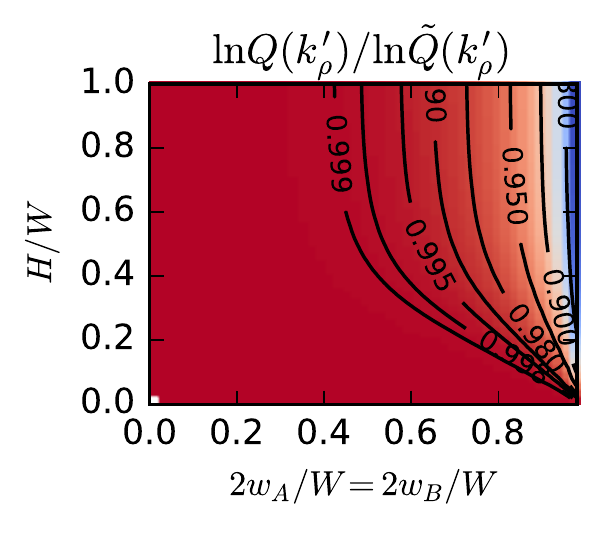} \quad
	\includegraphics{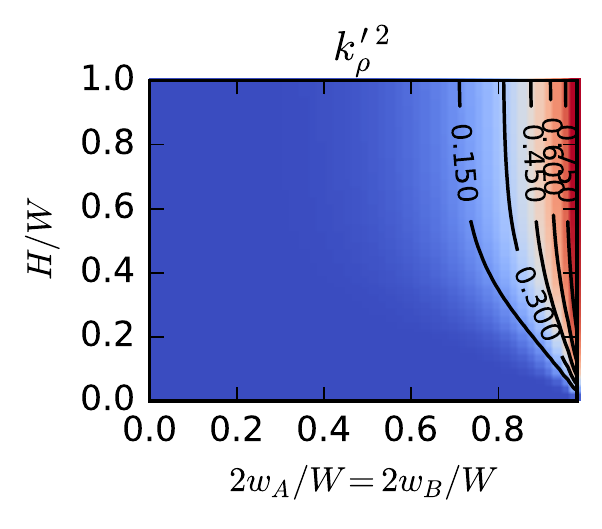}

	\includegraphics{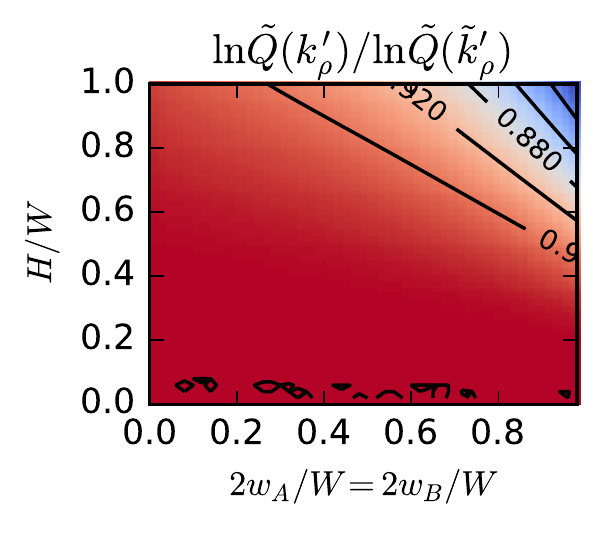} \quad
	\includegraphics{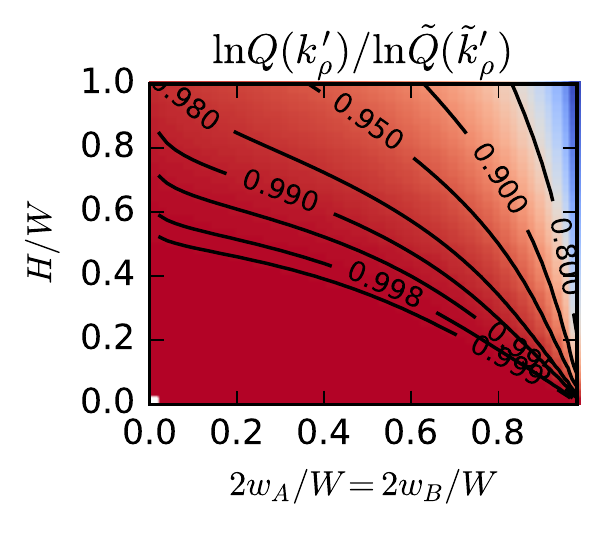}
	\caption{
		Propapation of the ratio of approximation
		$\ln\tilde{Q}(\tilde{k}_{\rho})/\ln Q(k_{\rho})$
		for the case of shallow cells and small electrodes.
		Contour lines of selected values are shown in black.
	}
	\label{si:fig:ratios:shallow:small}
\end{figure}

Fig. \ref{si:fig:ratios:shallow:small} shows the propagation of the total ratio,
in Eq. \eqref{si:eqn:ratio:k1_rho}, for the case of shallow cells and small electrodes.
The first ratio $\ln Q(k_{\rho}')/\ln\tilde{Q}(k_{\rho}') > \num{0,95}$
(relative error of \SI{-5}{\percent}) for combinations of $(2w_{E}/W, H/W)$
 roughly to the left of the line joining the points
$(1,0) \to (\num{0,82}, 1)$ (${k_{\rho}'}^{2} \lesssim \num{0,315}$),
which corresponds to a large portion of the parameter domain.
The approximation of the complementary modulus $\tilde{k}_{\rho}'$ makes
the second ratio $\ln\tilde{Q}(k_{\rho}')/\ln\tilde{Q}(\tilde{k}_{\rho}') < 1$
in most of the domain $(2w_{E}/W, H/W)$.
This approximation helps to simplify the expression for the complementary
modulus in Eq. \eqref{main:eqn:krho':Hinf} at the expense of shrinking
the domain where $\ln Q(k_{\rho}')/\ln\tilde{Q}(\tilde{k}_{\rho}') > \num{0,95}$
(relative error of approximation less than \SI{-5}{\percent}),
which corresponds to combinations of $(2w_{E}/W, H/W)$
approximately at the left of the line joining the points
$(1,0) \to (\num{0,36}, 1)$.

\paragraph{Output of Figs. \ref{main:fig:K'Krho:HW_tall} and \ref{main:fig:K'Krho:HW_shallow}}
\rule{0pt}{0pt}
\begin{lstlisting}[basicstyle={\ttfamily \scriptsize}]
mpmath:     0.19
numpy:      1.11.0
matplotlib: 1.5.1

H/W: tall, 2wE/W: large
------------------------------
error < +-5.0%: [  4.89246306e-02   4.39635532e-02   3.94513927e-02   3.53466127e-02
   3.16126327e-02   2.82170766e-02   2.51311534e-02   2.23291461e-02
   1.97879836e-02   1.74868819e-02   1.54070399e-02   1.35313795e-02
   1.18443217e-02   1.03315920e-02   8.98004862e-03   7.77752814e-03
   6.71270414e-03   5.77495248e-03   4.95421687e-03   4.24086520e-03
   3.62552072e-03   3.09883913e-03   2.65117058e-03   2.27196010e-03
   1.94846970e-03   1.66230242e-03   1.37515620e-03   4.57882292e-02
   4.08347283e-02   3.63337276e-02   3.22437476e-02   2.85282641e-02
   2.51549557e-02   2.20950872e-02   1.93229978e-02   1.68156760e-02
   1.45524017e-02   1.25144444e-02   1.06848066e-02   9.04800381e-03
   7.58987477e-03   6.29741706e-03   5.15864251e-03   4.16244923e-03
   3.29850683e-03   2.55715139e-03   1.92928748e-03   1.40629337e-03
   9.79925362e-04   6.42214866e-04   3.85347603e-04   2.01504085e-04
   8.26116380e-05   1.98491161e-05]
H/W:   [ 1.02  1.02  1.02  1.02  1.02  1.02  1.02  1.02  1.02  1.02  1.02  1.02
  1.02  1.02  1.02  1.02  1.02  1.02  1.02  1.02  1.02  1.02  1.02  1.02
  1.02  1.02  1.02  2.    2.    2.    2.    2.    2.    2.    2.    2.    2.
  2.    2.    2.    2.    2.    2.    2.    2.    2.    2.    2.    2.    2.
  2.    2.    2.    2.  ]
2wE/W: [ 0.46  0.48  0.5   0.52  0.54  0.56  0.58  0.6   0.62  0.64  0.66  0.68
  0.7   0.72  0.74  0.76  0.78  0.8   0.82  0.84  0.86  0.88  0.9   0.92
  0.94  0.96  0.98  0.46  0.48  0.5   0.52  0.54  0.56  0.58  0.6   0.62
  0.64  0.66  0.68  0.7   0.72  0.74  0.76  0.78  0.8   0.82  0.84  0.86
  0.88  0.9   0.92  0.94  0.96  0.98]

error < +-0.5%: [  4.95421687e-03   4.24086520e-03   3.62552072e-03   3.09883913e-03
   2.65117058e-03   2.27196010e-03   1.94846970e-03   1.66230242e-03
   1.37515620e-03   4.16244923e-03   3.29850683e-03   2.55715139e-03
   1.92928748e-03   1.40629337e-03   9.79925362e-04   6.42214866e-04
   3.85347603e-04   2.01504085e-04   8.26116380e-05   1.98491161e-05]
H/W:   [ 1.02  1.02  1.02  1.02  1.02  1.02  1.02  1.02  1.02  2.    2.    2.    2.
  2.    2.    2.    2.    2.    2.    2.  ]
2wE/W: [ 0.82  0.84  0.86  0.88  0.9   0.92  0.94  0.96  0.98  0.78  0.8   0.82
  0.84  0.86  0.88  0.9   0.92  0.94  0.96  0.98]

H/W: tall, 2wE/W: small
------------------------------
error < +-5.0%: [  1.34379841e-03   1.50625142e-03   1.55492737e-03   1.51063107e-03
   1.37568637e-03   1.14727975e-03   8.20469773e-04   3.89165578e-04
  -1.53497014e-04  -8.15004047e-04  -1.60341016e-03  -2.52733424e-03
  -3.59598269e-03  -4.81919075e-03  -6.20747734e-03  -7.77211171e-03
  -9.52519160e-03  -1.14797334e-02  -1.36497758e-02  -1.60504988e-02
  -1.86983601e-02  -2.16112531e-02  -2.48086897e-02  -2.83120129e-02
  -3.21446459e-02  -3.63323863e-02  -4.09037545e-02  -4.58904100e-02
  -1.40953695e-05  -7.59012833e-05  -1.94053115e-04  -3.77200484e-04
  -6.33290716e-04  -9.69995216e-04  -1.39493189e-03  -1.91580639e-03
  -2.54051530e-03  -3.27723003e-03  -4.13447098e-03  -5.12117737e-03
  -6.24677632e-03  -7.52125377e-03  -8.95522915e-03  -1.05600358e-02
  -1.23478090e-02  -1.43315832e-02  -1.65254008e-02  -1.89444353e-02
  -2.16051310e-02  -2.45253626e-02  -2.77246203e-02  -3.12242233e-02
  -3.50475704e-02  -3.92204342e-02  -4.37713099e-02  -4.87318322e-02]
H/W:   [ 1.02  1.02  1.02  1.02  1.02  1.02  1.02  1.02  1.02  1.02  1.02  1.02
  1.02  1.02  1.02  1.02  1.02  1.02  1.02  1.02  1.02  1.02  1.02  1.02
  1.02  1.02  1.02  1.02  2.    2.    2.    2.    2.    2.    2.    2.    2.
  2.    2.    2.    2.    2.    2.    2.    2.    2.    2.    2.    2.    2.
  2.    2.    2.    2.    2.    2.  ]
2wE/W: [ 0.02  0.04  0.06  0.08  0.1   0.12  0.14  0.16  0.18  0.2   0.22  0.24
  0.26  0.28  0.3   0.32  0.34  0.36  0.38  0.4   0.42  0.44  0.46  0.48
  0.5   0.52  0.54  0.56  0.02  0.04  0.06  0.08  0.1   0.12  0.14  0.16
  0.18  0.2   0.22  0.24  0.26  0.28  0.3   0.32  0.34  0.36  0.38  0.4
  0.42  0.44  0.46  0.48  0.5   0.52  0.54  0.56]

error < +-0.5%: [  1.34379841e-03   1.50625142e-03   1.55492737e-03   1.51063107e-03
   1.37568637e-03   1.14727975e-03   8.20469773e-04   3.89165578e-04
  -1.53497014e-04  -8.15004047e-04  -1.60341016e-03  -2.52733424e-03
  -3.59598269e-03  -4.81919075e-03  -1.40953695e-05  -7.59012833e-05
  -1.94053115e-04  -3.77200484e-04  -6.33290716e-04  -9.69995216e-04
  -1.39493189e-03  -1.91580639e-03  -2.54051530e-03  -3.27723003e-03
  -4.13447098e-03]
H/W:   [ 1.02  1.02  1.02  1.02  1.02  1.02  1.02  1.02  1.02  1.02  1.02  1.02
  1.02  1.02  2.    2.    2.    2.    2.    2.    2.    2.    2.    2.    2.  ]
2wE/W: [ 0.02  0.04  0.06  0.08  0.1   0.12  0.14  0.16  0.18  0.2   0.22  0.24
  0.26  0.28  0.02  0.04  0.06  0.08  0.1   0.12  0.14  0.16  0.18  0.2
  0.22]

\end{lstlisting}

\begin{lstlisting}[basicstyle={\ttfamily \scriptsize}]
H/W: shallow, 2wE/W: large
------------------------------
error < +-5.0%: [ 0.04266411  0.02856579  0.01599183  0.00475681 -0.00529221 -0.01428284
 -0.02232264 -0.02950306 -0.03590251 -0.04158867 -0.04662042 -0.04894622
 -0.04434866 -0.0380939 ]
H/W:   [ 1.  1.  1.  1.  1.  1.  1.  1.  1.  1.  1.  1.  1.  1.]
2wE/W: [ 0.28  0.3   0.32  0.34  0.36  0.38  0.4   0.42  0.44  0.46  0.48  0.94
  0.96  0.98]

error < +-0.5%: [ 0.00475681]
H/W:   [ 1.]
2wE/W: [ 0.34]

H/W: shallow, 2wE/W: small
------------------------------
error < +-5.0%: [  8.64649787e-04  -1.25293460e-04  -9.18026175e-07  -9.93122166e-08
   5.31949640e-09   3.78638454e-10  -1.31092914e-11  -2.44138043e-13
   1.04360964e-14  -4.44089210e-16  -2.22044605e-16  -1.11022302e-16
  -1.11022302e-16  -2.22044605e-16  -2.22044605e-16  -2.22044605e-16
  -1.11022302e-16  -2.22044605e-16  -4.44089210e-16  -2.22044605e-16
  -4.44089210e-16  -3.33066907e-16  -2.22044605e-16  -2.22044605e-16
  -2.22044605e-16   0.00000000e+00   0.00000000e+00  -1.11022302e-16
  -1.11022302e-16  -2.22044605e-16  -2.22044605e-16   0.00000000e+00
   0.00000000e+00  -4.44089210e-16  -2.22044605e-16  -2.22044605e-16
  -5.55111512e-16  -4.44089210e-16  -8.88178420e-16  -9.43689571e-15
  -2.18935980e-13  -5.68267655e-12  -1.50597090e-10  -4.07730671e-09
  -1.13676139e-07  -3.30817977e-06  -1.03186156e-04  -3.71923792e-03
  -1.71721015e-02  -2.00078218e-02  -2.21800755e-02  -2.40671766e-02
  -2.58058493e-02  -2.74645520e-02  -2.90842462e-02  -3.06925754e-02
  -3.23099728e-02  -3.39526894e-02  -3.56344602e-02  -3.73674984e-02
  -3.91631346e-02  -4.10322529e-02  -4.29856120e-02  -4.50340938e-02
  -4.71889128e-02  -4.94618008e-02]
H/W:   [ 0.02  0.02  0.02  0.02  0.02  0.02  0.02  0.02  0.02  0.02  0.02  0.02
  0.02  0.02  0.02  0.02  0.02  0.02  0.02  0.02  0.02  0.02  0.02  0.02
  0.02  0.02  0.02  0.02  0.02  0.02  0.02  0.02  0.02  0.02  0.02  0.02
  0.02  0.02  0.02  0.02  0.02  0.02  0.02  0.02  0.02  0.02  0.02  0.02
  1.    1.    1.    1.    1.    1.    1.    1.    1.    1.    1.    1.    1.
  1.    1.    1.    1.    1.  ]
2wE/W: [ 0.04  0.06  0.08  0.1   0.12  0.14  0.16  0.18  0.2   0.22  0.24  0.26
  0.28  0.3   0.32  0.34  0.36  0.38  0.4   0.42  0.44  0.46  0.48  0.5
  0.52  0.54  0.56  0.58  0.6   0.62  0.64  0.66  0.68  0.7   0.72  0.74
  0.76  0.78  0.8   0.82  0.84  0.86  0.88  0.9   0.92  0.94  0.96  0.98
  0.02  0.04  0.06  0.08  0.1   0.12  0.14  0.16  0.18  0.2   0.22  0.24
  0.26  0.28  0.3   0.32  0.34  0.36]

error < +-0.5%: [  8.64649787e-04  -1.25293460e-04  -9.18026175e-07  -9.93122166e-08
   5.31949640e-09   3.78638454e-10  -1.31092914e-11  -2.44138043e-13
   1.04360964e-14  -4.44089210e-16  -2.22044605e-16  -1.11022302e-16
  -1.11022302e-16  -2.22044605e-16  -2.22044605e-16  -2.22044605e-16
  -1.11022302e-16  -2.22044605e-16  -4.44089210e-16  -2.22044605e-16
  -4.44089210e-16  -3.33066907e-16  -2.22044605e-16  -2.22044605e-16
  -2.22044605e-16   0.00000000e+00   0.00000000e+00  -1.11022302e-16
  -1.11022302e-16  -2.22044605e-16  -2.22044605e-16   0.00000000e+00
   0.00000000e+00  -4.44089210e-16  -2.22044605e-16  -2.22044605e-16
  -5.55111512e-16  -4.44089210e-16  -8.88178420e-16  -9.43689571e-15
  -2.18935980e-13  -5.68267655e-12  -1.50597090e-10  -4.07730671e-09
  -1.13676139e-07  -3.30817977e-06  -1.03186156e-04  -3.71923792e-03]
H/W:   [ 0.02  0.02  0.02  0.02  0.02  0.02  0.02  0.02  0.02  0.02  0.02  0.02
  0.02  0.02  0.02  0.02  0.02  0.02  0.02  0.02  0.02  0.02  0.02  0.02
  0.02  0.02  0.02  0.02  0.02  0.02  0.02  0.02  0.02  0.02  0.02  0.02
  0.02  0.02  0.02  0.02  0.02  0.02  0.02  0.02  0.02  0.02  0.02  0.02]
2wE/W: [ 0.04  0.06  0.08  0.1   0.12  0.14  0.16  0.18  0.2   0.22  0.24  0.26
  0.28  0.3   0.32  0.34  0.36  0.38  0.4   0.42  0.44  0.46  0.48  0.5
  0.52  0.54  0.56  0.58  0.6   0.62  0.64  0.66  0.68  0.7   0.72  0.74
  0.76  0.78  0.8   0.82  0.84  0.86  0.88  0.9   0.92  0.94  0.96  0.98]

\end{lstlisting}

\paragraph{Output of Figs. \ref{si:fig:ratios:tall:large}--\ref{si:fig:ratios:shallow:small}}
\rule{0pt}{0pt}
\begin{lstlisting}[basicstyle={\ttfamily \scriptsize}]
mpmath:     0.19
numpy:      1.11.0
matplotlib: 1.5.1

H/W: tall, 2wE/W: large
------------------------------
0.95 <= first ratio <= 1.05: [ 1.04898787  1.02980031  1.01291437  1.04859517  1.02956362  1.01281356]
krho2: [ 0.3161621   0.22364877  0.11883769  0.314465    0.22236378  0.1181084 ]
H/W:   [ 1.02  1.02  1.02  2.    2.    2.  ]
2wE/W: [ 0.94  0.96  0.98  0.94  0.96  0.98]

0.95 <= total ratio <= 1.05: [ 1.04892463  1.04396355  1.03945139  1.03534661  1.03161263  1.02821708
  1.02513115  1.02232915  1.01978798  1.01748688  1.01540704  1.01353138
  1.01184432  1.01033159  1.00898005  1.00777753  1.0067127   1.00577495
  1.00495422  1.00424087  1.00362552  1.00309884  1.00265117  1.00227196
  1.00194847  1.0016623   1.00137516  1.04578823  1.04083473  1.03633373
  1.03224375  1.02852826  1.02515496  1.02209509  1.019323    1.01681568
  1.0145524   1.01251444  1.01068481  1.009048    1.00758987  1.00629742
  1.00515864  1.00416245  1.00329851  1.00255715  1.00192929  1.00140629
  1.00097993  1.00064221  1.00038535  1.0002015   1.00008261  1.00001985]
2nd ratio: [ 0.55532422  0.56798152  0.58098241  0.59433678  0.60805423  0.62214392
  0.63661436  0.65147323  0.66672707  0.68238099  0.69843827  0.71489984
  0.73176371  0.74902424  0.76667122  0.78468872  0.80305371  0.82173422
  0.84068702  0.8598545   0.87916049  0.89850429  0.91775183  0.93672129
  0.95515734  0.97267625  0.98860791  0.5552433   0.56788428  0.58086642
  0.5941994   0.60789256  0.62195482  0.63639446  0.65121891  0.66643449
  0.6820461   0.69805684  0.71446748  0.73127595  0.74847658  0.76605922
  0.78400814  0.80230065  0.82090535  0.83977983  0.8588677   0.87809452
  0.89736204  0.91653972  0.93545101  0.95384905  0.97136553  0.98736815]
1st ratio: [ 1.88885086  1.8380238   1.78912714  1.7420201   1.69657997  1.65269971
  1.61028593  1.56925734  1.5295434   1.49108327  1.45382503  1.41772501
  1.38274734  1.34886368  1.31605312  1.28430231  1.2536057   1.22396624
  1.19539637  1.16791953  1.1415726   1.11640962  1.09250795  1.06997884
  1.04898787  1.02980031  1.01291437  1.88347745  1.83282892  1.78411712
  1.73720093  1.69195731  1.64827881  1.60607162  1.56525399  1.52575488
  1.48751293  1.45047565  1.41459875  1.37984573  1.34618757  1.31360265
  1.2820768   1.25160368  1.22218537  1.19383333  1.16656999  1.1404311
  1.11546943  1.09176089  1.06941501  1.04859517  1.02956362  1.01281356]
H/W:   [ 1.02  1.02  1.02  1.02  1.02  1.02  1.02  1.02  1.02  1.02  1.02  1.02
  1.02  1.02  1.02  1.02  1.02  1.02  1.02  1.02  1.02  1.02  1.02  1.02
  1.02  1.02  1.02  2.    2.    2.    2.    2.    2.    2.    2.    2.    2.
  2.    2.    2.    2.    2.    2.    2.    2.    2.    2.    2.    2.    2.
  2.    2.    2.    2.  ]
2wE/W: [ 0.46  0.48  0.5   0.52  0.54  0.56  0.58  0.6   0.62  0.64  0.66  0.68
  0.7   0.72  0.74  0.76  0.78  0.8   0.82  0.84  0.86  0.88  0.9   0.92
  0.94  0.96  0.98  0.46  0.48  0.5   0.52  0.54  0.56  0.58  0.6   0.62
  0.64  0.66  0.68  0.7   0.72  0.74  0.76  0.78  0.8   0.82  0.84  0.86
  0.88  0.9   0.92  0.94  0.96  0.98]

H/W: tall, 2wE/W: small
------------------------------
0.95 <= first ratio <= 1.05: [ 1.          0.99999997  0.99999984  0.99999945  0.99999856  0.99999682
  0.99999377  0.99998879  0.99998112  0.99996982  0.99995371  0.99993141
  0.99990124  0.9998612   0.99980894  0.99974169  0.99965618  0.99954858
  0.99941442  0.99924843  0.9990445   0.99879542  0.99849279  0.99812674
  0.99768568  0.99715596  0.99652151  0.99576329  0.99485873  0.99378095
  0.99249779  0.99097061  0.98915275  0.98698747  0.98440536  0.98132079
  0.97762719  0.97319051  0.96784009  0.96135552  0.95344712  1.
  0.99999997  0.99999984  0.99999943  0.99999852  0.99999673  0.99999359
  0.99998847  0.99998059  0.99996897  0.99995243  0.99992953  0.99989855
  0.99985746  0.99980386  0.99973491  0.99964728  0.99953707  0.99939973
  0.99922992  0.99902141  0.99876692  0.99845791  0.9980844   0.99763467
  0.99709493  0.99644896  0.99567758  0.99475802  0.99366324  0.99236088
  0.99081212  0.98897008  0.98677784  0.98416578  0.98104806  0.97731795
  0.97284125  0.96744723  0.96091545  0.95295639]
krho'2: [  5.93150903e-08   9.50015949e-07   4.81769284e-06   1.52628329e-05
   3.73778983e-05   7.77998860e-05   1.44779035e-04   2.48264551e-04
   4.00008451e-04   6.13688859e-04   9.05054360e-04   1.29209132e-03
   1.79521639e-03   2.43749675e-03   3.24490122e-03   4.24658553e-03
   5.47521597e-03   6.96733582e-03   8.76378018e-03   1.09101450e-02
   1.34573176e-02   1.64620771e-02   1.99877732e-02   2.41050955e-02
   2.88929450e-02   3.44394228e-02   4.08429521e-02   4.82135546e-02
   5.66743026e-02   6.63629743e-02   7.74339430e-02   9.00603363e-02
   1.04436509e-01   1.20780878e-01   1.39339179e-01   1.60388219e-01
   1.84240192e-01   2.11247677e-01   2.41809420e-01   2.76377038e-01
   3.15462813e-01   6.08973189e-08   9.75319587e-07   4.94569082e-06
   1.56669182e-05   3.83630096e-05   7.98389784e-05   1.48548687e-04
   2.54679487e-04   4.10254733e-04   6.29255486e-04   9.27762981e-04
   1.32412372e-03   1.83913939e-03   2.49628411e-03   3.32195205e-03
   4.34573868e-03   5.60075980e-03   7.12401264e-03   8.95678452e-03
   1.11451149e-02   1.37403177e-02   1.67995724e-02   2.03865922e-02
   2.45723809e-02   2.94360901e-02   3.50659916e-02   4.15605805e-02
   4.90298300e-02   5.75966177e-02   6.73983523e-02   7.85888289e-02
   9.13403493e-02   1.05846149e-01   1.22323180e-01   1.41015308e-01
   1.62196985e-01   1.86177497e-01   2.13305850e-01   2.43976445e-01
   2.78635646e-01   3.17789420e-01]
H/W:   [ 1.02  1.02  1.02  1.02  1.02  1.02  1.02  1.02  1.02  1.02  1.02  1.02
  1.02  1.02  1.02  1.02  1.02  1.02  1.02  1.02  1.02  1.02  1.02  1.02
  1.02  1.02  1.02  1.02  1.02  1.02  1.02  1.02  1.02  1.02  1.02  1.02
  1.02  1.02  1.02  1.02  1.02  2.    2.    2.    2.    2.    2.    2.    2.
  2.    2.    2.    2.    2.    2.    2.    2.    2.    2.    2.    2.    2.
  2.    2.    2.    2.    2.    2.    2.    2.    2.    2.    2.    2.    2.
  2.    2.    2.    2.    2.    2.    2.  ]
2wE/W: [ 0.02  0.04  0.06  0.08  0.1   0.12  0.14  0.16  0.18  0.2   0.22  0.24
  0.26  0.28  0.3   0.32  0.34  0.36  0.38  0.4   0.42  0.44  0.46  0.48
  0.5   0.52  0.54  0.56  0.58  0.6   0.62  0.64  0.66  0.68  0.7   0.72
  0.74  0.76  0.78  0.8   0.82  0.02  0.04  0.06  0.08  0.1   0.12  0.14
  0.16  0.18  0.2   0.22  0.24  0.26  0.28  0.3   0.32  0.34  0.36  0.38
  0.4   0.42  0.44  0.46  0.48  0.5   0.52  0.54  0.56  0.58  0.6   0.62
  0.64  0.66  0.68  0.7   0.72  0.74  0.76  0.78  0.8   0.82]

0.95 <= total ratio <= 1.05: [ 1.0013438   1.00150625  1.00155493  1.00151063  1.00137569  1.00114728
  1.00082047  1.00038917  0.9998465   0.999185    0.99839659  0.99747267
  0.99640402  0.99518081  0.99379252  0.99222789  0.99047481  0.98852027
  0.98635022  0.9839495   0.98130164  0.97838875  0.97519131  0.97168799
  0.96785535  0.96366761  0.95909625  0.95410959  0.9999859   0.9999241
  0.99980595  0.9996228   0.99936671  0.99903     0.99860507  0.99808419
  0.99745948  0.99672277  0.99586553  0.99487882  0.99375322  0.99247875
  0.99104477  0.98943996  0.98765219  0.98566842  0.9834746   0.98105556
  0.97839487  0.97547464  0.97227538  0.96877578  0.96495243  0.96077957
  0.95622869  0.95126817]
2nd ratio: [ 1.0013438   1.00150628  1.00155509  1.00151118  1.00137713  1.00115046
  1.00082671  1.00040038  0.99986538  0.99921516  0.9984428   0.99754109
  0.99650243  0.99531896  0.99398243  0.99248425  0.99081547  0.9889667
  0.98692815  0.98468956  0.98224017  0.97956871  0.97666334  0.97351163
  0.97010048  0.96641614  0.9624441   0.95816907  0.99998591  0.99992413
  0.99980611  0.99962337  0.99936819  0.99903327  0.99861147  0.9980957
  0.99747884  0.99675369  0.9959129   0.99494894  0.99385405  0.99262023
  0.99123919  0.98970233  0.98800068  0.98612492  0.9840653   0.98181164
  0.97935325  0.97667896  0.97377703  0.97063513  0.96724028  0.96357883
  0.95963639  0.9553978 ]
1st ratio: [ 1.          0.99999997  0.99999984  0.99999945  0.99999856  0.99999682
  0.99999377  0.99998879  0.99998112  0.99996982  0.99995371  0.99993141
  0.99990124  0.9998612   0.99980894  0.99974169  0.99965618  0.99954858
  0.99941442  0.99924843  0.9990445   0.99879542  0.99849279  0.99812674
  0.99768568  0.99715596  0.99652151  0.99576329  1.          0.99999997
  0.99999984  0.99999943  0.99999852  0.99999673  0.99999359  0.99998847
  0.99998059  0.99996897  0.99995243  0.99992953  0.99989855  0.99985746
  0.99980386  0.99973491  0.99964728  0.99953707  0.99939973  0.99922992
  0.99902141  0.99876692  0.99845791  0.9980844   0.99763467  0.99709493
  0.99644896  0.99567758]
H/W:   [ 1.02  1.02  1.02  1.02  1.02  1.02  1.02  1.02  1.02  1.02  1.02  1.02
  1.02  1.02  1.02  1.02  1.02  1.02  1.02  1.02  1.02  1.02  1.02  1.02
  1.02  1.02  1.02  1.02  2.    2.    2.    2.    2.    2.    2.    2.    2.
  2.    2.    2.    2.    2.    2.    2.    2.    2.    2.    2.    2.    2.
  2.    2.    2.    2.    2.    2.  ]
2wE/W: [ 0.02  0.04  0.06  0.08  0.1   0.12  0.14  0.16  0.18  0.2   0.22  0.24
  0.26  0.28  0.3   0.32  0.34  0.36  0.38  0.4   0.42  0.44  0.46  0.48
  0.5   0.52  0.54  0.56  0.02  0.04  0.06  0.08  0.1   0.12  0.14  0.16
  0.18  0.2   0.22  0.24  0.26  0.28  0.3   0.32  0.34  0.36  0.38  0.4
  0.42  0.44  0.46  0.48  0.5   0.52  0.54  0.56]

H/W: shallow, 2wE/W: large
------------------------------
0.95 <= first ratio <= 1.05: [ 1.04904071  1.02983215  1.01292793]
krho2: [ 0.31638993  0.22382135  0.11893569]
H/W:   [ 1.  1.  1.]
2wE/W: [ 0.94  0.96  0.98]

0.95 <= total ratio <= 1.05: [ 1.04266411  1.02856579  1.01599183  1.00475681  0.99470779  0.98571716
  0.97767736  0.97049694  0.96409749  0.95841133  0.95337958  0.95105378
  0.95565134  0.9619061 ]
2nd ratio: [ 0.42181475  0.43000682  0.43849611  0.44729365  0.45641079  0.46585922
  0.47565095  0.48579829  0.49631383  0.50721041  0.5185011   0.90659377
  0.92796805  0.94962936]
1st ratio: [ 2.47185312  2.39197553  2.31699165  2.24630243  2.17941339  2.11591207
  2.05545128  1.99773643  1.94251589  1.88957344  1.8387224   1.04904071
  1.02983215  1.01292793]
H/W:   [ 1.  1.  1.  1.  1.  1.  1.  1.  1.  1.  1.  1.  1.  1.]
2wE/W: [ 0.28  0.3   0.32  0.34  0.36  0.38  0.4   0.42  0.44  0.46  0.48  0.94
  0.96  0.98]

H/W: shallow, 2wE/W: small
------------------------------
0.95 <= first ratio <= 1.05: [ 1.00086465  0.99987471  0.99999908  0.9999999   1.00000001  1.          1.
  1.          1.          1.          1.          1.          1.          1.
  1.          1.          1.          1.          1.          1.          1.
  1.          1.          1.          1.          1.          1.          1.
  1.          1.          1.          1.          1.          1.          1.
  1.          1.          1.          1.          1.          1.          1.
  1.          1.          0.99999989  0.99999669  0.99989681  0.99628076
  1.          0.99999997  0.99999984  0.99999945  0.99999856  0.99999683
  0.99999379  0.99998883  0.99998119  0.99996993  0.99995388  0.99993166
  0.99990159  0.9998617   0.99980962  0.99974259  0.99965736  0.99955011
  0.99941636  0.99925089  0.99904756  0.9987992   0.99849742  0.99813236
  0.99769245  0.99716406  0.99653114  0.99577468  0.99487212  0.9937966
  0.992516    0.9909917   0.98917706  0.98701538  0.98443727  0.98135713
  0.97766841  0.97323708  0.96789251  0.96141426  0.95351265]
krho'2: [  2.71140512e-66   7.22163073e-65   1.71965216e-63   4.00303667e-62
   9.27470780e-61   2.14678050e-59   4.96806284e-58   1.14965685e-56
   2.66039170e-55   6.15633361e-54   1.42461838e-52   3.29666567e-51
   7.62871273e-50   1.76533697e-48   4.08511201e-47   9.45323215e-46
   2.18754340e-44   5.06212693e-43   1.17141123e-41   2.71072673e-40
   6.27280941e-39   1.45157155e-37   3.35903710e-36   7.77304450e-35
   1.79873634e-33   4.16240047e-32   9.63208298e-31   2.22893072e-29
   5.15790006e-28   1.19357380e-26   2.76201244e-25   6.39148810e-24
   1.47903462e-22   3.42258854e-21   7.92010695e-20   1.83276761e-18
   4.24115118e-17   9.81431759e-16   2.27110107e-14   5.25548518e-13
   1.21615567e-11   2.81426846e-10   6.51241214e-09   1.50701728e-07
   3.48734236e-06   8.06995176e-05   1.86744273e-03   4.32139183e-02
   5.91055407e-08   9.46664673e-07   4.80073984e-06   1.52093102e-05
   3.72474076e-05   7.75297599e-05   1.44279607e-04   2.47414567e-04
   3.98650641e-04   6.11625725e-04   9.02044199e-04   1.28784453e-03
   1.78939212e-03   2.42969993e-03   3.23468005e-03   4.23342948e-03
   5.45855441e-03   6.94653734e-03   8.73815269e-03   1.08789368e-02
   1.34197196e-02   1.64172259e-02   1.99347562e-02   2.40429572e-02
   2.88206955e-02   3.43560481e-02   4.07474272e-02   4.81048598e-02
   5.65514426e-02   6.62250015e-02   7.72799856e-02   8.98896312e-02
   1.04248441e-01   1.20575027e-01   1.39115374e-01   1.60146601e-01
   1.83981292e-01   2.10972506e-01   2.41519569e-01   2.76074802e-01
   3.15151338e-01]
H/W:   [ 0.02  0.02  0.02  0.02  0.02  0.02  0.02  0.02  0.02  0.02  0.02  0.02
  0.02  0.02  0.02  0.02  0.02  0.02  0.02  0.02  0.02  0.02  0.02  0.02
  0.02  0.02  0.02  0.02  0.02  0.02  0.02  0.02  0.02  0.02  0.02  0.02
  0.02  0.02  0.02  0.02  0.02  0.02  0.02  0.02  0.02  0.02  0.02  0.02
  1.    1.    1.    1.    1.    1.    1.    1.    1.    1.    1.    1.    1.
  1.    1.    1.    1.    1.    1.    1.    1.    1.    1.    1.    1.    1.
  1.    1.    1.    1.    1.    1.    1.    1.    1.    1.    1.    1.    1.
  1.    1.  ]
2wE/W: [ 0.04  0.06  0.08  0.1   0.12  0.14  0.16  0.18  0.2   0.22  0.24  0.26
  0.28  0.3   0.32  0.34  0.36  0.38  0.4   0.42  0.44  0.46  0.48  0.5
  0.52  0.54  0.56  0.58  0.6   0.62  0.64  0.66  0.68  0.7   0.72  0.74
  0.76  0.78  0.8   0.82  0.84  0.86  0.88  0.9   0.92  0.94  0.96  0.98
  0.02  0.04  0.06  0.08  0.1   0.12  0.14  0.16  0.18  0.2   0.22  0.24
  0.26  0.28  0.3   0.32  0.34  0.36  0.38  0.4   0.42  0.44  0.46  0.48
  0.5   0.52  0.54  0.56  0.58  0.6   0.62  0.64  0.66  0.68  0.7   0.72
  0.74  0.76  0.78  0.8   0.82]

0.95 <= total ratio <= 1.05: [ 1.00086465  0.99987471  0.99999908  0.9999999   1.00000001  1.          1.
  1.          1.          1.          1.          1.          1.          1.
  1.          1.          1.          1.          1.          1.          1.
  1.          1.          1.          1.          1.          1.          1.
  1.          1.          1.          1.          1.          1.          1.
  1.          1.          1.          1.          1.          1.          1.
  1.          1.          0.99999989  0.99999669  0.99989681  0.99628076
  0.9828279   0.97999218  0.97781992  0.97593282  0.97419415  0.97253545
  0.97091575  0.96930742  0.96769003  0.96604731  0.96436554  0.9626325
  0.96083687  0.95896775  0.95701439  0.95496591  0.95281109  0.9505382 ]
2nd ratio: [ 1.          1.          1.          1.          1.          1.          1.
  1.          1.          1.          1.          1.          1.          1.
  1.          1.          1.          1.          1.          1.          1.
  1.          1.          1.          1.          1.          1.          1.
  1.          1.          1.          1.          1.          1.          1.
  1.          1.          1.          1.          1.          1.          1.
  1.          1.          1.          1.          1.          1.          0.9828279
  0.97999221  0.97782008  0.97593336  0.97419555  0.97253853  0.97092178
  0.96931825  0.96770823  0.96607636  0.96441002  0.96269829  0.96093143
  0.95910039  0.95719662  0.95521179  0.95313767  0.95096603]
1st ratio: [ 1.00086465  0.99987471  0.99999908  0.9999999   1.00000001  1.          1.
  1.          1.          1.          1.          1.          1.          1.
  1.          1.          1.          1.          1.          1.          1.
  1.          1.          1.          1.          1.          1.          1.
  1.          1.          1.          1.          1.          1.          1.
  1.          1.          1.          1.          1.          1.          1.
  1.          1.          0.99999989  0.99999669  0.99989681  0.99628076
  1.          0.99999997  0.99999984  0.99999945  0.99999856  0.99999683
  0.99999379  0.99998883  0.99998119  0.99996993  0.99995388  0.99993166
  0.99990159  0.9998617   0.99980962  0.99974259  0.99965736  0.99955011]
H/W:   [ 0.02  0.02  0.02  0.02  0.02  0.02  0.02  0.02  0.02  0.02  0.02  0.02
  0.02  0.02  0.02  0.02  0.02  0.02  0.02  0.02  0.02  0.02  0.02  0.02
  0.02  0.02  0.02  0.02  0.02  0.02  0.02  0.02  0.02  0.02  0.02  0.02
  0.02  0.02  0.02  0.02  0.02  0.02  0.02  0.02  0.02  0.02  0.02  0.02
  1.    1.    1.    1.    1.    1.    1.    1.    1.    1.    1.    1.    1.
  1.    1.    1.    1.    1.  ]
2wE/W: [ 0.04  0.06  0.08  0.1   0.12  0.14  0.16  0.18  0.2   0.22  0.24  0.26
  0.28  0.3   0.32  0.34  0.36  0.38  0.4   0.42  0.44  0.46  0.48  0.5
  0.52  0.54  0.56  0.58  0.6   0.62  0.64  0.66  0.68  0.7   0.72  0.74
  0.76  0.78  0.8   0.82  0.84  0.86  0.88  0.9   0.92  0.94  0.96  0.98
  0.02  0.04  0.06  0.08  0.1   0.12  0.14  0.16  0.18  0.2   0.22  0.24
  0.26  0.28  0.3   0.32  0.34  0.36]

\end{lstlisting}